\documentclass[titlepage,12pt]{article}

\usepackage[dvips]{graphicx}

\usepackage[myheadings]{fullpage}
\usepackage{pmetrika}

\usepackage{submit}
\usepackage{apacite}

\RequirePackage{amsmath,amsfonts,amssymb}
\RequirePackage[authoryear]{natbib}
\RequirePackage[colorlinks,citecolor=blue,urlcolor=blue]{hyperref}
\RequirePackage{graphicx}

\usepackage{amsmath,natbib}
\usepackage{multirow}

\usepackage{graphicx,bm,mathtools,amsfonts,bbm,hyperref,enumitem}
\numberwithin{equation}{section}

\usepackage{algorithmicx}
\usepackage{algorithm}
\usepackage[noend]{algpseudocode}
\algnewcommand\algorithmicinput{\textbf{Input:}}
\algnewcommand\Input{\item[\algorithmicinput]}
\algnewcommand\algorithmicoutput{\textbf{Output:}}
\algnewcommand\Output{\item[\algorithmicoutput]}

\newcommand{\bX}{\textbf{X}}

\newcommand{\bx}{\textbf{x}}

\newcommand{\bR}{\textbf{R}}
\newcommand{\btheta}{\bm{\theta}}
\newcommand{\bbeta}{\bm{\beta}}

\newcommand{\E}{\mathbb{E}}

\usepackage{tikz}
\usetikzlibrary{shadows,arrows,positioning,matrix, trees}
\pgfdeclarelayer{background}
\pgfdeclarelayer{foreground}
\pgfsetlayers{background,main,foreground}

\tikzstyle{materia}=[draw, fill=blue!20, text width=6.0em, text centered,
minimum height=1.5em,drop shadow]
\tikzstyle{etape} = [materia, text width=8em, minimum width=10em,
minimum height=3em, rounded corners, drop shadow]
\tikzstyle{texto} = [above, text width=6em, text centered]
\tikzstyle{linepart} = [draw, thick, color=black!50, -latex', dashed]
\tikzstyle{line} = [draw, thick, color=black!50, -latex']
\tikzstyle{ur}=[draw, text centered, minimum height=0.01em]

\tikzset{
	treenode/.style = {shape=rectangle, rounded corners,
		draw, align=center, font=\scriptsize},
	root/.style     = {treenode},
	env/.style      = {treenode},
	leaf/.style     = {shape=circle,draw,align=center,font=\scriptsize},
	every node/.style       = {font=\tiny},
	dummy/.style    = {coordinate},
	plain/.style     = {align=center,font=\scriptsize},
}

\tikzset{plain/.style={font=\large}} 

\usepackage{soul}
\setstcolor{magenta}

\renewcommand{\hat}[1]{\widehat{#1}}

\renewcommand{\tilde}[1]{\widetilde{#1}}

\begin{document}

\begin{titlepage}

\title{}


\author{Daniel Suen}

\affil{University of Washington}

\author{Yen-Chi Chen}
\affil{University of Washington}

\vspace{\fill}\centerline{\today}\vspace{\fill}

\comment{Competing interests: The authors declare none.}

\comment{Financial support: DS was supported by NSF DGE-2140004.
	YC was supported by NSF DMS-195278, NSF DMS-2112907, NSF DMS-2141808, and NIH U24-AG072122.}
\linespacing{1}
\contact{Correspondence should be sent to\\

\noindent E-Mail: dsuen@uw.edu \\
\noindent Phone: 1-206-486-0446
}

\end{titlepage}

\setcounter{page}{2}
\vspace*{2\baselineskip}

\RepeatTitle{Modeling Multivariate Missingness with Tree Graphs and Conjugate Odds}\vskip3pt

\linespacing{1.5}
\abstracthead
\begin{abstract}
In this paper, we analyze a specific class of missing not at random (MNAR)  assumptions called \textit{tree graphs}, extending upon the work of pattern graphs.  We build off previous work by introducing the idea of a conjugate odds family in which certain parametric models on the selection odds can preserve the data distribution family across all missing data patterns.  Under a conjugate odds family and a tree graph assumption, we are able to model the full data distribution elegantly in the sense that
for the observed data, we obtain a model that is conjugate from the complete-data, and for the missing entries, we create a simple imputation model.
In addition, we  investigate the problem of graph selection, sensitivity analysis, and statistical inference. 	
Using both simulations and real data, we illustrate the applicability of our method.

\end{abstract}\vspace{\fill}\pagebreak

	\section{Introduction}


Missing data are pervasive across healthcare, social sciences, economics, and machine learning. They arise from survey nonresponse, equipment failure, privacy concerns, and other sources, and the manner in which data are missing strongly influences the validity of statistical analyses. When ignored, missingness can bias results and reduce statistical power, especially in large-scale studies where incomplete records are common \citep{LittleRubin02}.

Rubin's framework classifies missingness into three categories \citep{little1989analysis}: missing completely at random (MCAR), missing at random (MAR), and missing not at random (MNAR). Standard approaches are effective under MCAR or MAR, but MNAR poses a fundamentally harder problem: the probability of missingness depends on unobserved values, rendering the distribution unidentifiable without further assumptions. The challenge is particularly acute in multivariate and nonmonotone settings, where missingness occurs irregularly across variables.

Most practical methods rely on imputation, such as multiple imputation by chained equations (\textit{mice}; \citealt{mice}) or MissForest \citep{missforest2011}, which are valued for their flexibility but implicitly assume MAR or rely on potentially incompatible conditionals. Moreover, methods such as MissForest are also single imputation methods, which can lead to inconsistent estimators, depending on the parameter of interest. These limitations make them vulnerable to bias or incoherence under MNAR. Direct modeling of imputation distributions is also difficult because of high dimensionality and interdependence among variables, motivating the search for methods that are both interpretable and theoretically principled.

Two classical approaches to MNAR are selection models \citep{diggle1994informative} and pattern-mixture models \citep{little1993pattern}, which respectively specify missingness probabilities or stratify by missingness patterns. While widely used, both require untestable assumptions for identifiability. More recent strategies include ``no self-censoring" assumptions \citep{shpitser2016consistent, sadinle2017itemwise}, auxiliary variables \citep{miao2016doubly}, and CCMV-type restrictions \citep{tchetgen2018discrete}. Graphical frameworks, such as missing data DAGs \citep{mohan2013graphical} and pattern graphs \citep{chen2022}, provide powerful representations of missingness assumptions, though their generality can make model selection challenging.

This paper builds on these advances by focusing on a structured and tractable subclass of pattern graphs, which we term tree graphs. Tree graphs simplify model specification, connect naturally to existing MNAR assumptions, and form the basis for scalable imputation strategies. To complement this structure, we introduce the conjugate odds property, which provides a flexible parametric tool for modeling conditional distributions. Together, tree graphs and conjugate odds yield a unified framework that ensures nonparametric identification, facilitates inference, and enables practical sensitivity analysis.


{\bf Outline.}	
We study tree graphs, a special case of pattern graph wtih
nice properties in Section \ref{sect:treegraphs} and derive related theories.
In Section \ref{sect:conjugateodds}, we introduce the idea of conjugate odds
that is useful in domain adaptation.
We study how the conjugate odds can be used in handling missing data  with 
tree graphs in Section \ref{sect:treegraphandconjugate}, which leads to 
an imputation model and a model on the observed data simultaneously. 
We introduce three approaches for selecting a tree graph in Section \ref{sect:choosing-treegraph}: prior knowledge, 
partial-ordering, and data-driven approaches. 
In Section \ref{sect:realdata}, we apply the tree graph and conjugate odds to 
an Alzheimer's disease data. 
In appendices, we also investigate tree graph performance via simulation studies (Appendix \ref{sect:simulations}), and study the problem of statistical inferences (Appendix \ref{sec::inference}) and sensitivity analysis (Appendix \ref{section:sensitivity}).

\subsection{Notation}

We use a capital boldface variable to denote a vector-valued random variable.  In this paper, we consider a general problem setup, where $\bX = (X_1, X_2, \ldots, X_d)\in\mathbb{R}^d$ is a random vector of variables.  Each of the $d$ variables can possibly be missing for a total of up to $2^d$ missing patterns.  Let $\bR \in \mathcal{R} \subseteq \{0,1\}^d$ be the random binary vector that describes the missing pattern associated with $\bX$.

We write $R_j=0$ if and only if variable $X_j$ is missing.  For a fixed pattern $r$, let $\bX_r = (X_j : r_j=1)$ denote the observed random variables and $X_{\bar{r}} = (X_j : r_j=0)$ denote the missing random variables.  When we write ``For $j$ in $r$,'' this refers to the indices that contain 1.  For example, suppose $\bX=(X_1, X_2, X_3, X_4)$ and $r=1001$.  We have $\bX_r = (X_1, X_4)$ and $\bX_{\bar{r}} = (X_2, X_3)$.  Then, the statement ``For $j$ in $r$,'' corresponds to ``For $j$ in 1, 4.''  We assume that the complete data is generated by sampling i.i.d. from the joint distribution $p(x,r)$, and the resulting associated pattern $r$ generates the observed data.  In this paper, we will use the terminology \textit{full-data distribution} and \textit{pattern-specific joint distribution} to refer to $p(x,r)$ and $p(x|r)$, respectively.

\section{Tree graphs and identification theory}

\label{sect:treegraphs}


\subsection{Pattern graphs}

\cite{chen2022} originally proposed pattern graphs as a way to model nonmonotone missingness and nonparametrically identify the full-data distribution $p(x,r)$.  A pattern graph is a directed graph of missing data patterns that encodes a missing data assumption capable of nonparametrically identifying the full data distribution.  In the paper, he proposed selection odds and pattern-mixture model formulations with respect to a given graph and showed that the two are equivalent.  Estimation procedures using inverse probability weighting, regression adjustment, and semiparametric efficiency theory were explored.  Building on this work, we focus on a strict subset of pattern graphs called \textit{tree graphs}.

In this subsection, we first broadly summarize the previous work by introducing the notion of a pattern graph.  We impose a partial order on the patterns in $\mathcal{R}$, where for any distinct $s,r\in\mathcal{R}$, we say $s>r$ if and only if the observed variables in pattern $r$ are also observed in pattern $s$.  From this partial order, we can construct a directed graph of all patterns, which forms the aforementioned \textit{pattern graph}.

\begin{definition}
	A regular pattern graph is a directed graph of all patterns in $\mathcal{R}$ such that
	\begin{enumerate}
		\item \textbf{Single source node.} Pattern $1_d := (\underbrace{1,1,1,1,\ldots,1}_{d \ \text{times}})$ is the only source.
		\item \textbf{Regularity.} If there is an arrow present in the graph $T$ from pattern $s$ to pattern $r$, then $s>r$.
	\end{enumerate}
\end{definition}

The second property refers to the regularity and ensures that the graph is directed in a way that preserves the partial ordering of the patterns.  Pattern graphs represent information flow, which translates to a specific missing data assumption.  Since we have a partial ordering of the missing patterns, we assume that a pattern borrows information from its parents to model its missing data.  Denote $\text{PA}_{T}(r)$ as the set of parents for pattern $r$ in graph $T$.  Formally, the pattern-mixture model of $(\bX, \bR)$ factorizes with respect to $T$ if, we have
\begin{equation}
	p(x_{\bar{r}} | x_r, \bR=r) \stackrel{T}{=} p(x_{\bar{r}} | x_r, \bR\in\text{PA}_{T}(r)) \qquad \forall r\in\mathcal{R}.
	\tag{P1}
	\label{eq:pmm-patterngraph}
\end{equation}
Equation \eqref{eq:pmm-patterngraph} represents the pattern mixture model factorization property, which states that the extrapolation distribution under pattern $r$ can be identified using information from its parents.  Additionally, the selection odds model of $(\bX, \bR)$ factorizes with respect to the graph if
\begin{equation}
	\frac{P(\bR=r|x)}{P(\bR\in\text{PA}_{T}(r)|x)} \stackrel{T}{=} \frac{P(\bR=r|x_r)}{P(\bR\in\text{PA}_{T}(r)|x_r)} \qquad \forall r\in\mathcal{R}.
	\tag{P2}
	\label{eq:selection-patterngraph}
\end{equation}
Equation \eqref{eq:selection-patterngraph} represents the selection odds factorization property, which states that the conditional odds of a pattern $r$ against its parents only depends on the commonly observed variables.  
\cite{chen2022} shows that equations \eqref{eq:pmm-patterngraph} and \eqref{eq:selection-patterngraph} are equivalent under very mild positivity condition,
so we can interchangeably using these two definitions.
The pattern-mixture model formulation illuminates the kind of assumption imposed by pattern graph $T$.  In particular, $T$ associates each pattern $r$ with a set $\text{PA}_{T}(r)$ comprising closely related patterns whose observed variables also include those of $r$.

Further work by \cite{Zamanian2023patterngraph} studied the sensitivity analysis within the pattern graph framework.  Patterns graphs have also recently been used by \cite{dong2025efficientestimationmultiplemissing} in the context of estimating equations.  We note that pattern graphs are not the conventional graphical model because the nodes here are represented by missing data patterns rather than individual variables.  In previous missing data literature that use missing data DAGs or m-DAGs, one augments the usual directed acyclic graph of variables with nodes for missingness indicators and edges capturing dependencies \citep{mohan2013graphical, tian2015missing, Mohan2021graphical}. \cite{phung2025recursiveequationsimputationmissing} recently used a pattern DAG along with an m-DAG to help with identification of the full-data distribution.






\subsection{Definition and algebraic properties}

While pattern graphs encompass an exceptionally broad class of missing data assumptions, there are some obvious shortcomings due to the flexibility of the graph.
A complex graph, while it is mathematically valid, is not practically useful due to the fact that it would require one model on one edge within the graph. 
To resolve this complexity issue while maintaining the validity of a graph, 
we now focus on a particular subclass of pattern graphs known as tree graphs. Tree graphs represent a rich subset of pattern graphs, exhibiting notable algebraic and statistical properties that facilitate simpler estimation procedures and graph construction. The term \textit{tree graph} is used to reflect its graphical structure, which resembles a tree with a single root node.  Moreover, several established missing data assumptions from the literature can be incorporated into the pattern graph framework and reformulated as tree graphs.

\begin{definition}[Tree graph]
\label{def:treegraph}
A tree graph is a regular pattern graph in which every pattern $r\neq 1_d$ has exactly one directed path originating from $1_d$.
\end{definition}

\begin{proposition}
\label{prop:mnar}
Every tree graph corresponds to a unique missing not at random (MNAR) assumption that nonparametrically identifies the full data distribution.  Additionally, the selection odds $P(\bR=r_\ell|X)/P(\bR=1_d|X)$ admits the following identification formula
\begin{align*}
	\frac{P(\bR=r_\ell|X)}{P(\bR=1_d|X)}
	&\stackrel{\text{tree graph}}{=} \prod_{i=1}^\ell \frac{P(\bR=r_i|X_{r_i})}{P(\bR=r_{i-1}|X_{r_{i}})},
\end{align*}
where $1_d =: r_0 \to r_1 \to r_2 \to \cdots \to r_\ell$ is the unique path in the tree graph from the source $1_d$ to pattern $r_{\ell}$.
\end{proposition}

In this paper, we denote the set of tree graphs formed from $d$ variables as $\mathcal{T}_d$.  For brevity, we omit the subscript $d$ and simply write $\mathcal{T}$ when the context makes it clear that we are considering $d$ variables.  From the definition, we can see that a tree graph is a directed graph of the patterns in which there is a unique single path from the source $1_d$ to a given pattern $r$.  In classical graph theory, this structure is also known as an \textit{arborescence} \citep{fournier2013graphs}.  Proposition \ref{prop:mnar} highlights a key identification result for tree graphs.  That is, a tree graph is a MNAR assumption that automatically nonparametrically identifies the full data distribution.  MNAR assumptions are notably difficult to formulate.
\begin{proposition}
	\label{prop:mcar}
	If the data is missing completely at random (MCAR), then a tree graph assumption will still recover the true data distribution $p(x)$.
\end{proposition}

Proposition \ref{prop:mcar} further emphasizes the fact that a tree graph assumption can still be applied when the data could be MCAR.  To make the graphical formulation more concrete, we include two specific examples of common missing data assumptions that can be reframed as tree graphs.  In Example \ref{example:ccmv}, we discuss the complete-case missing value (CCMV) assumption \citep{little1993pattern, tchetgen2018discrete, Tan02102023}.


\begin{example}[Complete-case missing value (CCMV)]
\label{example:ccmv}
Our first example is the complete-case missing value, which is equivalent to
$$p(x_{\bar{r}} | x_r, \bR=r) \stackrel{\text{CCMV}}{=} p(x_{\bar{r}}| x_r, \bR=1_d)$$
for all $r\in\mathcal{R}$ \citep{little1993pattern, tchetgen2018discrete}.  This can be viewed as a relaxation of a complete case analysis to an assumption that does not place constraints on the observed data.  In particular, the complete case distribution is only used to define the extrapolation distributions.

In contrast, a complete case analysis makes the assumption that
$$p(x|\bR=r) \stackrel{\text{CCA}}{=} p(x|\bR=1_d)$$
for any $r\in\mathcal{R}$, which is essentially the missing completely at random.  The right-hand side of the equation is the complete-case distribution while the left-hand side is the distribution of the data under a given pattern $\bR=r$.  Since the LHS decomposes as $p(x|\bR=r) = p(x_{\bar{r}}|x_r,\bR=r)p(x_r,\bR=r)$, a complete-case analysis implictly places an assumption on the observed data and thus, may not agree with the observed data.  The CCMV assumption bypasses this by only placing assumptions on the distribution of the missing variables, conditional on the observed data, and can be viewed as a first step above a naive CCA.  For a visual example, we visualize the CCMV assumption in Figure \ref{fig:ccmv-ncmv-example} for $d=3$ variables.
\end{example}

From a graphical perspective, the CCMV assumption represents the most natural tree graph, as it forms the shallowest structure.  More broadly, tree graphs can be viewed as generalizations of the CCMV assumption, allowing for more complex paths from $1_d$ to the remaining patterns.  In Example \ref{example:ncmv}, we discuss another tree graph assumption, nearest-case missing value (NCMV) in the context of monotone missingness.

\begin{example}[Nearest-case missing value under monotone missingness]
\label{example:ncmv}
In our second example, we consider a setting of monotone missingness in which the missing patterns form an ordered set.  For simplicity, we assume that the missingness arises from dropout such that if variable $X_j$ is missing, then variable $X_{j'}$ is also missing for any $j'>j$.  For notational convenience, we denote each pattern by a positive integer that denotes the index of the first $0$ in the missing pattern such that $\mathcal{D} = \{0,1,2,3,4,5,6,\ldots,d\}$.  Then, the set $\mathcal{D}$ has a one-to-one correspondence with $\mathcal{R} = \{1_d, 1_{d-1}0,1_{d-2}00,\ldots,0_d\}$, where the subscript denotes the number of $1$s and $0$s.  Letting $D$ denote the random variable associated with $\mathcal{D}$, the NCMV assumption is equivalent to
$$p(x_{>t} | x_{\leq t}, D=t) = p(x_{>t}| x_{\leq t}, D=t+1),$$
$$\frac{P(D=t|x)}{P(D=t+1|x)} = \frac{P(D=t|x_t)}{P(D=t+1|x_t)}$$
for all $t\in\mathcal{D}$.  We visualize the NCMV assumption in Figure \ref{fig:ccmv-ncmv-example} for $d=3$ variables.
\end{example}

\begin{figure}
\centering
\includegraphics[width=0.85\textwidth]{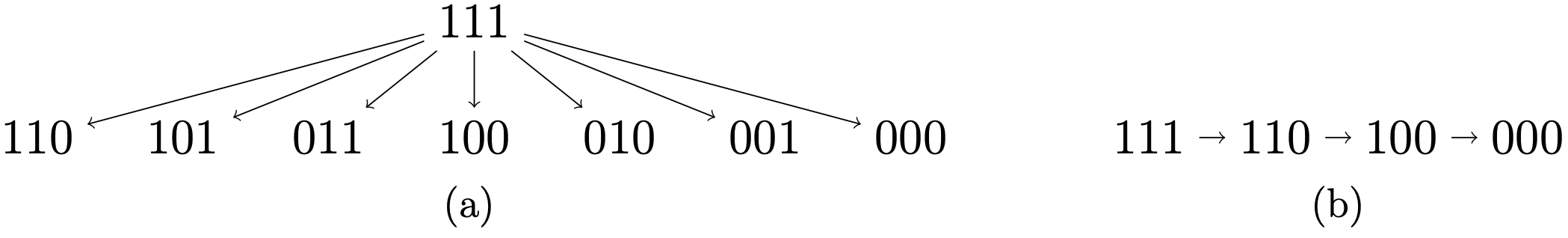}
\caption{(a) The CCMV tree graph for \( d=3 \) variables. (b) The NCMV tree graph for \( d=3 \) variables.}
\label{fig:ccmv-ncmv-example}
\end{figure}

From Examples \ref{example:ccmv} and \ref{example:ncmv}, we see that previously proposed assumptions from the literature can be cast in the tree graph framework.  Through Proposition \ref{prop:equivprop}, we now introduce equivalent definitions for tree graphs.

\begin{proposition}[Equivalence definitions for tree graphs]
\label{prop:equivprop}
Let $T$ be a pattern graph.  The following statements are equivalent:
\begin{enumerate}
	\item \textbf{Unique directed path from $1_d$.}  The pattern graph $T$ is a tree graph with $d$ variables.
	\item \textbf{Single parent.}  Every pattern $r\neq 1_d$ in $T$ has exactly one parent. \label{fact:reduct}
	\item \textbf{Minimal.}  There are $2^d-1$ edges in $T$. That is, $T$ achieves the lower bound on the number of edges that a pattern graph must have.
\end{enumerate}
\end{proposition}

Several key practical insights arise from these properties.  First, since these formulations are equivalent, the proposition provides multiple possible equivalent definitions of a tree graph.  Moreover, as each pattern has exactly one parent, this gives us a straightforward method to both enumerate the class of tree graphs and construct a specific tree graph.  The construction process is further discussed in Section \ref{sect:choosing-treegraph}.  Next, minimality is closely linked to model complexity.  Missing not at random assumptions can be notably exponentially complex.  As the existence of an edge requires fitting an additional selection odds model, minimality ensures that the model complexity for the global model $p(x,r)$ is minimized within the space of pattern graphs and selection models.



\subsection{Enumeration}

\label{subsection:enumeration}

The size of the pattern graph set is astronomical as a function of the number of variables $d$ \citep{chen2022}, illustrating that pattern graphs represents a huge class
of MNAR assumptions. 
While tree graph is just a subset of pattern graphs,
the number of tree graphs stills grows significantly with the dimension $d$, so it also includes many MNAR assumptions.  
This is formalized with a lower bound, which is presented in the following proposition.

\begin{proposition}[Enumeration of tree graphs]
\label{prop:enumerate}
The number of tree graphs is super-exponential in the number of variables $d$.  In particular, $\log |\mathcal{T}_d| \gtrsim d\cdot 2^d$.
\end{proposition}

The size of the tree graph class grows rapidly.  The exact form of our lower bound is provided in the proof, but we note that when $d=5$, we have $|\mathcal{T}_d| \geq 2^{18}$, and when $d=6$, we have $|\mathcal{T}_d| \geq 2^{66}$.  Observe that the size of the class is largely due to the fact that the number of missing patterns is $2^d$, exponential in the number of variables.  However, in practice, many of these patterns may not be observed in a given real data set.  For example, when the missingness is monotone, the number of missing data patterns is $d$.  Through some careful algebra, one can show that this reduces to the $2^{O(d^2)}$, which is still substantial.  Thus, the tree graph set remains a rich class of missing not at random assumptions while having significant simplifications in the resulting model complexity.
In Section \ref{sect:choosing-treegraph}, we discuss strategies for selecting a reasonable tree graph.

\section{Conjugate odds families and domain adaptation}

\label{sect:conjugateodds}

In the previous section, we established that tree graphs provide a graphical representation for an MNAR assumption that identifies the full data distribution. 
An additional benefit of the tree graph is that it allows a simple modeling framework to estimate the pattern-specific data distribution via the graph structure
by transferring the complete data distribution into each observed data distribution along the branch within a tree.
This is inspired by Tukey's factorization \citep{Tukey1986}, which we will discuss in more detail in Section \ref{sect:treegraphandconjugate}. First, we discuss the idea of learning a target distribution from a source distribution.

A common problem in statistics and machine learning is learning a target distribution $p(x \mid A=a')$ given knowledge of a related distribution $p(x \mid A=a)$. This setting is studied under domain adaptation and transfer learning, where knowledge from a source distribution is adapted to a target one under distributional shift. From a generative modeling perspective, this is closely related to density ratio estimation \citep{sugiyama_suzuki_kanamori_2012}. In the missing data setting, we view the distribution of complete cases ($\bR=1_d$) as the source domain and the distribution under another missingness pattern $\bR=r$ as the target. Since rare patterns often have few observations, direct estimation of the target distribution can be infeasible, making domain adaptation particularly well-suited.

\subsection{Exponential tilting}

A natural starting point is exponential tilting (or exponential change of measure) \citep{Esscher1932}. Given a baseline density $p_0(x)$, the tilted distribution with parameter $\lambda$ takes the form
$$
p_1(x) \propto p_0(x)e^{\lambda T(x)} \quad \iff \quad p_1(x) = \frac{p_0(x)e^{\lambda T(x)}}{\E_{p_0}[ e^{\lambda T(X)}]},
$$
where $T(x)$ is a statistic that is often
chosen to be the sufficient statistic in exponential family models, and the denominator ensures normalization. 


A key property of exponential tilting is that it preserves the exponential family structure. If the base distribution belongs to an exponential family with natural parameter $\eta_0$, then the tilted distribution corresponds to a simple shift in the natural parameter, $\eta_\lambda = \eta_0 + \lambda$. This property enables efficient statistical computations, as it allows reweighting while maintaining sufficient statistics and conjugate relationships. 
In our framework, exponential tilting provides a principled way to adapt the complete-case distribution to approximate distributions under other missingness patterns, linking ideas from domain adaptation with tractable exponential family models.

\begin{figure}[!b]
\label{fig:exptilt}
\centering
\includegraphics[width=0.95\textwidth]{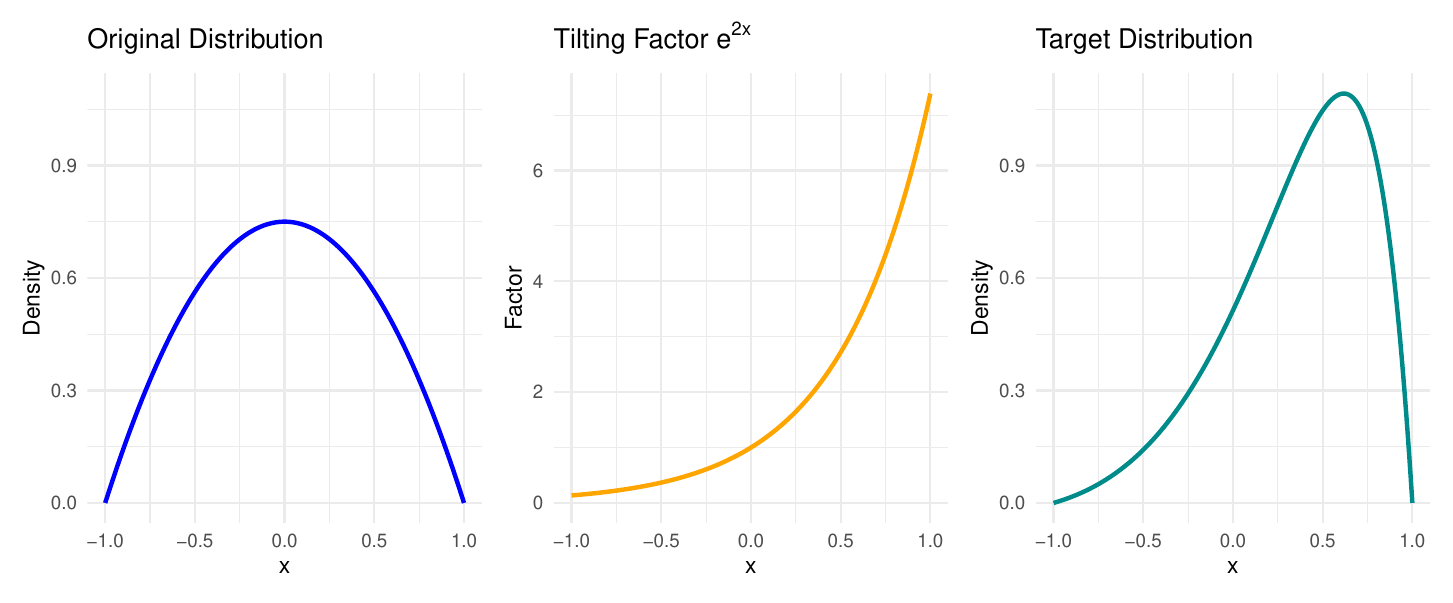}
\caption{The left and right panels depict the original and target distribution (after tilting) on the interval $[0,1]$.  The middle panel is the exponential tilting factor $e^{2x}$, which corresponds to the weight placed at each part of the original distribution.}
\end{figure}

\subsection{Generalizations to a conjugate odds property}

In this section, we discuss a general modeling strategy in which a parametric model is posited for the source domain and under nice conditions, a simple parametric model can also be obtained for the target domain.  A starting point is to first consider the factorization
\begin{equation}
p(x|A=a') \propto p(x|A=a) \cdot \frac{P(A=a'|x)}{P(A=a|x)}. \label{eq:tilted}
\end{equation}

From this factorization, we see that the source distribution can be perturbed towards the target distribution by multiplying by an odds factors.  In some situations, the odds and the target distributions have nice forms, which leads to the following idea of a conjugate odds.

\begin{definition}[Conjugate odds]
\label{def:conjugateodds}
Let $A\in\mathcal{A}$ be a categorical random variable that is auxiliary to the primary data $X$.  Suppose that $p(x|A=a)$ and $p(x|A=a')$ belong to the same probability model $\mathcal{P}$.  Then, we say the model formed by
$$\mathcal{O}(\mathcal{P}) := \{O_{a',a}(x) := P(A=a'|x)/P(A=a|x) : \forall a,a'\in\mathcal{A}\}$$
is a \textit{conjugate odds} for $\mathcal{P}$.
\end{definition}

In the definition for conjugate odds, we use the term conjugacy to relate it to the Bayesian literature and the idea of conjugate priors.  In Bayesian analysis, conjugate priors offer an algebraic convenience in that it provides a closed-form expression for the posterior given a specific likelihood function, thereby bypassing the need for numerical integration or computational methods.  In this paper, we say that a given odds family is \textit{conjugate} to a given family if it satisfies Definition \ref{def:conjugateodds}.  Moreover, the notion of conjugacy extends to a mixture model, where each component belongs to the same parametric family, as seen in Proposition \ref{prop:conjugatemixtures}.

\begin{proposition}[Conjugate odds holds under mixtures]
\label{prop:conjugatemixtures}
Suppose that $\mathcal{O}(\mathcal{P})$ is a conjugate odds for probability model $\mathcal{P}$.  Then, $\mathcal{O}(\mathcal{P})$ is a conjugate odds for the probability $K$-mixture model, where each component is an element of $\mathcal{P}$,
$$\mathcal{M}_K(\mathcal{P}) := \left\{p = \sum_{j=1}^K w_j p_j \ \biggr| \ p_j \in\mathcal{P}, \sum_{j=1}^K w_j = 1, w_j > 0\ \forall j \right\}.$$
\end{proposition}

We note that in general, a rejection sampling scheme is also possible.  One can posit a distribution for $p(x|A=a)$, and fit any odds model for $O_{a'}(x)$.  This can be any binary classifier, which extends this methodology to a suite of machine learning tools.  Then, as long as the odds factor $O_{a',a}(x)$ is bounded, then we can do a rejection sampling scheme by using $p(x|A=a)$ as a proposal distribution to shift towards our desired $p(x|A=a')$.  A bounded odds factor is reasonable if the variables $\bX$ belong to a bounded set.  Although we no longer have a closed-form expression for the target distribution, we are able to perform sampling.  This idea is further explored in Section \ref{sect:treegraphandconjugate}.


\subsection{Logistic odds}

Now, we provide our first example of a conjugate odds family by demonstrating that logistic regression is a conjugate odds for the exponential family.

\begin{proposition}[Exponential family, vector-valued random variable]
\label{prop:expfam-conjugate}
Suppose that $p(x|A=a)$ belongs to the exponential family parameterized by $\eta\in\mathcal{H}$,
$$p(x|A=a;\eta) = h(x)g(\eta)\exp(\eta^\top T(x))$$
Then, the associated odds model
$$O_{a'}(x;\bm{\gamma}) := \log \frac{P(A=a'|x)}{P(A=a|x)} = \gamma_0 + \gamma^\top T(x), \quad \gamma_0 := \log \frac{P(A=a')}{P(A=a)} + \log \frac{g(\eta')}{g(\eta)}$$
holds if and only if
$$p(x|A=a';\eta') = h(x)g(\eta')\exp((\eta')^\top T(x)),$$
where $\eta' := \eta+\gamma \in \mathcal{H}$.
\end{proposition}

A natural corollary of Proposition \ref{prop:expfam-conjugate} is the following result, which establishes a link between a logistic regression model and an exponential tilting factor.

\begin{corollary}[Exponential tilting and logistic regression]
\label{cor:expontentialtilting}
Imposing a logistic regression model on the odds $P(A=a'|x)/P(A=a|x)$ is equivalent to tilting a distribution $p(x|A=a)$ by an exponential factor.
\end{corollary}

Proposition \ref{prop:expfam-conjugate} has numerous applications, as the exponential family encompasses a broad class of parametric distributions for both discrete and continuous random variables, including the normal, exponential, binomial, Poisson, and negative binomial distributions. Since exponential tilting via logistic regression corresponds to a translation in the natural parameter space, the range of possible values for the natural parameter is of fundamental importance. This proposition further implies that when logistic regression is performed using the sufficient statistics of an exponential family, the fitted coefficients of these statistics directly determine the parameterization of the new distribution $p(x|A=a')$.

One key element of the proposition is the final condition $\eta' := \eta+\gamma\in\mathcal{H}$.  Although any pair of identically parameterized exponential family distributions permits a logistic regression representation of the odds, not all logistic regression and exponential family distribution pairs yield an exponential family representation for the target distribution.  Similarly, not all exponential tiltings lead to an exponential family and may not even result in a valid distribution.  This discrepancy arises when the translation shifts the natural parameter beyond its valid domain. Generally, this issue is mitigated when the natural parameter belongs to an unbounded space. For instance, in the case of the binomial distribution, the natural parameter $\eta:= \log p/(1-p)$ belongs to $\mathbb{R}$, and any translation stays within the set.

%
%

The result of Proposition \ref{prop:conjugatemixtures} can be applied to the exponential family, as seen in Corollary \ref{cor:mixexpfam-conjugate}.  There are a few illuminating examples that fall under these specific conditions such as the Gaussian mixture model and binomial product mixture model \citep{suen2023modelingmissingrandomneuropsychological}.

\begin{corollary}[Mixture of exponential family]
\label{cor:mixexpfam-conjugate}
Suppose that
$$p(x|A=a) = \sum_{k=1}^K w_k \cdot h(x) g(\eta_k) \exp(\eta^\top_k T(x)), \quad \log \frac{P(A=a'|x)}{P(A=a|x)} = \gamma_0 + \gamma^\top T(x).$$
Then, we have
$$p(x|A=a') = \sum_{k=1}^K \widetilde{w}_k \cdot h(x) g(\eta_k+\gamma) \exp((\eta_k+\gamma)^\top T(x)),$$
where
$$\widetilde{w}_k := \frac{w_k \cdot g(\eta_k)}{g(\eta_k+\gamma)} \bigg/ \sum_{k'=1}^K \frac{w_{k'} \cdot g(\eta_{k'})}{g(\eta_{k'}+\gamma)}.$$
\end{corollary}



\begin{example}[Gaussian mixture model with isotropic variance]
\label{example:GMM}
Since there are not many convenient options for off-the-shelf modeling of multivariate continuous data, practitioners often  use a Gaussian mixture model for its flexibility and relatively easy associated estimation procedure.  Suppose that
$$p(x|A=a) = \sum_{k=1}^K w_k \cdot \prod_{j=1}^d \frac{1}{\sigma_{k,j}^2\sqrt{2\pi}} \exp\left(-\frac{1}{2\sigma_{k,j}^2} (x_j - \mu_{k,j})^2\right),$$
$$\log \frac{P(A=a'|x)}{P(A=a|x)} = \gamma_0 + \gamma_1^\top x + \gamma_2^\top x^2.$$
Then,
$$p(x|A=a') = \sum_{k=1}^K w_k \cdot \prod_{j=1}^d \frac{1}{{\sigma_{k,j}'}^2\sqrt{2\pi}} \exp\left(-\frac{1}{2{\sigma_{k,j}'}^2} (x_j - \mu_{k,j}')^2\right),$$
where
$$\mu_{k,j}' := \frac{\mu_{k,j}/\sigma_{k,j}^2 + \gamma_{1,j}}{1/\sigma_{k,j}^2 - 2\gamma_{2,j}} , \quad (\sigma_{k,j}^2)' := \frac{1}{1/\sigma_{k,j}^2 - 2\gamma_{2,j}}, \quad \text{and} \quad w_k' = w_k \cdot \frac{g(\eta_k)}{g(\eta_k+\gamma)} \bigg/ \sum_{k=1}^K \frac{w_k \cdot g(\eta_k)}{g(\eta_k+\gamma)}$$
for $g(\eta_1,\eta_2) = \prod_{j=1}^d g(\eta_{j,1}, \eta_{j,2}) = \prod_{j=1}^d \exp\left(\frac{\eta_{j,1}^2}{4\eta_{j,2}^2}\right) \cdot \sqrt{-2\eta_{j,2}}$. 
\end{example}

\begin{example}[Binomial product mixture model]
Previously, \cite{suen2023modelingmissingrandomneuropsychological} introduced the binomial product mixture model to model multivariate discrete data.  Suppose that
$$p(x|A=a;w,p) = \sum_{k=1}^K w_k \cdot \prod_{j=1}^d \binom{N_j}{x_j} p_{k,j}^{x_j} (1-p_{k,j})^{N_j-x_j}, \quad \log \frac{P(A=a'|x)}{P(A=a|x)} = \gamma_0 + \gamma^\top x.$$
Then,
$$p(x|A=a';w',p') = \sum_{k=1}^K w_k' \cdot \prod_{j=1}^d \binom{N_j}{x_j} {p_{k,j}'}^{x_j} (1-{p_{k,j}'}^{x_j})^{N_j-x_j},$$
where
$$p_{k,j}' = \frac{\exp(\text{logit}(p_{k,j}) + \gamma_j)}{1+\exp(\text{logit}(p_{k,j}) + \gamma_j)} \quad \text{and} \quad w_k' = w_k \cdot \frac{g(\eta_k)}{g(\eta_k+\gamma)} \bigg/ \sum_{k=1}^K \frac{w_k \cdot g(\eta_k)}{g(\eta_k+\gamma)}$$
for $g(\zeta) = \prod_{j=1}^d \dfrac{1}{(1+\exp(-\zeta_j))^{N_j}}$.
\end{example}

\begin{example}[Gaussian kernel density estimator]
Suppose that $p(x|A=a)$ is fit nonparametrically using a kernel density estimator with a product Gaussian kernel as follows
$$\widehat{p}(x|A=a) = \frac{1}{nh^d} \sum_{i=1}^n \prod_{j=1}^d K\left(\frac{x_j-X_{i,j}}{h}\right),$$
where $K(t) = \frac{1}{\sqrt{2\pi}} \exp(-\frac{1}{2}t^2)$.

Thus, the KDE is a Gaussian mixture model with $n$ components (equally weighted), each being a multivariate Gaussian centered at each data point with covariance matrix $\text{diag}(h^2, h^2, \ldots, h^2)$.  Then, from Example \ref{example:GMM}, it follows that $\hat{p}(x|A=a')$ is a weighted Gaussian kernel density estimator.
\end{example}

%

The logistic model for odds is not the only possible model for conjugate odds; in Appendix \ref{sec::power}, we provide
an example of power law odds.

\section{Tree graphs and conjugate odds}

\label{sect:treegraphandconjugate}

With the conjugate odds, we develop an easy way to construct estimates of 1) the imputation distribution $p(x_{\bar{r}} | x_r, \bR=r)$ and 2) the conditional distribution $p(x|\bR=r)$.
We demonstrate that both of these tasks can be achieved in one shot by unifying the two frameworks (tree graph and conjugate odds) through the idea of Tukey's factorization.
A feature of tree graph is that our model  on $p(x|\bR=r)$ includes both observed variables as well as the missing variables. Therefore, the marginal distribution $p(x)$ can be obtain easily.

\begin{definition}[Tukey's factorization, \citep{Tukey1986}]
Consider a univariate $Y\in\mathbb{R}$ that is observed if $R=1$ and not observed if $R=0$.  We have the following factorization
$$p(y | R=r) = p(y|R=1) \cdot \frac{P(R=r|y)}{P(R=1|y)} \cdot \frac{P(R=1)}{P(R=r)}.$$
\end{definition}
Introduced by Tukey in a discussion \citep{Tukey1986}, the advantage of the above factorization is that identifies the missing data distribution as a product of two terms (one of which is $p(y|R=1)$ and can be estimated easily) and an odds term, which can be easier to think about and can naturally arise in many applications.  The key observation is that the above equation is reminiscent of an aforementioned factorization for tilting a distribution (Equation \eqref{eq:tilted}).  The term $p(y|R=1)$ is directly identifiable from the observed data.  The odds term $P(R=r|y)/P(R=1|y)$ depends on unobserved data, but can be identified using the tree graph framework.  From here, we can expect to borrow the tools from conjugate odds framework to tilt the complete case distribution $p(x|1_d)$.

\cite{Franks19045} previously built on the idea of Tukey's factorization as an alternative method from pattern-mixture models and selections models for modeling the full-data distribution.  In their work, they discuss this modeling strategy with a a single variable and two possible missing patterns.  We extend this work to handle the multivariate case.  Tukey's original factorization can be naturally generalized to a multivariate setting, as seen in the following definition that we propose.

\begin{definition}[Multivariate Tukey's factorization]
\label{def:multivariateTukey}
Consider a multivariate $X\in\mathcal{X}\subseteq\mathbb{R}^d$ with an associated missing pattern $R$.  We have the following factorization
$$p(x | \bR=r) =  p(x|\bR=1_d) \cdot \frac{P(\bR=r|x)}{P(\bR=1|x)} \cdot \frac{P(\bR=1_d)}{P(\bR=r)}.$$
\end{definition}

As in the univariate case, the above factorization demonstrates that $p(x|\bR=r)$ is proportional to a product of two terms: the complete case distribution $p(x|\bR=1_d)$ and an odds term $P(\bR=r|x)/P(\bR=1_d|x)$.  That is, this is another factorization for tilting a distribution as in Equation \eqref{eq:tilted}.  Importantly, this selection odds term is not directly identifiable without further assumptions.  However, Proposition \ref{prop:mnar} shows that under a tree graph assumption, these selection odds admits an elegant identification formula and can be estimated using the observed data.

\begin{assumption}[Absolute continuity with respect to the complete case distribution]
\label{assump:absolutecontinuity}
The distribution $p(x|\bR=r)$ is absolutely continuous with respect to the complete case distribution $p(x|\bR=1_d)$ for any $r\neq 1_d$.
\end{assumption}

When utilizing Tukey's factorization, one implicitly is making an assumption.  Assumption \ref{assump:absolutecontinuity} arises from the nonnegativity of the selection odds nonnegative: if $p(x|\bR=1_d)=0$, then $p(x|r) = 0$ must hold.  If the complete case distribution satisfies a positivity condition where $p(x|\bR=1_d)>0$ for all $x\in\mathcal{X}$, then this assumption will be trivially satisfied.  For instance, the mixture models presented in Section \ref{sect:conjugateodds} satisfy this positivity condition since each mixture component has positive probability on all of $\mathcal{X}$.


We now harmonize the two frameworks with the following theorem.

\begin{theorem}[Modeling pattern-specific joint distributions using tree graphs and conjugate odds]
\label{theorem:treegraphandconjugateodds}
Suppose the following conditions hold:
\begin{enumerate}
	\item The missingness mechanism is specified using a tree graph assumption.
	\item The odds model for the selection odds is conjugate to the $p(x|\bR=1_d)$ model.
\end{enumerate}
Then, the pattern-specific joint distributions $p(x|\bR=r)$ for all $r$ belong to the same family as $p(x|\bR=1_d)$.
\end{theorem}

Theorem \ref{theorem:treegraphandconjugateodds} is very powerful because it combines a tree graph and the conjugate odds property in the missing data context and demonstrates how that can lead to elegant modeling of the pattern-specific joint distributions $p(x|\bR=r)$.

\begin{example}[Tree graph with logistic regression and Gaussian model]
Suppose that a tree graph assumption holds, and the selection odds can all be modeled using a logistic regression.  Then, if $p(x|\bR=1_d)$ belongs to an exponential family, then $p(x|\bR=r)$ for any $r\in\mathcal{R}$ is also exponential family.  In particular, suppose that
\begin{itemize}
	\item $p(x|\bR=1_d)$ is a multivariate Gaussian
	\item all the selection odds can be modeled using logistic regression ($\log[P(\bR=r|x)/P(\bR\in \text{PA}_T(r)|x)] = \gamma_{0,r} + \gamma_{r}^\top T(x_r)$ with $T(\cdot)$ being a linear function)
\end{itemize}
then all the pattern-specific joint distributions $p(x|\bR=r)$ for all $r$ are multivariate Gaussian.  This idea generalizes to exponential family models and mixtures of exponential family models due to Propositions \ref{prop:conjugatemixtures} and \ref{prop:expfam-conjugate}.

\end{example}

A further consequence of Definition \ref{def:multivariateTukey} is that imposing a tree graph assumption and models for the odds leads to an explicit closed-form expression form $p(x|\bR=r)$.  Notably, this distribution factorizes as $p(x_{\bar{r}}|x_r,r)p(x_r|r)$, so the aforementioned procedure models both the observed data distribution and the missing data distribution in one shot.  This is an advantage over other methods such as \textit{mice}, which are able to generate Monte Carlo estimates from the imputation distribution but do not specify a form for the density of the observed or missing data distributions.  Because we obtain a specific form for the pattern-specific joint distribution $p(x|\bR=r)$ due to conjugacy, it is easier to interpret and also perform imputation without having to refit anything.

\subsection{Imputation via a conjugate odds approach}



In some settings, estimating a joint model $p(x,r)$ is not the end goal.  For example, some might want to complete the data using an imputation.  As previously mentioned, we are able to obtain a closed-form expression for the imputation distribution due to conjugacy.  We outline this in Algorithm \ref{alg:conjugateoddsimputation}.

\begin{algorithm}
{\bf Require:} $\{(\bX_{i,\bR_i}, \bR_i)\}_{i=1}^n$, a tree graph $T$
\begin{algorithmic}[1]
	\State{Fit the complete case model $p(x|1_d;\theta)$ with natural parameter $\eta$.}
	\For{$r\neq 1_d$}
	\State{Fit the logistic regression model $O_r(x_r; \beta_r ) := P(\bR=r|x_r)/P(\bR\in\text{PA}_{T}(r)|x_r)$ under the tree graph $T$}
	\EndFor
	\For{$r\neq 1_d$}
	\State{Compute the selection odds with respect to the source $1_d$ via $O_{r:1_d}(x;\lambda_r) := \prod_{r'} O_{r'}(x_{r'}; \beta_{r'})$.}
	\State{Form the parameter of the distribution $p(x|\bR=r; \theta_r)$ via $\theta_r := \eta^{-1}(\eta+\lambda_r)$.}
	\EndFor
	
	\For{$m=1,2,\ldots,M$}
	\For{$i=1,2,\ldots,n$}
	\If{$\bR_i\neq1_d$}
	\State{Impute $\tilde{\bX}_{i,\bar{\bR}_i}^{(m)}$ using the distribution $p(x_{\bar{r}}|x_r,\bR=r)$.}
	\State{Set the $m$-th imputed data to be $\tilde{\bX}_{i}^{(m)} := (\bX_{i,\bR_i}, \tilde{\bX}_{i,\bar{\bR}_i}^{(m)})$.}
	\EndIf
	\EndFor
	\EndFor
	
	\State{\textbf{return} $\left\{ \tilde{\bX}_{i}^{(m)} \right\}_{\substack{i=1,\ldots,n \\ m=1,\ldots,M}}$}
\end{algorithmic}
\caption{Conjugate odds imputation under logistic regression and exponential family}
\label{alg:conjugateoddsimputation}
\end{algorithm}

\subsection{Rejection sampling for imputation}

When odds are not modeling using a conjugate odds, then there may be challenges in finding a closed form expression for the imputation distribution.  However, provided that the odds terms are bounded away from infinity, it is possible to perform rejection sampling.  The key requirement is that there exists a constant $U_r$ such that the target density $f(x)$ satisfies $f(x) \leq U_r \cdot g(x)$ for all $x$, where $g(x)$ is the proposal density.  Here we would simply take the proposal distribution to be the complete case distribution $p(x_{\bar{r}}|x_r,\bR=1_d)$ and the target distribution to be $p(x_{\bar{r}}|x_r,\bR=r)$, our true imputation distribution.  Since the odds terms are bounded, this ensures that such an $U_r$ exists, making the rejection sampling procedure feasible. 
We outline this method in Algorithm \ref{alg:rejectionsampling}.  Although this approach may introduce additional computational overhead, it offers a flexible alternative when traditional sampling methods are not applicable due to the lack of a closed-form expression.

\begin{algorithm}
{\bf Require:} $\{(\bX_{i,\bR_i}, \bR_i)\}_{i=1}^n$, $U_r$ (an upper bound on the odds for pattern $r$), a tree graph $T$
\begin{algorithmic}[1]
	\State{Fit the complete case model $p(x|1_d;\theta)$ with natural parameter $\eta$.}
	\For{$r\neq 1_d$}
	\State{Fit the logistic regression model $O_r(x_r; \beta_r ) := P(\bR=r|x_r)/P(\bR\in\text{PA}_{T}(r)|x_r)$ under the tree graph $T$}
	\EndFor
	\For{$r\neq 1_d$}
	\State{Compute the selection odds with respect to the source $1_d$ via $O_{r:1_d}(x;\lambda_r) := \prod_{r'} O_{r'}(x_{r'}; \beta_{r'})$.}
	\State{Form the parameter of the distribution $p(x|\bR=r; \theta_r)$ via $\theta_r := \eta^{-1}(\eta+\lambda_r)$.}
	\EndFor
	
	\For{$m=1,2,\ldots,M$}
	\For{$i=1,2,\ldots,n$}
	\If{$\bR_i\neq1_d$}
	\While{$Y^{(m)}$ is not accepted}
	\State{Sample a proposal $Y^{(m)} \sim p(x_{\bar{r}}|x_r,\bR=1_d)$.}
	\State{Accept $Y^{(m)}$ with probability $\frac{p(Y^{(m)}|x_r,\bR=r)}{U_r \cdot p(Y^{(m)}|x_r,\bR=1_d)}$.}
	\EndWhile
	\State{Set the $m$-th imputed data to be $\tilde{\bX}_{i}^{(m)} := (\bX_{i,\bR_i}, Y^{(m)})$.}
	\EndIf
	\EndFor
	\EndFor
	
	\State{\textbf{return} $\left\{ \tilde{\bX}_{i}^{(m)} \right\}_{\substack{i=1,\ldots,n \\ m=1,\ldots,M}}$}
\end{algorithmic}
\caption{Imputation via rejection sampling under logistic regression and exponential family}
\label{alg:rejectionsampling}
\end{algorithm}

\section{Strategies for tree graph selection}

\label{sect:choosing-treegraph}

%

As established in Section~\ref{subsection:enumeration}, the number of tree graphs grows super-exponentially with the number of variables, making graph selection challenging. Moreover, Proposition~\ref{prop:mnar} shows that each tree graph encodes an MNAR assumption that cannot be rejected from observed data, underscoring the challenge for systematic selection strategies.

From the single-parent property of Proposition~\ref{prop:equivprop}, selecting a tree graph is equivalent to assigning each pattern a unique parent. This defines a function $\text{PA}_G : \mathcal{R}\backslash\{1_d\} \to \mathcal{R}$ with $\text{PA}_G(r) > r$ for all $r \in \mathcal{R}$, offering a compact and efficient way to encode tree structures. To guide practical construction, we propose three principles:		
\begin{enumerate}[leftmargin=*, label=\textbf{\arabic*.}]
	\item \textbf{Prior knowledge.} Select a parent for each pattern that follows prior or scientific knowledge.
	
	\item \textbf{Partial ordering.} Select a parent for each pattern based on an existing partial ordering principle (such as CCMV or NCMV). We discuss generalizations of the NCMV assumption to nonmonotone data in a later subsection.
	
	\item \textbf{Observed data distribution alignment.} If the observed data distributions under two missing patterns are similar, we may expect that the missing data distributions corresponding to the same missing patterns are similar as well.  Here we can use the data to identify most relevant parents to a given child. We provide two methods based on distributional distance. 
\end{enumerate}

In addition to the above three principles, one may randomly choose a tree graph and perform inference.
We provide a simple algorithm on how to sample a tree graph in Appendix \ref{sec::random}.


\subsection{Prior knowledge}

The first and most fundamental principle is to leverage prior knowledge when selecting a parent for each pattern. Scientific insights, domain expertise, or well-established theoretical foundations can provide strong guidance in determining plausible parent-child relationships. For instance, in a biological setting, hierarchical dependencies between genetic markers may be informed by known pathways or functional interactions. Similarly, in causal inference, domain knowledge may suggest directional dependencies between observed variables. By incorporating prior knowledge into the selection process, we ensure that the tree graph aligns with meaningful, interpretable structures that reflect real-world mechanisms.

\begin{example}[Longitudinal study with missingness due to dropout]
Consider a longitudinal study where the same test is measured with at regualar time intervals.  Then, suppose there is monotone missingness due to dropout.  We might hypothesize that individuals with missing pattern $\bR=1000$ and $\bR=0000$ are closely related because we might reason that individuals that never showed up to the study are most similar to individuals who only showed up to the first time point.  Then, one can connect the patterns $1000 \to 0000$.  This is related to the nearest-case missing value assumption (NCMV; \cite{Thijs}).
\end{example}

\begin{example}[Hierarchical data collection processes]
Suppose we have four collected variables: $X_1$, which corresponds to a routine check-up measure such as blood pressure, $X_2$, representing a disease state like chronic kidney disease (CKD), $X_3$, which measures swelling (a common symptom of CKD), and $X_4$, a clinical test result assessing kidney function. In medical settings, it is common for $X_3$ and $X_4$ to be recorded only when $X_2$ exceeds a certain threshold, indicating a more severe condition. Consequently, missing data patterns such as $R = 1000$, $0100$, and $1100$ may arise. Since individuals missing $X_3$ and $X_4$ are likely healthier, it is plausible to infer hierarchical relationships between these missing patterns, such as $1100 \to 1000$ and $1100 \to 0100$, where the presence of both symptom and test data informs cases where only one or neither is recorded.

\end{example}

\begin{example}[Group similarity]
Suppose we have three variables: $X_1$, a self-reported stress level, $X_2$, alcohol consumption (e.g., self-reported drinks per week), and $X_3$ exercise habits (e.g., frequency of physical activity per week).  There may be a social stigma associated with alcohol, which is related to underreporting and even missingness.  We posit that the groups $\bR=100$ and $101$ are similar in that they are more likely to suffer from such social stigma, so we may suggest a relationship $101\to100$.
\end{example}

\subsection{Partial ordering}

In monotone missing data problem, 
some assumptions, such as nearest case missing value (NCMV), utilizes  an ordering on missing data patterns and also admit scientific interpretations. 
This is possible because in the monotone missing data setting, and each pattern has a parent that is the unique pattern that contains exactly one more observed variable.
In the nonmonotone missing data setting, the possible parent is no longer unique because there are multiple possible patterns that contain one more observed variable.  For example, in the monotone missing data situation $1000$ would have parent $1100$, but in the nonmonotone missing data setting, it could have parent $1100$, $1010$, or $1001$.
To resolve this issue, we relax the ordering into partial ordering and propose the following generalization. 

\begin{definition}[Generalized nearest case missing value (GNCMV)]
\label{def:GNCMV}
A tree graph is called a \textit{generalized nearest case missing value} assumption if every pattern in the graph has a parent that contains exactly one more observed variable.
\end{definition}

The GNCMV is still a large class of tree graphs. 
To choose a reasonable tree under GNCMV, we consider two special cases: the \textit{leftmost first approach} (LNCMV)
and the  \textit{rightmost first approach} (RNCMV). 
LNCMV is the tree graph where the parent is the pattern where the leftmost first $0$ (missing varaible) is replaced by $1$. 
RNCMV is defined similarly but we replace the rightmost first $0$ by $1$.
For example, the pattern 01010 has three possible parents under GNCMV: 11010, 01110, 01011. 
The LNCMV chooses 11010 as its parent while RNCMV chooses 01011. 
We visualize both of these ideas in Figure \ref{fig:impute_one}.
%


\begin{figure}[t]
\centering
\includegraphics[width=0.65\textwidth]{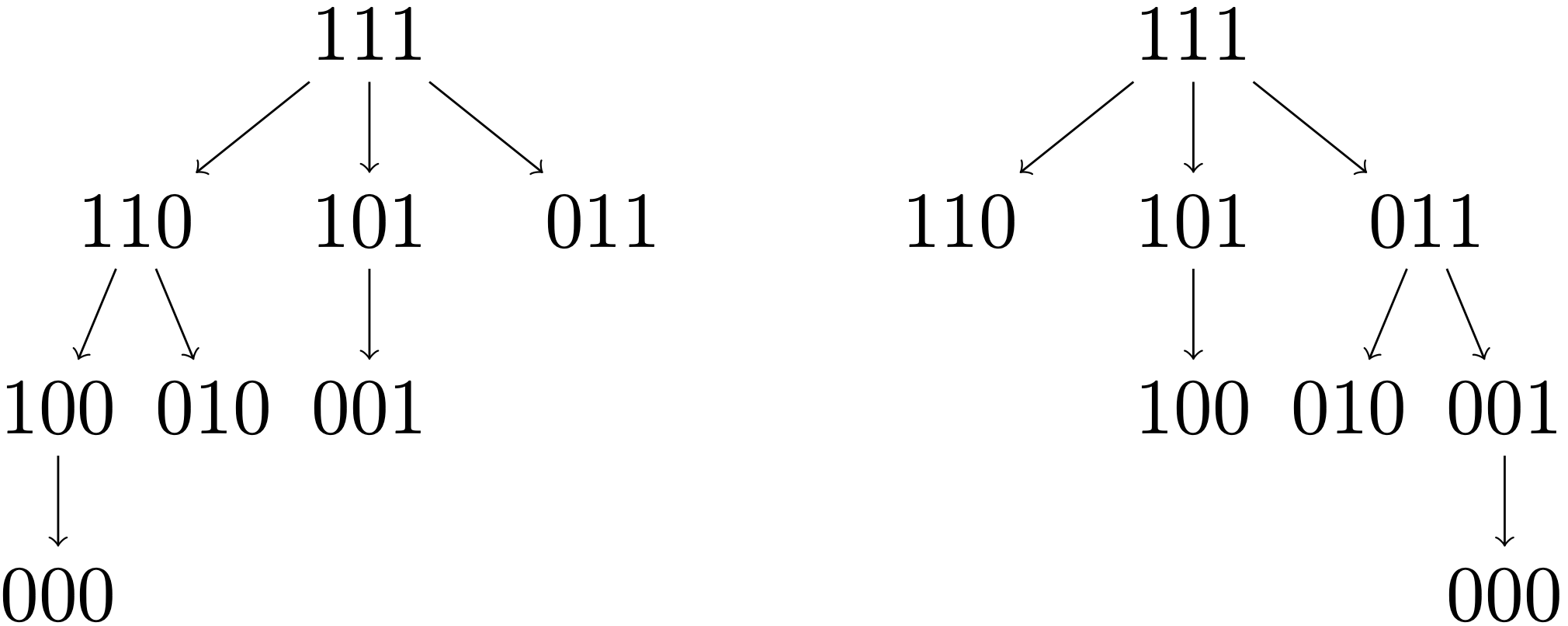}
\caption{The left panel corresponds to the tree graph for $d=3$ variables and the leftmost-first NCMV (LNCMV) assumption.  The right panel corresponds to the tree graph for $d=3$ variables and the rightmost-first NCMV (RNCMV) assumption.}
\label{fig:impute_one}
\end{figure}

\begin{proposition}[Generalized nearest-case missing value tree graphs]
Denote the subset of tree graphs that exhibit the GNCMV property as $\mathcal{T}_{\text{GNCMV}}$.  We have $\log |\mathcal{T}_{\text{GNCMV}}| \gtrsim 2^d$.  Moreover, for any $T \in \mathcal{T}_{\text{GNCMV}}$, $T$ exhibits the following properties:
\begin{enumerate}
	\item It achieves the maximum possible depth of $d$.
	\item Every pattern $r$ in $T$ is positioned at the maximum possible distance from the source node $1_d$, thereby corresponding to the most information flow.
\end{enumerate}
\label{prop:GNCMV}
\end{proposition}

Observe that by pruning the nonmonotone patterns from each graph, we recover the tree structure that would exist under the NCMV assumption with monotone missingness. Moreover, the $k$-th layer contains $\binom{n}{k}$ patterns. This assumption stands in direct contrast to the CCMV assumption, as each pattern is positioned at the maximum possible distance from the source, representing the opposite structural arrangement.

\subsection{Observed data distribution alignment}

When attempting to infer the structure of a tree graph from data, a natural question is: how should we identify the most plausible parent nodes for a given node?  Our data-driven method offers a principled way to rank candidate parents using observed distributions.  The data-driven method is an approach one can use to rank potential parents from the data, thereby informing the tree graph structure from the existing data.  By assumption, the tree graph asserts the following equality for every pattern $r$
\begin{equation*}
p(x_{\bar{r}}|x_r,\bR=r) \stackrel{T}{=} p(x_{\bar{r}}|x_r,\bR=\text{PA}_{T}(r)).
\end{equation*}
Thus, one natural idea is to only match extrapolation distributions if the observed data distributions under both $\bR=r$ and $\bR\in\text{PA}_{T}(r)$ are similar.  More precisely, we would desire $d(p(x_r|\bR=r), p(x_r|\bR=\text{PA}_{T}(r)))$ to be small for some probability metric or divergence $d$.  This motivates two possible matching approaches.  While we describe them in the context of a likelihood method, we note that matching approaches can be more general.

\textbf{Parent-based alignment.}  In the first, suppose we obtain data $X_{1,r}, X_{2,r}, \ldots, X_{n_r,r} \sim p(x_r|r)$ and attempt to determine which parent distribution $p(x_r|s)$ has the best fit, among all possible parents $s$.  In practice, for each $s\in\text{PPA}(r)$, we fit a parametric model for $p(x_r|s)$ and estimate the expected log-likelihood calculated on the data.  This procedure can be expressed in the population version as
$$\text{argmax}_{s\in\text{PPA}(r)}\E_{X_r\sim p(x_r|\bR=r)}[\log p (X_r|\bR=s)].$$

The KL divergence provides an alternate perspective.  Through a series of equalities, we have
\begin{align*}
\text{argmin}_{s\in\text{PPA}(r)}  \ D_{\text{KL}}(p(x_r|\bR=r) \ || \ p(x_r|\bR=s))
&= \text{argmin}_{s\in\text{PPA}(r)}  \int_{-\infty}^\infty p(x_r|r) \log\frac{p(x_r|r)}{p(x_r|s)} \ dx_r \\
&= \text{argmax}_{s\in\text{PPA}(r)} \int_{-\infty}^\infty p(x_r|r) \log p(x_r|s) \ dx_r \\
&= \text{argmax}_{s\in\text{PPA}(r)} \E[\ell(X_r|s)].
\end{align*}
This highlights the fact that the maximization procedure we propose is directly equivalent to picking the pattern that minimizes the sample version of the KL divergence between $p(x_r|\bR=r)$ and $p(x_r|\bR=s)$.  Implementation of this procedure in practice can be most efficiently done by first estimating each model $p(x_r|r)$ for all $r\in\mathcal{R}$ and then storing each model.  We present connections to the KL divergence, but we also note that one could certainly extend this to other distances.  More generally, other $f$-divergences or metrics such as the Wasserstein distance could be explored, particularly when distributional smoothness or support mismatch is a concern.  For example, the Hellinger distance can also be utilized and has the nice property that it is a bounded metric.  While the KL divergence is easy to implement with a given model, it is generally not possible in nonparametric settings.  In those settings, one could consider distances between distributions via an energy-based approach. We provide an example of how this could be done in Appendix \ref{appendix:MARsim}.

\textbf{Child-based alignment.} There is an alternative approach through a child-based alignment approach.  In contrast to the above, suppose we obtain data $X_{1,s}, X_{2,s}, \ldots, X_{n_s,s} \sim p(x_s|s)$ and attempt to determine which child distribution $p(x_r|r)$ has the best fit.  We outline this in further detail in the Appendix.  Provided the fitted models are stored in memory, both the parent-based and child-based approaches have similar computational complexity, but the parent-based method has an illuminating theoretical interpretation when using the KL divergence.  Note that if a proper distance/metric is used to compare distributions, the parent-based and child-based alignment be the same; their difference is due to the asymmetry of the KL divergence. In simulation, we demonstrate that both the parent-based and child-based modeling approaches are able to learn the correct tree graph given enough sample size in some settings.  This is discussed in Appendix \ref{sect:simulations}.

\section{Real data}

\label{sect:realdata}
Here we illustrate the applicability of our method using an Alzheimer's disease data with a mixture of binomial product model. 
We also provide an example of using KDE on wine data in Appendix \ref{appendix:kde}.

\subsection{NACC data}

We consider the analysis of neuropsychological test scores in the database of the National Alzheimer's Coordinating Center (NACC)\footnote{https://naccdata.org/}.
The National Alzheimer's Coordinating Center (NACC), funded by the NIH and NIA, oversees the largest longitudinal database on Alzheimer's disease in the United States. It serves as a coordinating hub for 33 Alzheimer's Disease Research Centers (ADRCs) across the country.  This data set comprises individuals of varying cognitive status: cognitively normal to mild cognitive impairment (MCI) to dementia.  Each individual is assigned a CDR (clinical dementia rating) from clinician with $0$ corresponding to cognitively normal, 0.5 corresponding to mild cognitive impairment, and 1, 2, and 3 corresponding to mild, moderate, and severe dementia, respectively.

Typically, neuropsychological assessments are conducted annually, but incomplete outcome data is common for various reasons. In some cases, specific tests are discontinued over time and substituted with alternative measures.  In others, missing scores may result from documentation errors or from participants being too unwell to complete further testing.

\subsubsection{Description of outcome variables and covariates}

Our main goal is to measure and model the cognitive ability of the Alzheimer's disease patients.  We focus on the following variable UDSBENTD, which is the total score for ten to fifteen minute delayed drawing of Benson figure.  In the Benson figure test, a participant is presented with a diagram of a complex figure and is asked to copy it.  After a period of about ten to fifteen minutes, they are asked to recopy it again from memory, and they are assigned a score from 0 to 17 based on how well it resembles the original figure.  This test measures visuospatial, visual memory, and executive abilities.  We look at individuals who entered the study from the years 2015 to 2019 and follow them for five years total, examining the repeated delayed Benson figure test score each year.  We do not use the CDR score in the model, but we use to help report and interpret the results.

We first plot the missing pattern distribution in Figure \ref{fig:missingpatterndist}.  We can initially observe that the complete cases are very small with $n_{cc}=271$ individuals out of a possible $n=13440$.  Additionally, every possible pattern of the 16 possible is observed, ranging from $\bR=10000$ to $\bR=11111$.  The distribution is primarily dominated by the monotone missing patterns $10000$, $11000$, $11100$, $11110$, and $11111$, likely due to dropout.  Of primary interest, we will examine the patterns $10000$, $11000$, $10100$, and $10010$ because they are some of the larger patterns.

\begin{figure}
\centering
\includegraphics[height=3in]{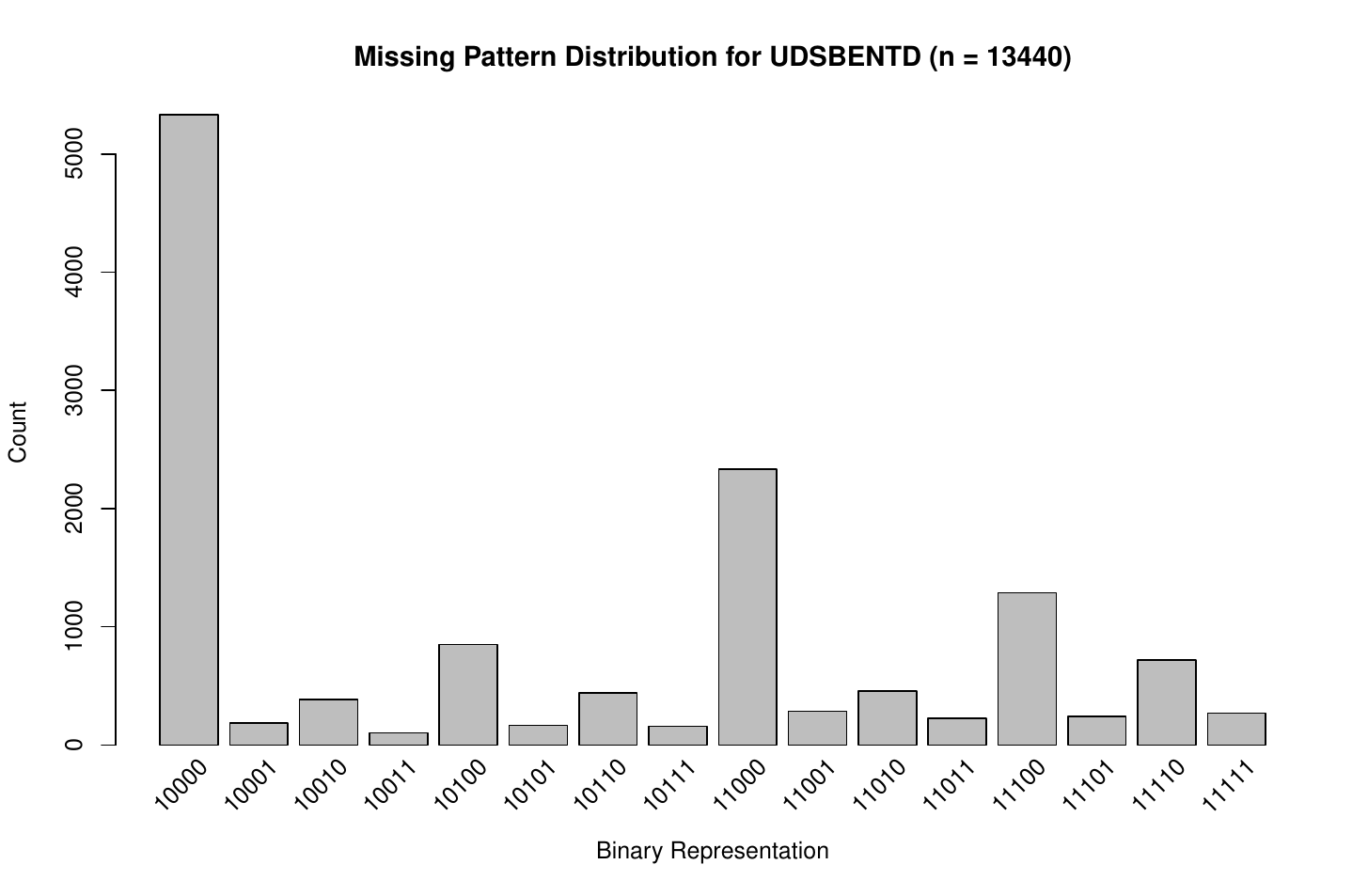}
\caption[Missing pattern distribution for the UDSBENTD variable over 5 years.]{This is the missing pattern distribution for the UDSBENTD variable over 5 years.}
\label{fig:missingpatterndist}
\end{figure}

\subsubsection{Analysis of NACC data}

We next plot four different tree graphs of interest: LNCMV, RNCMV, parent-based modeling, and child-based modeling.  These four tree graphs are reported in Figure \ref{fig:treegraphs_NACC}.  Interestingly, they all share similar maximum depth.  The parent-based modeling is able to generally able to recover the LNCMV principle for many of the patterns, including most of the monotone missing patterns.  On the other hand, the child-based modeling appears to incorporate a mix of both LNCMV and RNCMV principles when assigning a pattern to a given parent.  Because the patterns $10000$, $11000$, $10100$, and $10010$ each share similar ancestors in both the LNCMV and parent-based modeling tree graph, we would expect that the two tree graphs lead to similar fitted distributions at the end.

\begin{figure}
\centering
\includegraphics[width=0.47\textwidth]{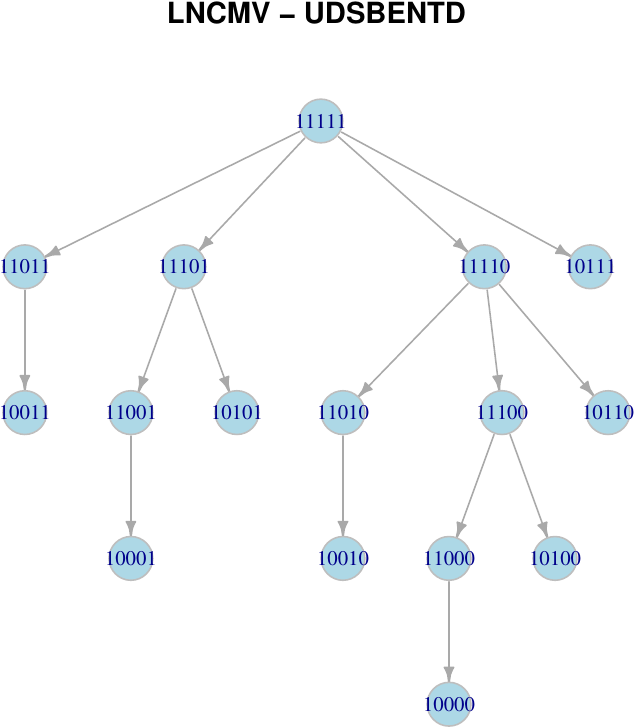}
\includegraphics[width=0.47\textwidth]{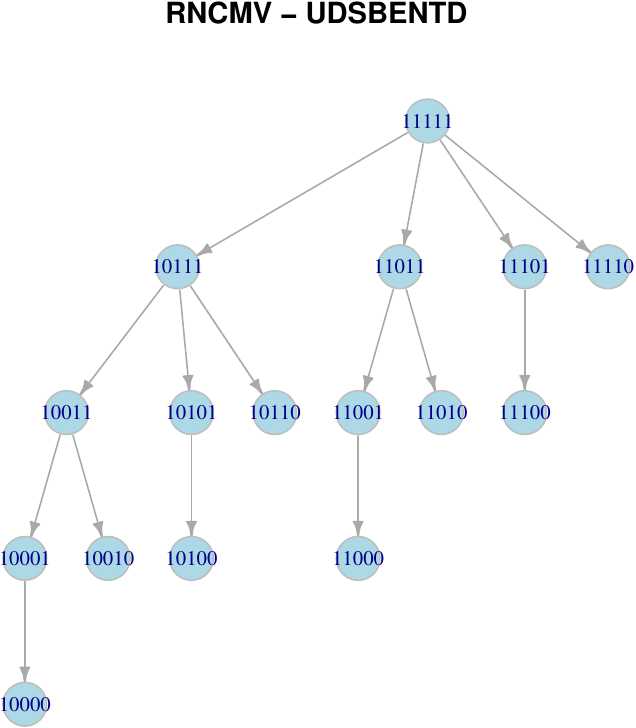}
\vskip 10pt
\includegraphics[width=0.47\textwidth]{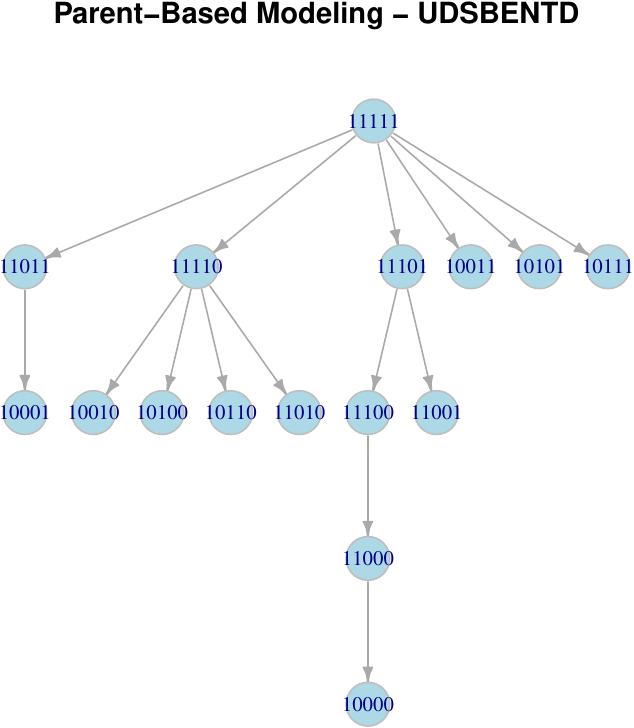}
\includegraphics[width=0.47\textwidth]{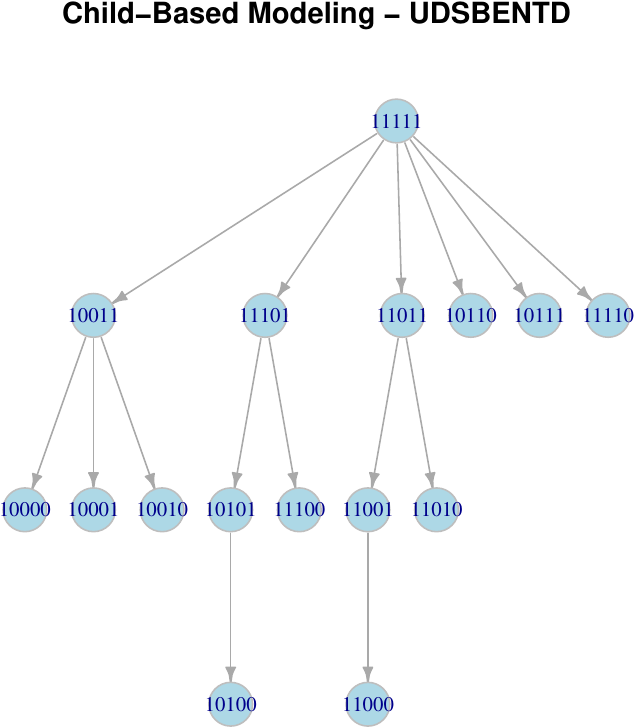}
\caption[Tree graphs for UDSBENTD data analysis]{These are the tree graphs for the UDSBENTD data analysis obtained via different methods: LNCMV, RNCMV, parent-based modeling, and child-based modeling.}
\label{fig:treegraphs_NACC}
\end{figure}

In Figure \ref{fig:cc_marginal_NACC}, we plot the results of the fitted model $p(x|\bR=1_d)$ using a mixture of binomial products.  We fit it using 5 clusters because of the recommendation from BIC.  From the first five panels, we can see that it is roughly able to capture the shape of the marginal distributions.  In the sixth panel, we include plots of the five clusters we obtain as latent trajectories over the five time points.  Each cluster is represented by a curve with the observed average CDR score reported at each dot.  A given dot corresponds to the predicted mean UDSBENTD score from the model for a given cluster and year.  Because the average CDR score is close to 0, the first two clusters represent cognitively normal people.  There is also some evidence of a practice effect between years 1 and 3 because the scores increase over those years \citep{Goldberg2015Practice}.  The third clusters can be interpreted as mild cognitively impaired people because the average CDR score is close to 0.5.  The fourth cluster appears to represent mild cognitively impaired people transitioning to dementia.  The fifth cluster appears to be mostly mildly cognitive impaired or dementia individuals.

\begin{figure}
\centering
\includegraphics[height=2.5 in]{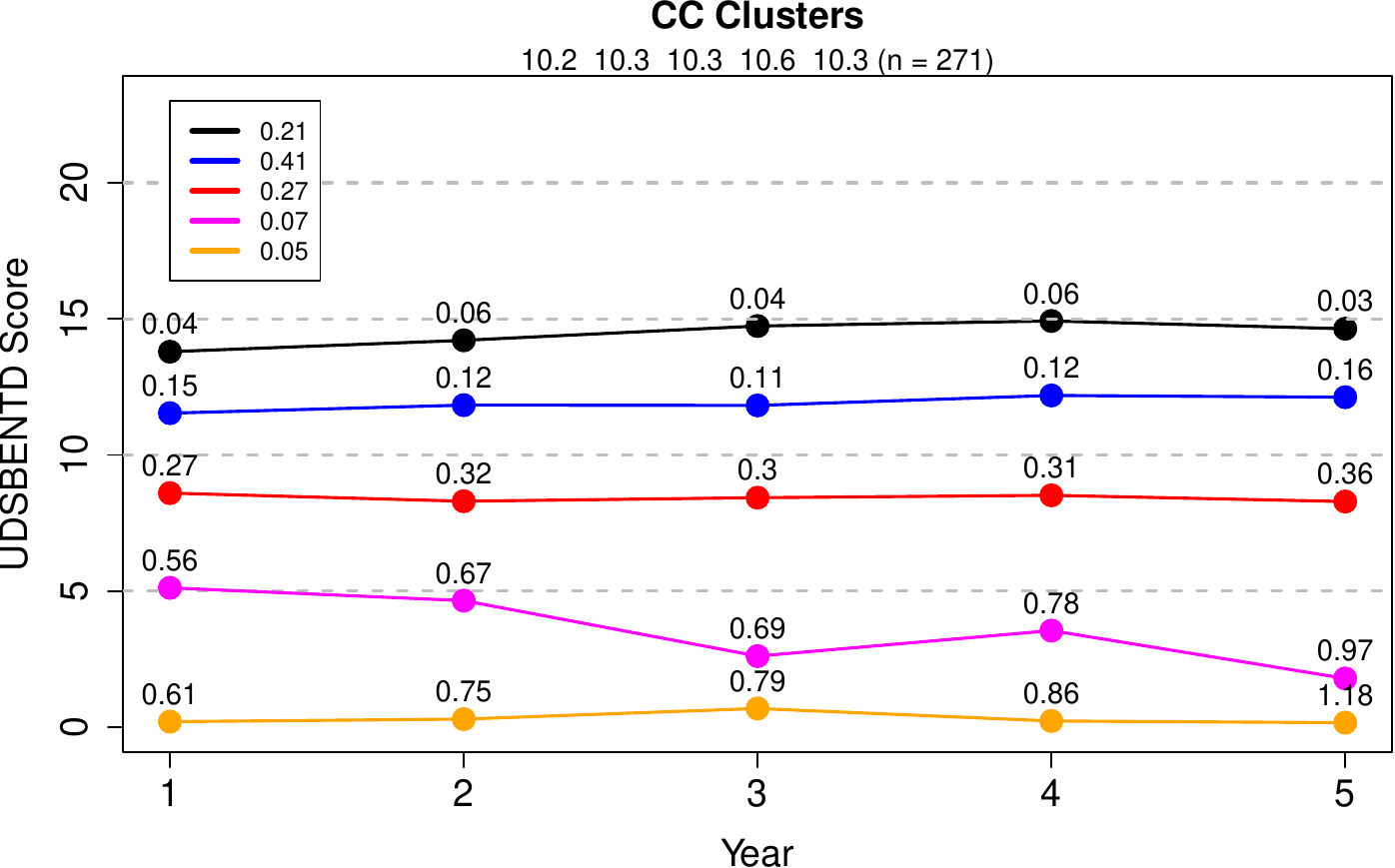}
\caption{Fitting the mixture of binomial products model on the complete case distribution with $k=5$ components. 
	Each component's parameters correspond to a curve. The numbers on top of each parameter value indicates
	the average CDR score, a clinical measurement of the cognitive decline level, which was not used in our model fitting (the CDR score serves as a external validation of 
	our fitted model). Note that for the pink component, it shows a learning effect on year 3 to year 4 that the score was improving while the clinical assessment (CDR score) 
	of the cognitive ability is declining. }
\label{fig:cc_marginal_NACC}
\end{figure}

In Figure \ref{fig:obs_dist_NACC}, we plot our fitted model using the conjugate odds method against the observed marginal distributions as a diagnostic check.  Because our method models both the imputation distribution and observed data distribution in one shot, it is important to perform this diagnostic check to have confidence in the imputation distribution results.  The parent and LNCMV graphs have similar results while the CCMV and RNCMV graphs have similar results.  So, we report results from RNCMV, parent-based modeling, and child-based modeling.  All resulting models for the observed-data distribution appear to fit to the data reasonably well in the majority of settings and the complete-case distribution (CCA) generally fails to capture the peak at $0$ for most of the observed data distributions.

\begin{figure}
\centering
\includegraphics[width=0.47\textwidth]{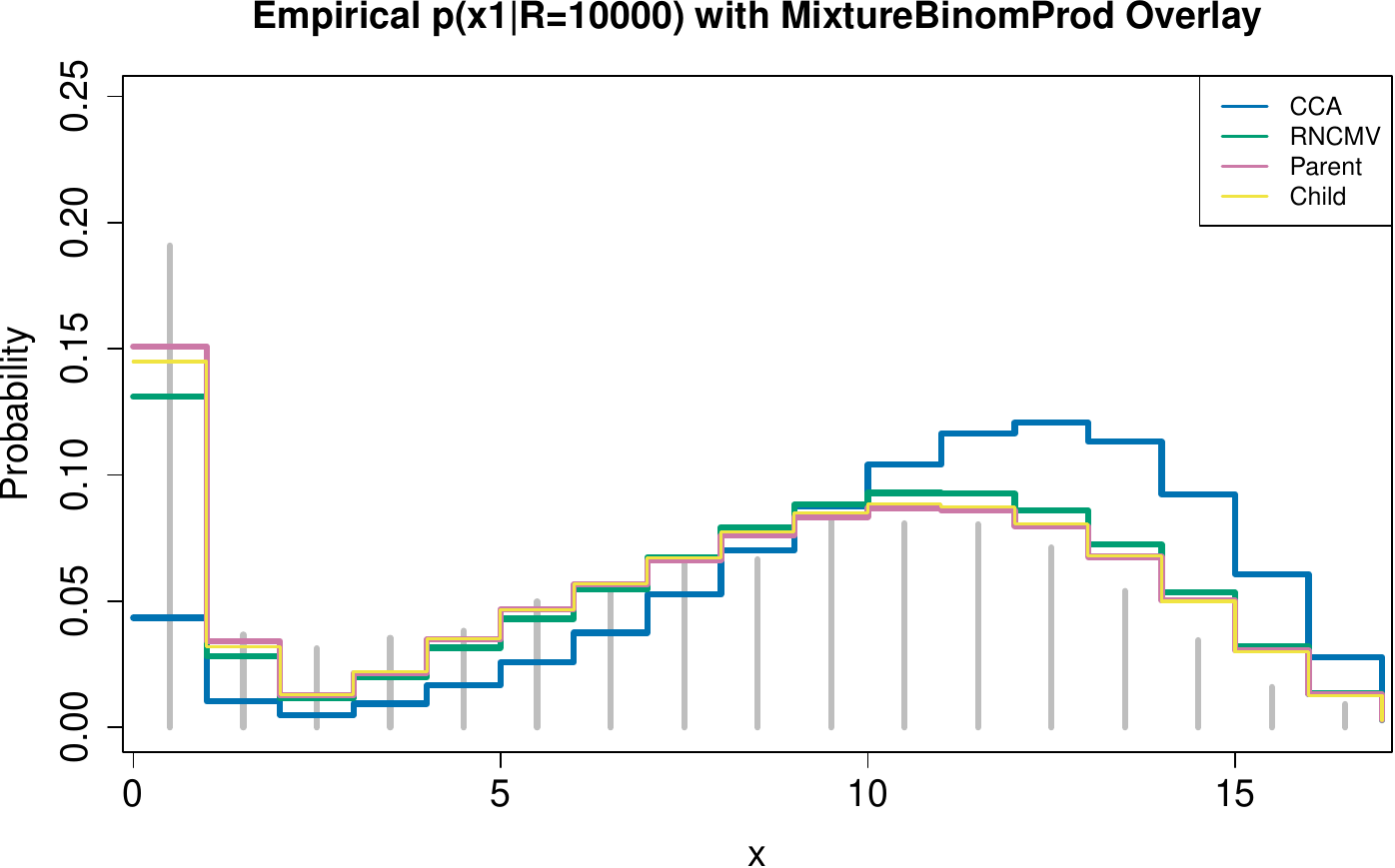}
\includegraphics[width=0.47\textwidth]{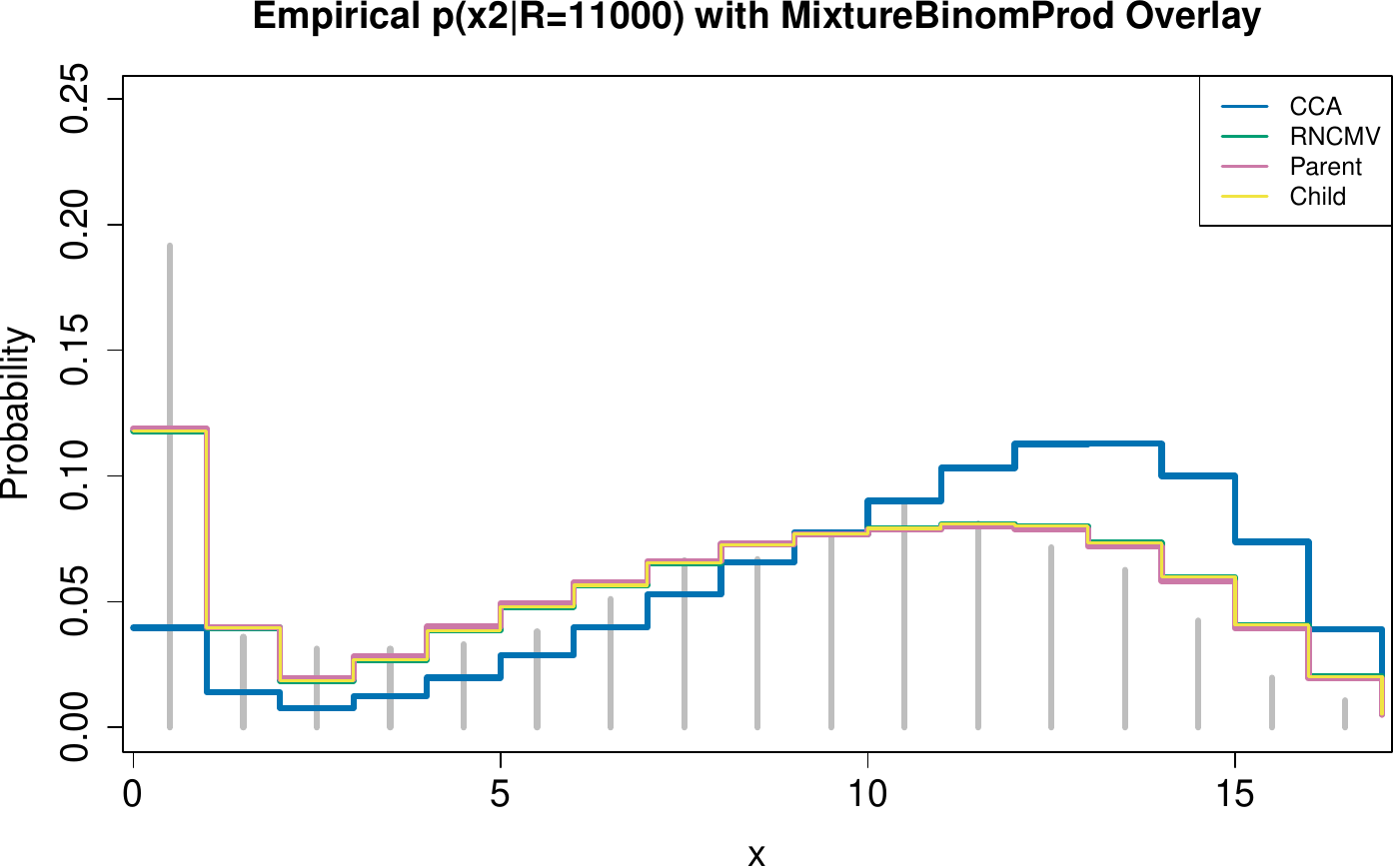}
\caption{
	Two examples of how the observed-data distribution model is improved by the conjugate odds. 
	The grey vertical lines indicates the empirical distribution of the variable. 
	The four colored histograms indicate the fitted distribution on each variable from: CCA, RNCMV, parent-based, and child-baesd. 
	The latter three methods are tree graph with conjugate odds and they all show a huge improvement over the CCA. 
}
\label{fig:obs_dist_NACC}
\end{figure}

For the observed marginal distributions for patterns $11000$ and $10010$, the fit from the different tree graphs is comparable.  From the marginal distributions $p(x_1|\bR=10000)$ and $p(x_1|\bR=10100)$, we generally see that the parent-based modeling tree graph yields fitted models that generally approximate the observed distributions better.  Thus, for the following plots in Figure \ref{fig:clusteringNACC}, we report the fitted result for $p(x|\bR=r)$ using parent-based modeling and contrast it with imputing with \textit{mice} and then fitting a mixture of binomial products model.  We see that the clusters across the different patterns $10000, 11000, 10100$, and $10010$ are generally very similar for parent-based, but the weights change.  For example, for $\bR=10000$, there is more weight towards the unhealthier clusters.  We also note that \textit{mice} yields relatively similar clusters as well in terms of trends, but it suffers from the model incompatibility problem \citep{meng1994multiple}, being longer to fit, and cannot handle MNAR data.

In the complete data (Figure \ref{fig:cc_marginal_NACC}), we observe a learning effect for the pink component.
Such learning effect was visible when we perform imputation via tree graphs (left column).
However, for the MICE, this effect was only observed in the case of $\bR=10010$ (bottom-right panel). 
Note the average CDR score (the number on top of each dot) is only observed partially because when the individual is missing from that year's data, the CDR score is missing as well. 

\begin{figure}
\centering
\includegraphics[width=0.47\textwidth]{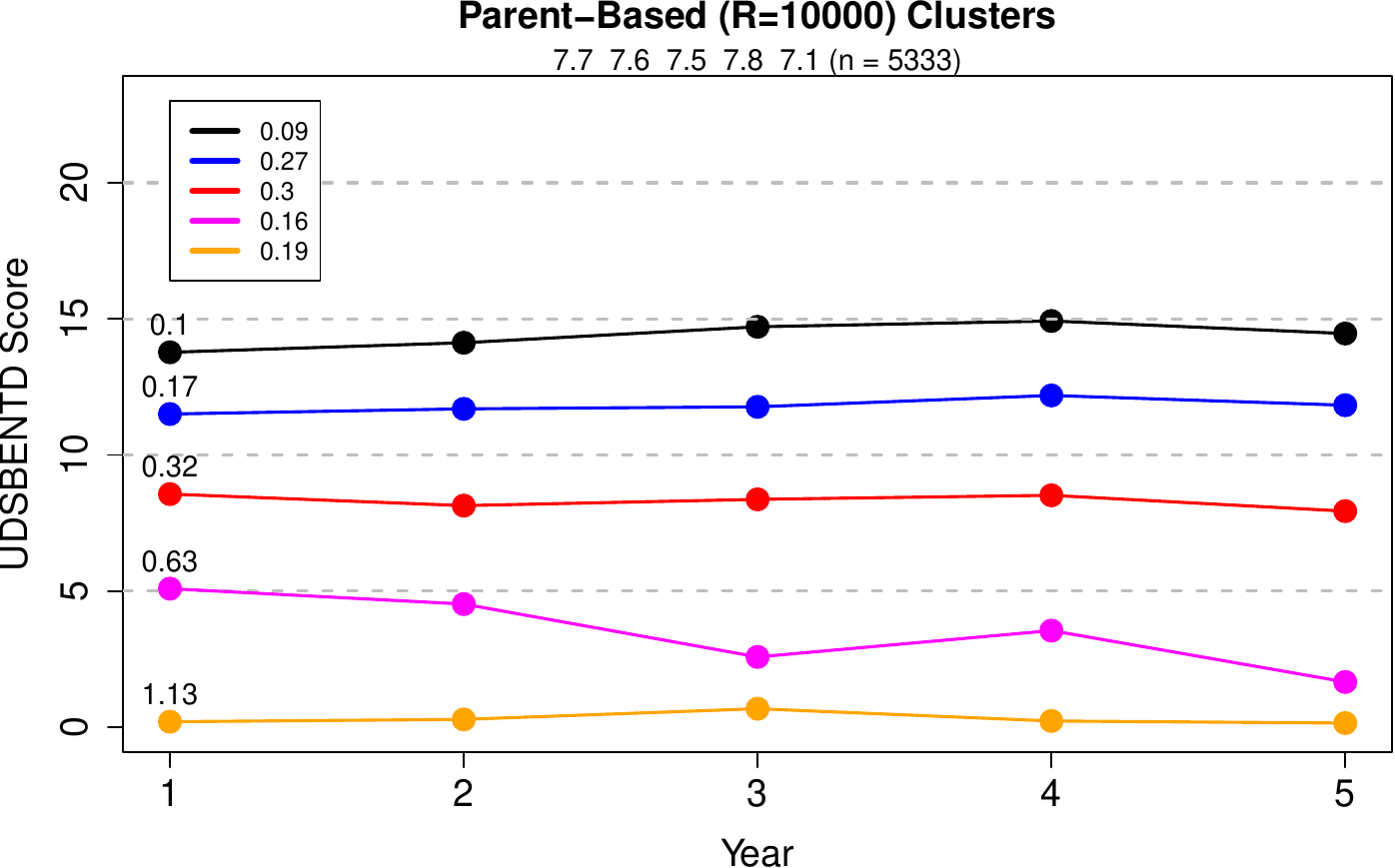}
\hspace{0.2in}
\includegraphics[width=0.47\textwidth]{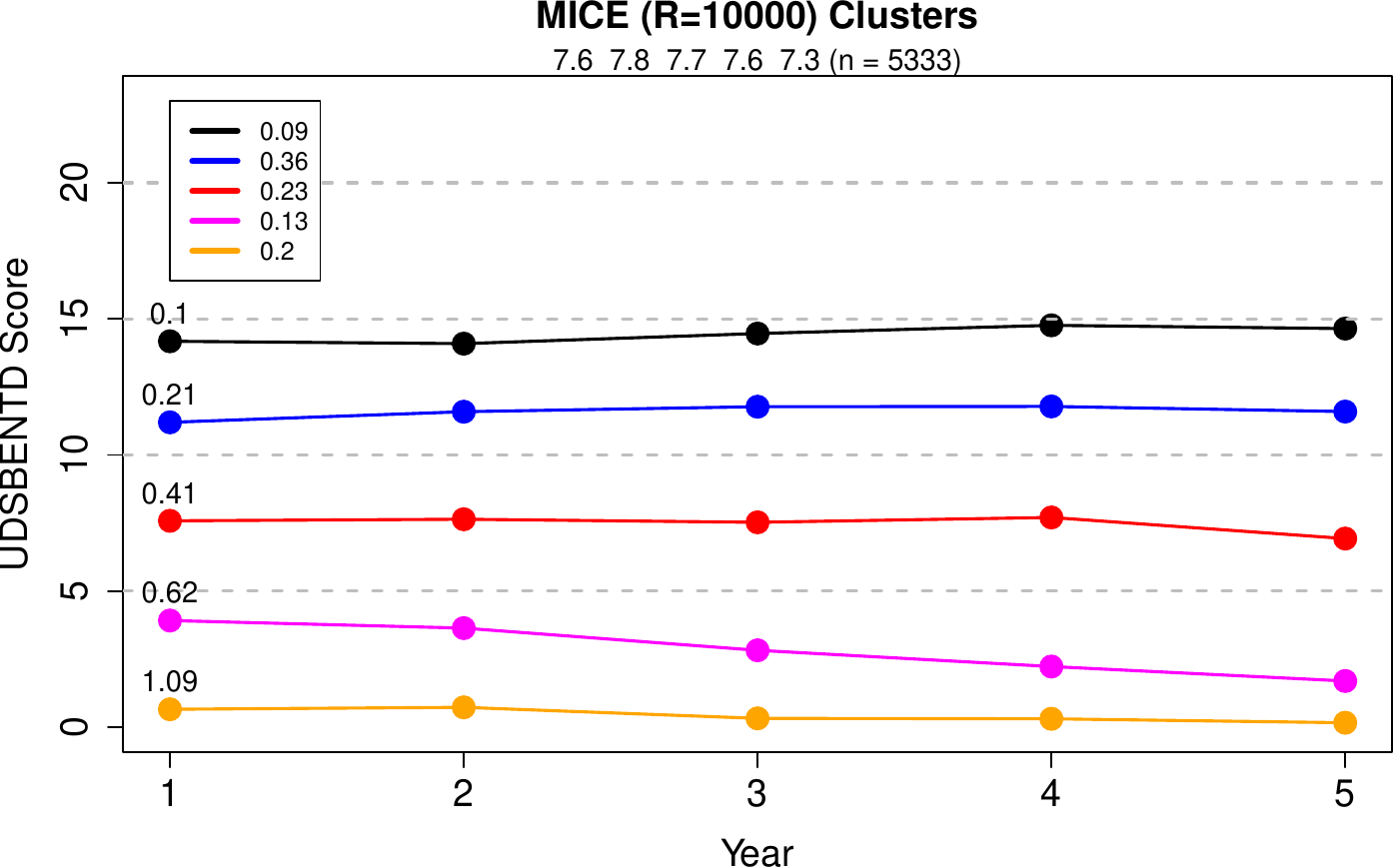}
\vskip 10pt
\includegraphics[width=0.47\textwidth]{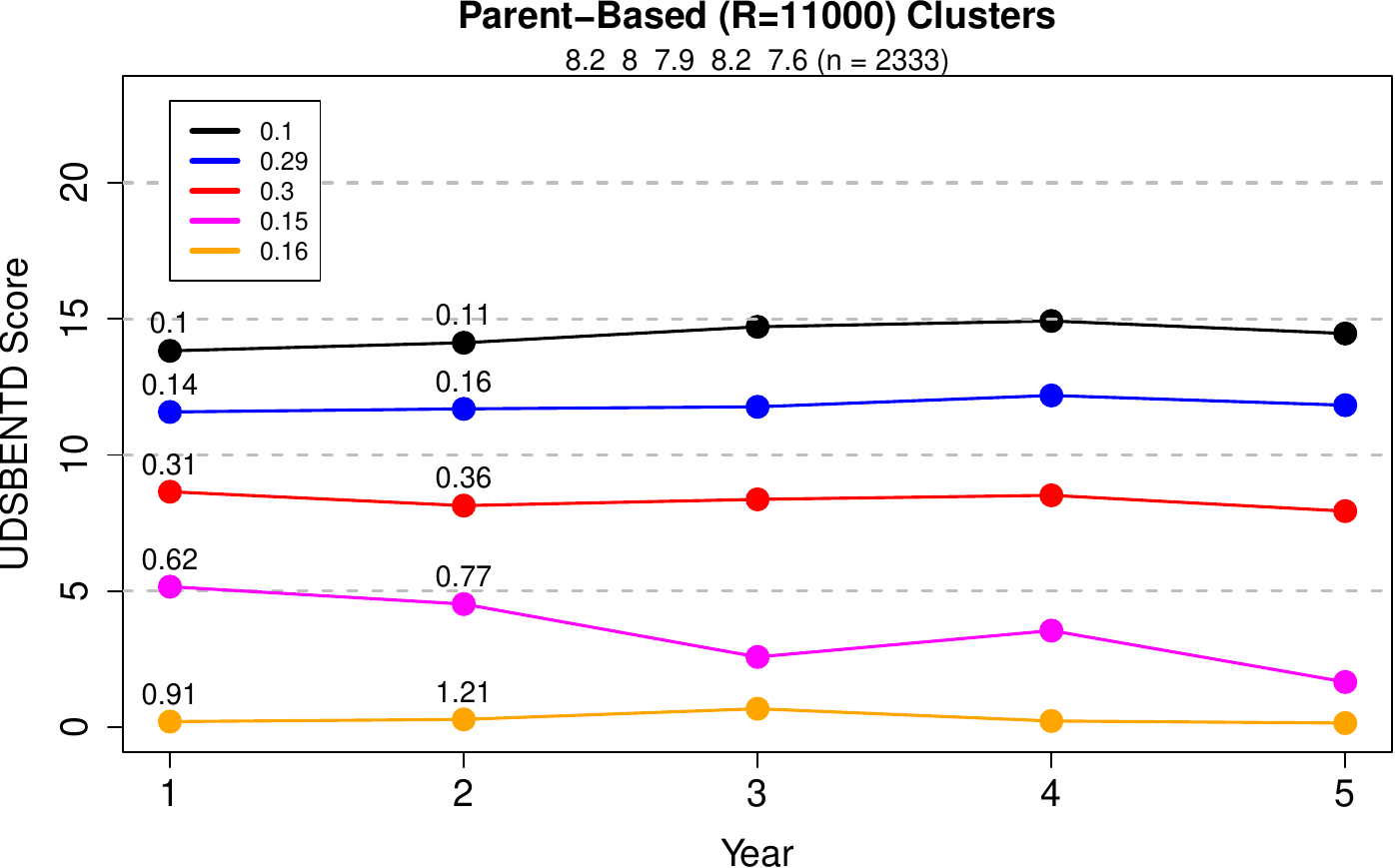}
\hspace{0.2in}
\includegraphics[width=0.47\textwidth]{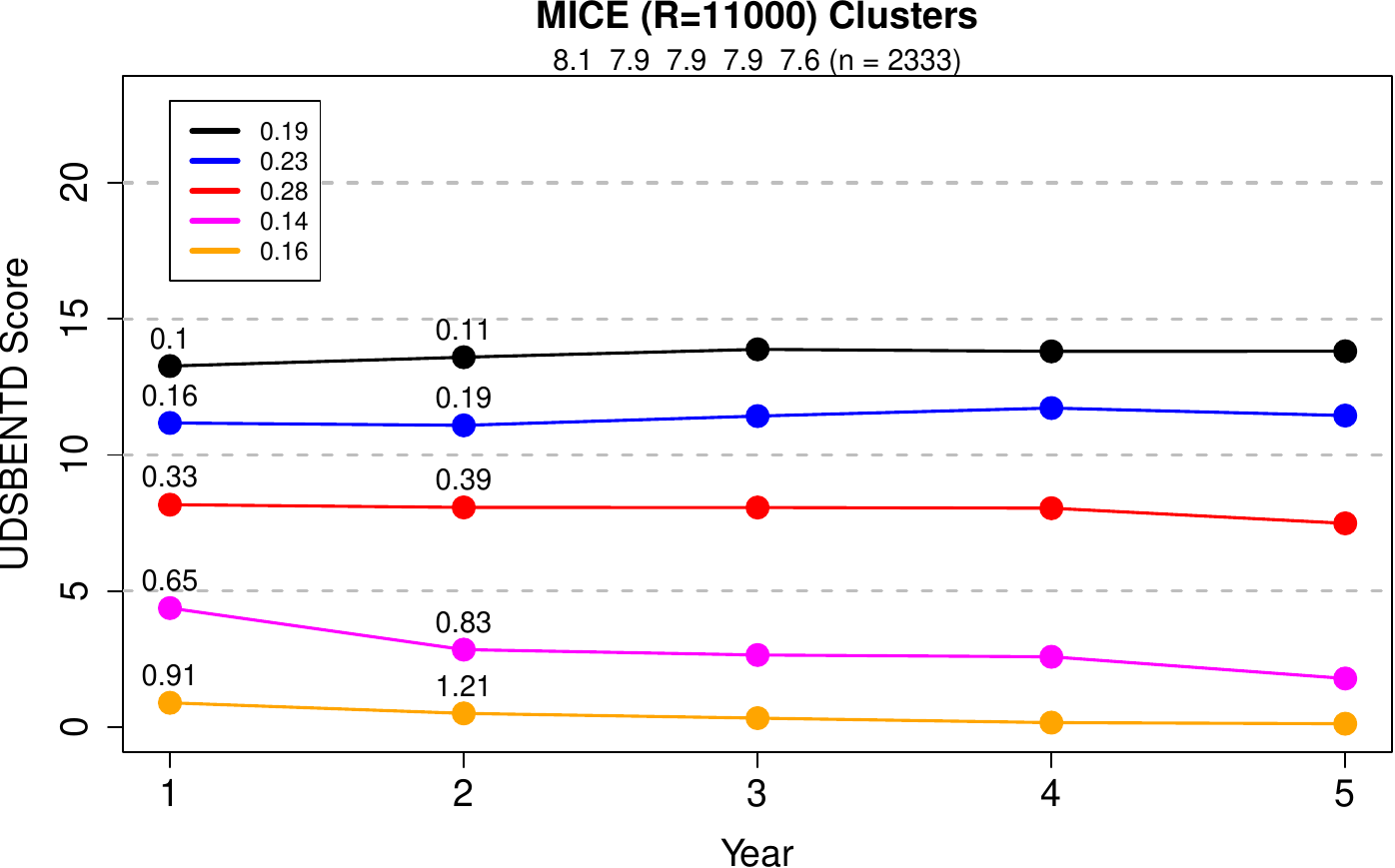}
\vskip 10pt
\includegraphics[width=0.47\textwidth]{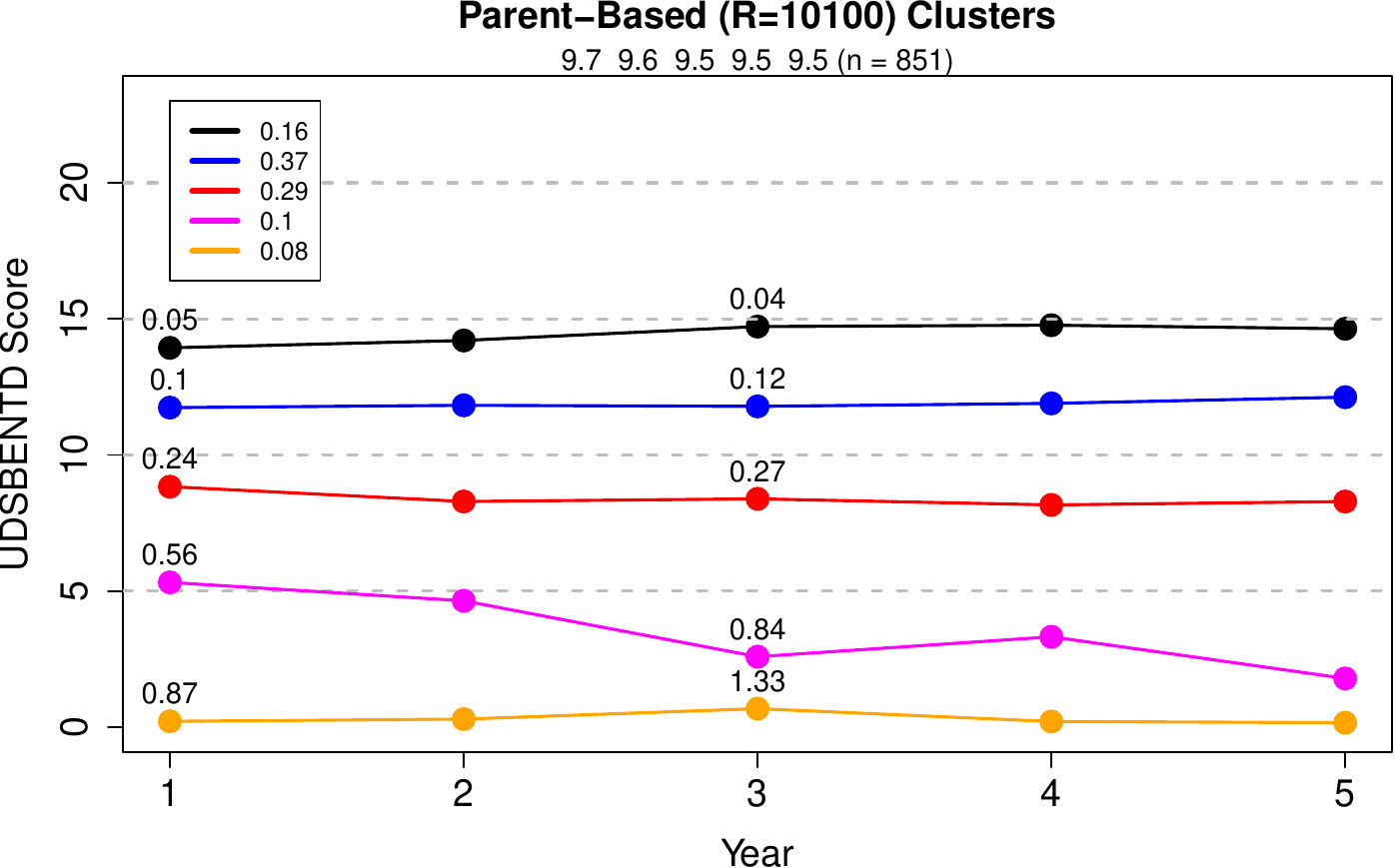}
\hspace{0.2in}
\includegraphics[width=0.47\textwidth]{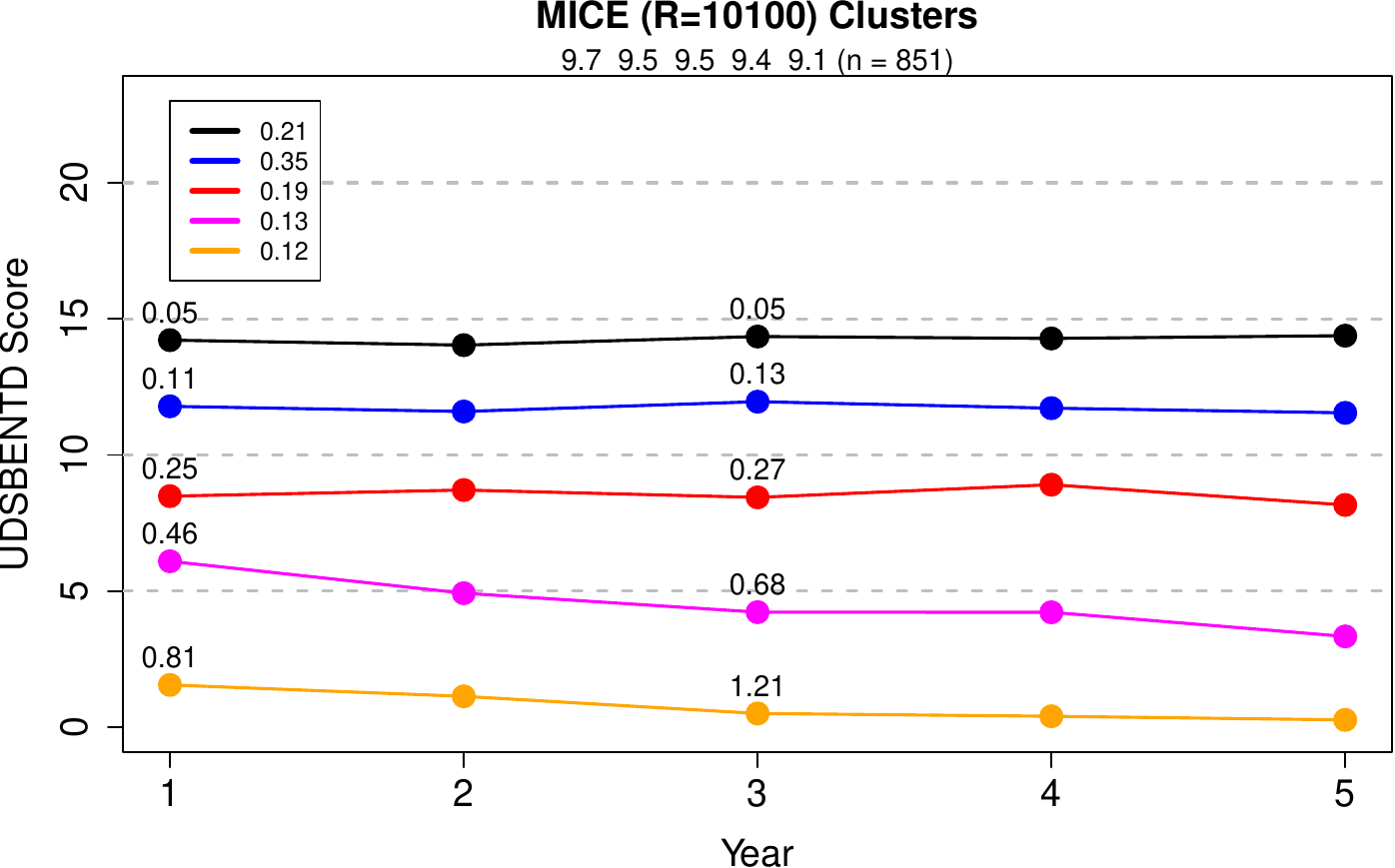}
\vskip 10pt
\includegraphics[width=0.47\textwidth]{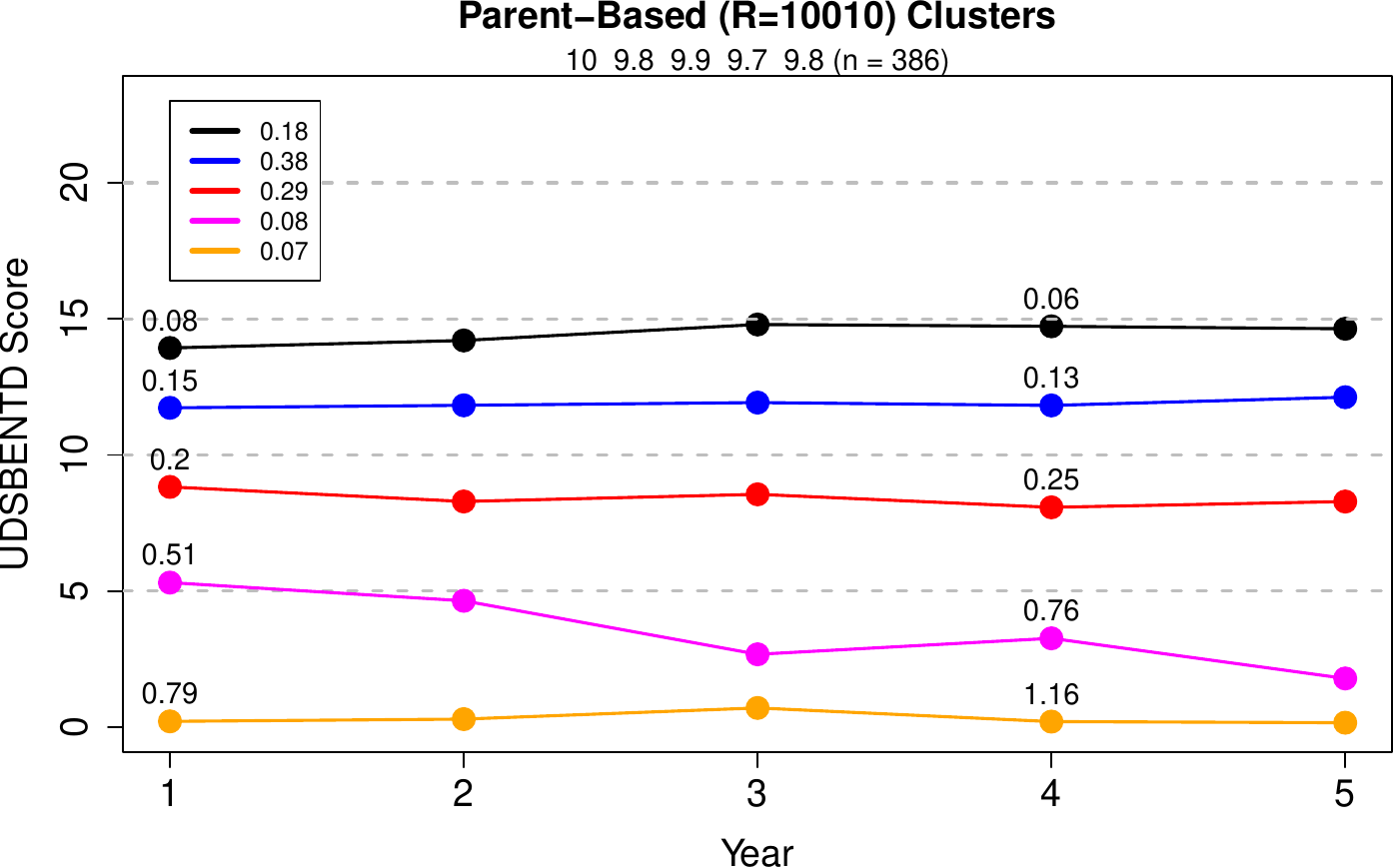}
\hspace{0.2in}
\includegraphics[width=0.47\textwidth]{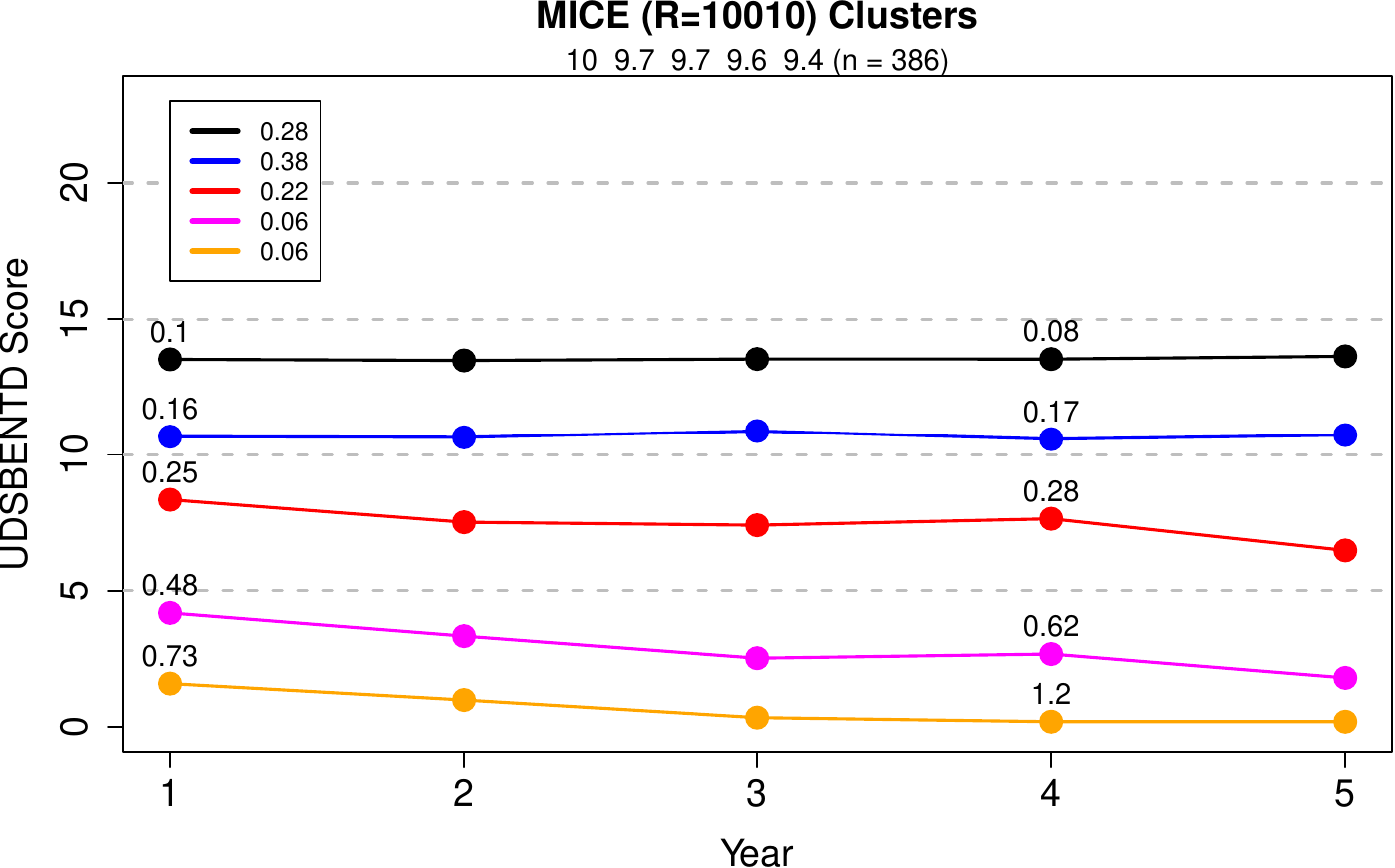}
\caption{These are the results of the fitted trajectories.  Each row corresponds to a pattern with $10000$, $11000$, $10100$, and $10010$, respectively.  Column 1 is the parent-based modeling tree graph assumption, and column 2 is fitting the model after \textit{mice} imputation.}
\label{fig:clusteringNACC}
\end{figure}

\section{Discussion}

In this paper, we introduced a new strategy for modeling multivariate missing not at random data.  This strategy combines two frameworks: 1) the tree graph framework for identifying the selection odds and 2) the conjugate odds property to ensure simple modeling.  We demonstrate that the tree graph is an incredibly rich subclass of the general pattern graph.  Each tree graph represents a missing not at random assumption and provides an elegant form of the selection odds, thereby overcoming a shortcoming of a general pattern graph.  Moreover, the conjugate odds property is introduced and used to model all distributions of the form $p(x|r)$.  We provide examples of the conjugate odds property with applications to mixtures of exponential family models.  Furthermore, we provide simulations to argue finite sample performance of our method, and we analyze two data sets comprising multivariate discrete and multivariate continuous data.

There are several ways to extend the ideas in this paper.  As presented, our framework works using mixture of exponential family models with logistic odds and mixture of Pareto distributions with power law odds.  Previous proposed models from the literature such as mixture of binomial products \citep{suen2023modelingmissingrandomneuropsychological} and the Rasch model \citep{rasch1960probabilistic} could be utilized here.  It would be interesting to explore other parametric families and determine what others might fall under this framework.  Since the data is longitudinal by nature, there may be a more sophisticated way to incorporate time in the $p(x|1_d)$ model.

Furthermore, while we discussed multiple methods for choosing a tree graph and performing sensitivity analysis, this remains an active area of research.  Since a tree graph is a nonparametric identifying restriction that cannot be rejected by the observed data, it is critical to choose it in such a way that is reasonable.  We have outlined a few different principles, but there may be more extensions.  For example, when performing a parent-based or child-based modeling approach, one could consider distributional distances such as the Wasserstein or Hellinger metrics.  A natural way to conduct sensitivity analysis is through exponential tilting \citep{kim2011semiparametric, shao2016semiparametric, zhao2017semiparametric}, but there may be more other methods that exploit the geometry of the pattern graph space to interpolate between different tree graph assumptions.  We leave this for future work.

\clearpage

\newpage

\appendix

\section{Empirical analysis: Kernel density estimation}

\label{appendix:kde}
We now consider a data in the continuous setting.  We consider white vinho verde wine samples from the north of Portugal. This data can be downloaded from the UCI repository.  This data consists of 4898 observations and was originally collected to model wine quality based on physicochemical tests.  We select three continuous variables to study the modeling effect: pH, sulphate, and alcohol levels.  Initially, we normalize the data such that it has a mean 0 and standard deviation 1.

\subsection{Missing Not at Random Mechanism}

\label{appendix:MNARkde}

First, we generate the missing data via a missing not at random mechanism 100 times through a tree graph and a prespecified selections odds model.  On each iteration, we consider four density estimators for each conditional distribution $p(x|\bR=r)$.  The first is a multivariate kernel density estimator using a Gaussian kernel on the complete-case data.  Then, we construct our tree graph KDE, where we exploit the conjugate odds property with the Gaussian kernel.  We also include an available case marginal Gaussian KDE estimator, where we fit the distribution based on all data that is observed for that dimension.  For example, if we are considering dimension $3$, then we pool the data from patterns $111$, $101$, and $001$, and fit a one-dimensional KDE $\hat{p}(x_3|\bR\in\{001,101,111\})$.  One clear disadvantage of the available case KDE is that we are unable to construct a joint KDE.  Additionally, when there is missingness, we also perform \textit{mice} imputation 20 times and construct the multivariate KDE on the \textit{mice} imputed data.

In Figure \ref{fig:realdata-wine}, for each pattern-dimension pair, we plot the marginalized KDEs averaged over all 100 iterations.  Of primary interest, we plot the tree graph KDE obtained after applying the conjugate odds property.  For comparison, we also plot the complete case KDE, the available case KDE, and when the data is missing, the mice KDE.  Since we have access to the true data and generate the missingness ourselves, we also can construct the oracle KDE, based on the true data.  Thus, we include the oracle KDE, which is the KDE fitted using the true data.  We expect the tree graph KDE to agree with the oracle KDE, and largely, we observe that the tree graph KDE is able to generally identify the same shape as the oracle KDE.  In contrast, the competing kernel density estimators generally do not capture the correct shape of the distribution, and it is clear they have different means and modes.  We note that imputing with \textit{mice} and fitting a KDE provides a similar result to the available case KDE, but it is not similar to the oracle KDE.  This provides further evidence of the need to be careful when applying \textit{mice}, especially when the data is MNAR.

\begin{figure}[!b]
	\centering
	{\color{blue}\fbox{\includegraphics[width=0.31\textwidth]{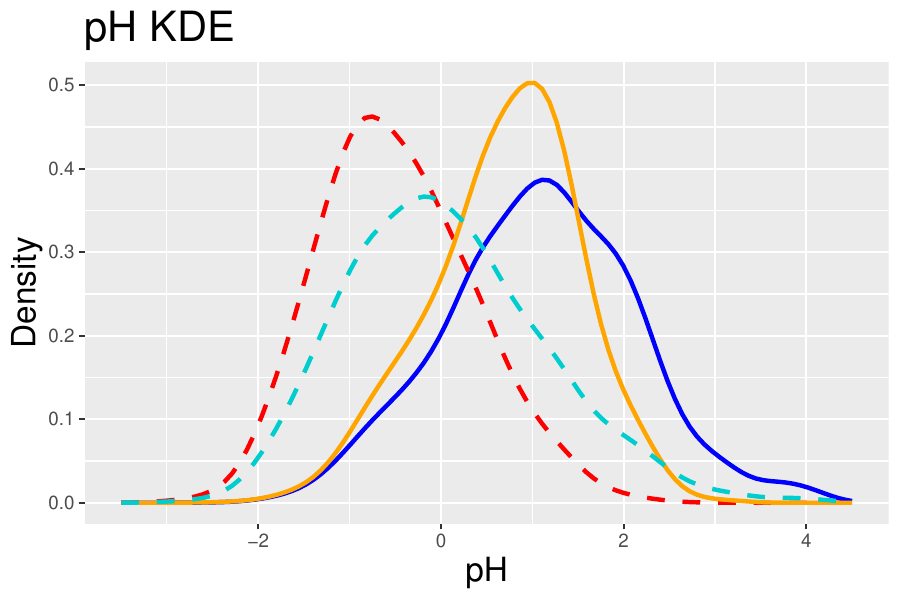}}}
	{\color{blue}\fbox{\includegraphics[width=0.31\textwidth]{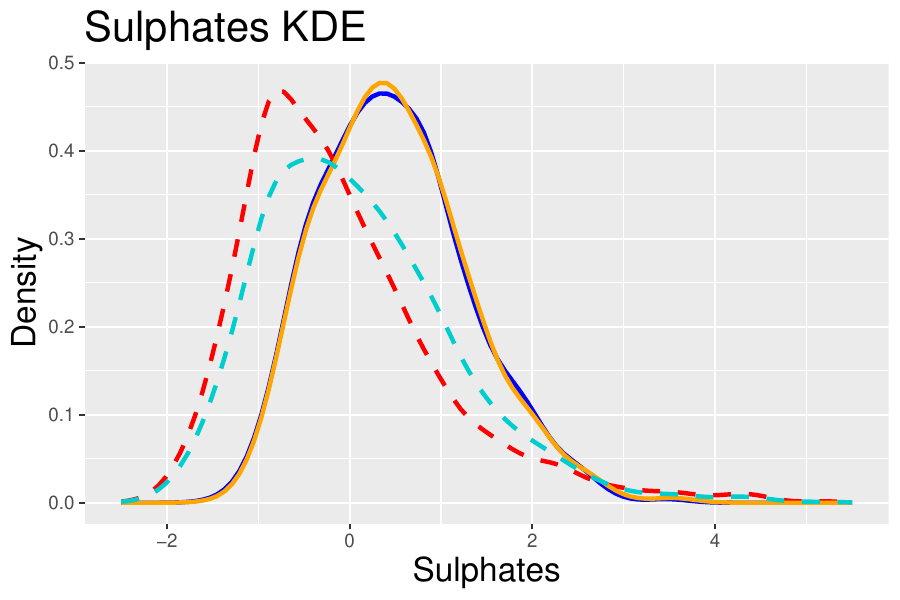}}}
	{\color{magenta}\fbox{\includegraphics[width=0.31\textwidth]{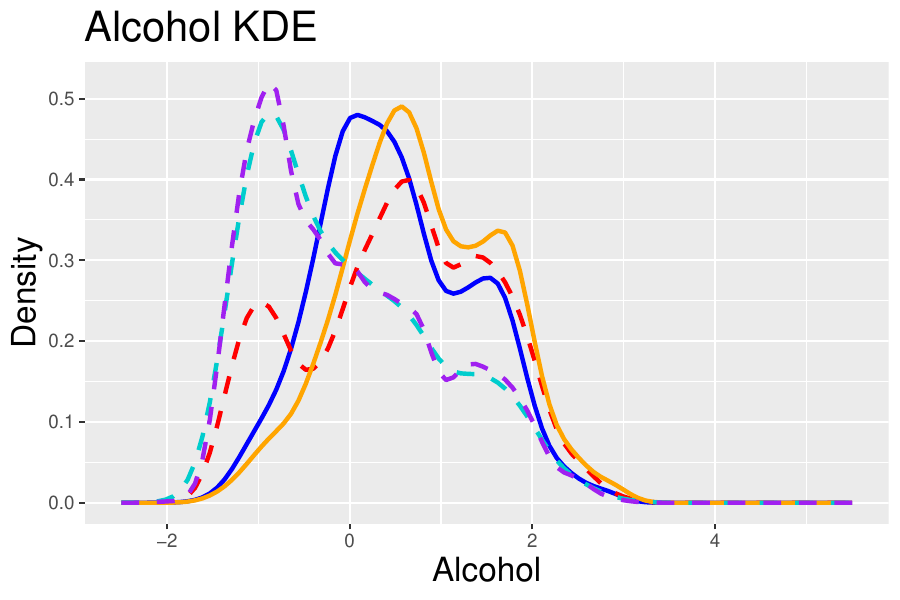}}}
	\vskip 5pt
	{\color{blue}\fbox{\includegraphics[width=0.31\textwidth]{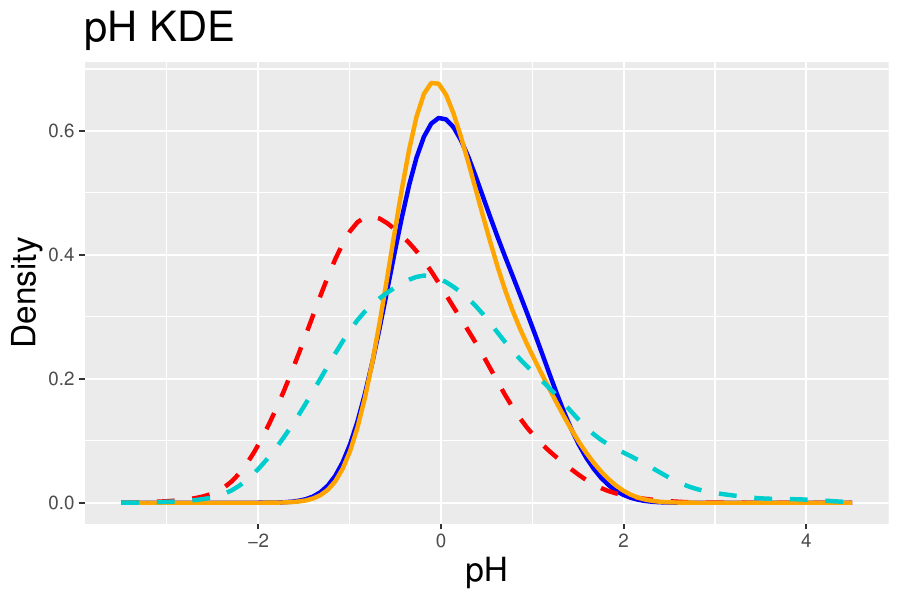}}}
	{\color{magenta}\fbox{\includegraphics[width=0.31\textwidth]{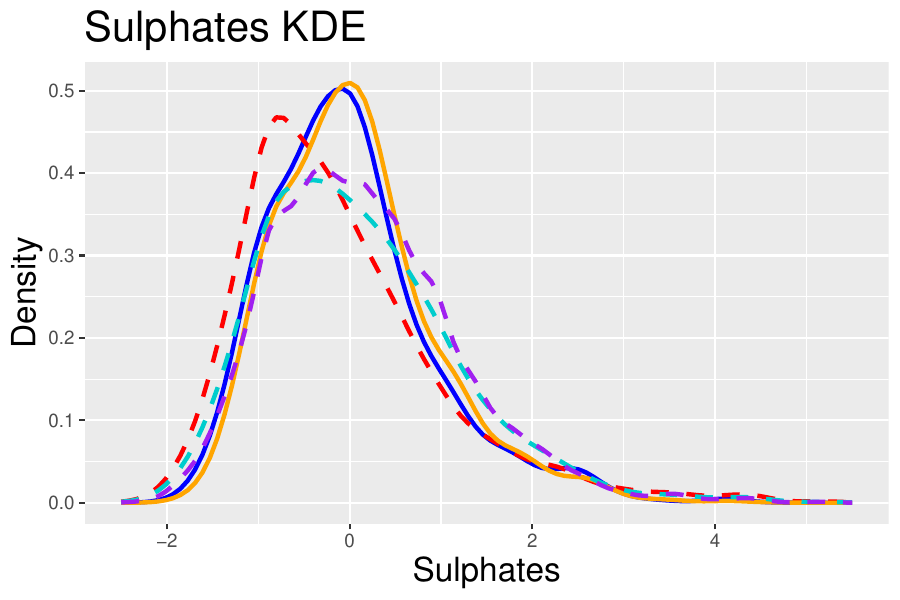}}}
	{\color{blue}\fbox{\includegraphics[width=0.31\textwidth]{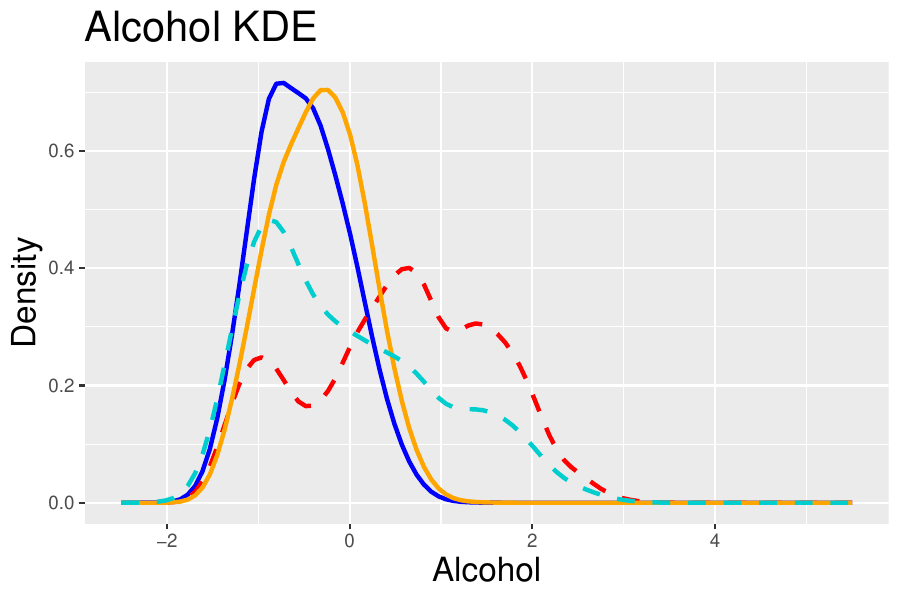}}}
	\vskip 5pt
	{\color{magenta}\fbox{\includegraphics[width=0.31\textwidth]{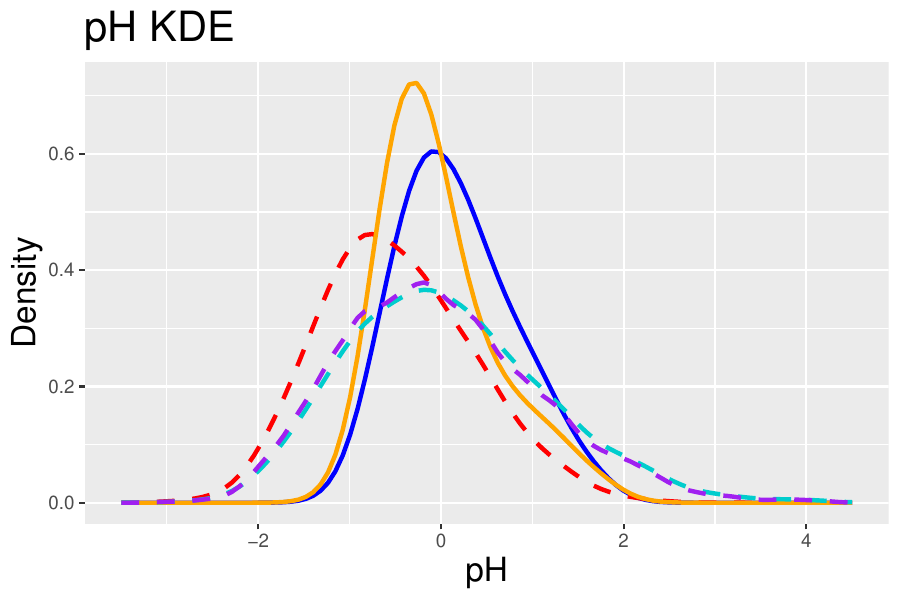}}}
	{\color{magenta}\fbox{\includegraphics[width=0.31\textwidth]{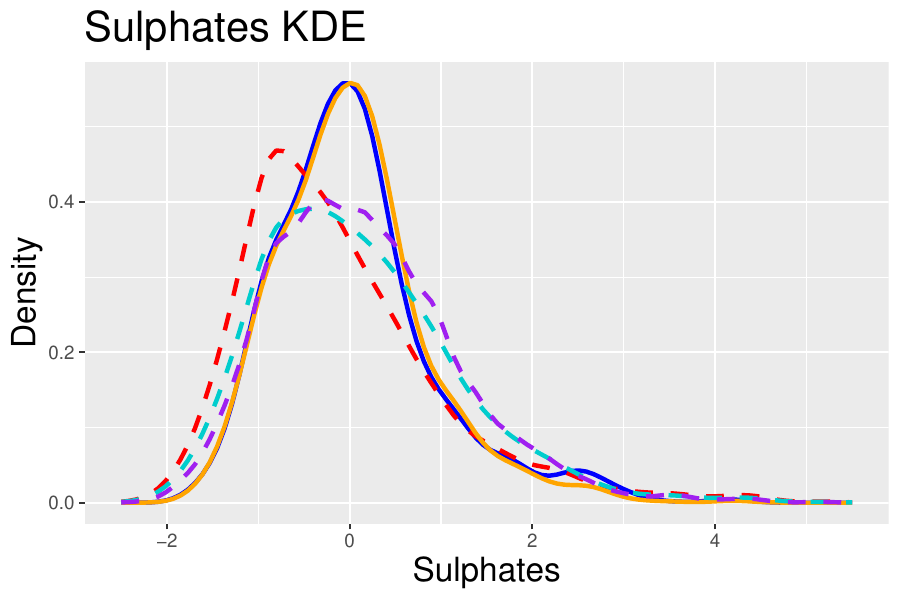}}}
	{\color{blue}\fbox{\includegraphics[width=0.31\textwidth]{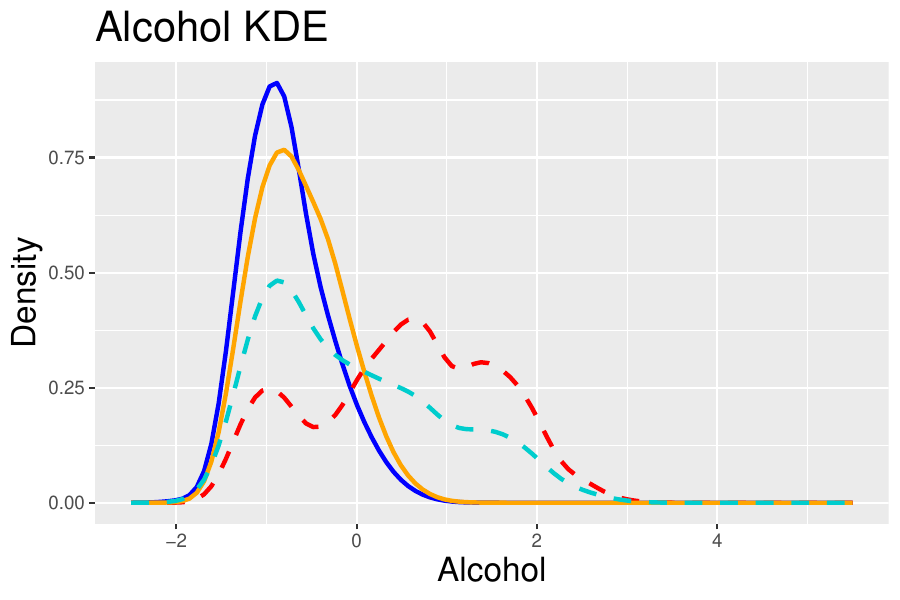}}}
	\vskip 5pt
	\includegraphics[width=0.25\textwidth]{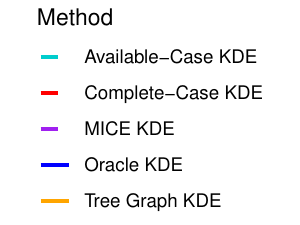}
	\hspace{0.5in}
	\includegraphics[width=0.15\textwidth]{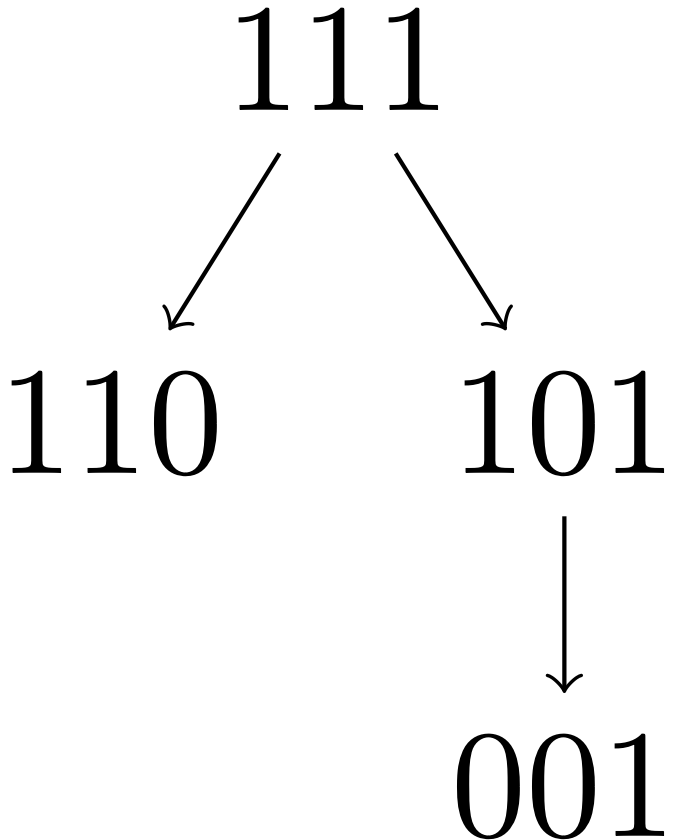}
	\caption{Each column corresponds to dimensions 1, 2, and 3, respectively.  Rows 1, 2, and 3 depicts the fitted KDEs for the pattern $110$, $101$, and $001$, respectively.  The last row includes the plot of the tree graph used to generate the missingness.}
	\label{fig:realdata-wine}
\end{figure}

As we mentioned before, our method can model the missing data distribution and the observed data distribution in one shot, so we report both results.  In the cases where we modeling the missing data distribution, we outline the plot in magenta.  In the cases where we model the observed data distribution, we outline the plot in blue.  Plots that are outlined in blue can be viewed more as diagnostic plots.

Each row of Figure \ref{fig:realdata-wine} corresponds to the marginal distributions for patterns $110$, $101$, and $001$, respectively.  For pattern $110$, the first and second plots correspond to marginal observed data distributions, and the third plot corresponds to a marginal missing data distribution.  For pattern $101$, the first and third plots correspond to marginal observed data distributions, and the second plot corresponds to a marginal missing data distribution.  For pattern $001$, the third plot corresponds to a marginal observed data distribution, and the first and second plots correspond to marginal missing data distributions, respectively.

\subsection{Missing at Random Mechanism}

\label{appendix:MARsim}

As in Appendix \ref{appendix:MNARkde}, we consider a simulation using the same real data.  However, we generate the missingness to be missing at random using the following logistic regression
\begin{align*}
	& \frac{P(\bR=110|x)}{P(R\neq110|x)} \propto \exp(0.6x_1 - 0.3x_2), \\
	& \frac{P(\bR=101|x)}{P(R\neq101|x)} \propto \exp(-0.6x_1 + 0.3x_3), \\
	& \frac{P(\bR=001|x)}{P(R\neq001|x)} \propto \exp(0.8x_3)
\end{align*}
with proportionality constants chosen such that $P(\bR=001) \approx 0.2$, $P(\bR=101) \approx 0.2$, $P(\bR=110) \approx 0.3$, and $P(\bR=111) \approx 0.3$.

For each pattern-dimension pair, we plot the marginalized KDEs averaged over all 100 iterations.  The provided KDEs are the same as those in Section \ref{appendix:kde}.  In this case, we have to learn a tree graph, so we run a data-driven parent-based approach, using energy distance on the empirical distributions.  For two distributions $P_X$ and $P_Y$, the energy distance can be written as
$$D_{\text{EN}}(P_X,P_Y) = 2\E\lVert X- Y \rVert - \E\lVert X- X'\rVert - \E\lVert Y - Y'\rVert$$
for $X,X'\sim P_X$ and $Y,Y'\sim P_Y$ and where $\lVert \cdot \rVert$ denotes the Euclidean norm.
We can estimate this using a sample version via
$$\widehat{D} = \frac{2}{n_X \cdot n_Y}\sum_{i,j} \lVert X_i - X_j\rVert - \frac{1}{n_X^2}\sum_{i,j}\lVert X_i - X_j\rVert - \frac{1}{n_Y^2}\sum_{i,j}\lVert Y_i - Y_j\rVert.$$

There are two possible tree graphs we can learn, and we provide a visualization of them in Figure \ref{fig:realdata-wine-MAR}.  Tree Graph 1 is the deepest possible graph, and Tree Graph 2 is the shallowest possible, corresponding to a CCMV assumption.  Tree Graph 1 was learned 100 times out of the total 100 randomly generated data sets, and Tree Graph 2 (CCMV) was never learned.  Therefore, we do not plot the results of fitting a CCMV graph and only of the first tree graph.  In Figure \ref{fig:realdata-wine-MAR}, we refer to the KDE from this learned tree graph as tree graph KDE.

We plot the tree graph KDE obtained after applying the conjugate odds property.  As before, we also plot the complete case KDE, the available case KDE, and when the data is missing, the mice KDE.  As we have accesss to the true data and generate the missingness ourselves, we also can construct the oracle KDE, based on the true data.  Therefore, we also include the oracle KDE, which is the KDE fitted using the true data. Surprisingly, the tree graph KDE and mice KDE generally agrees with the oracle KDE in most scenarios.  This suggests in some scenarios, there may be some robustness of our tree graph method to missingness generated via missing at random.  Additionally, since the tree graph KDE method is more computationally tractable than the mice KDE method, there may also be scenarios where it is preferred.

\begin{figure}[!b]
	\centering
	{\color{blue}\fbox{\includegraphics[width=0.31\textwidth]{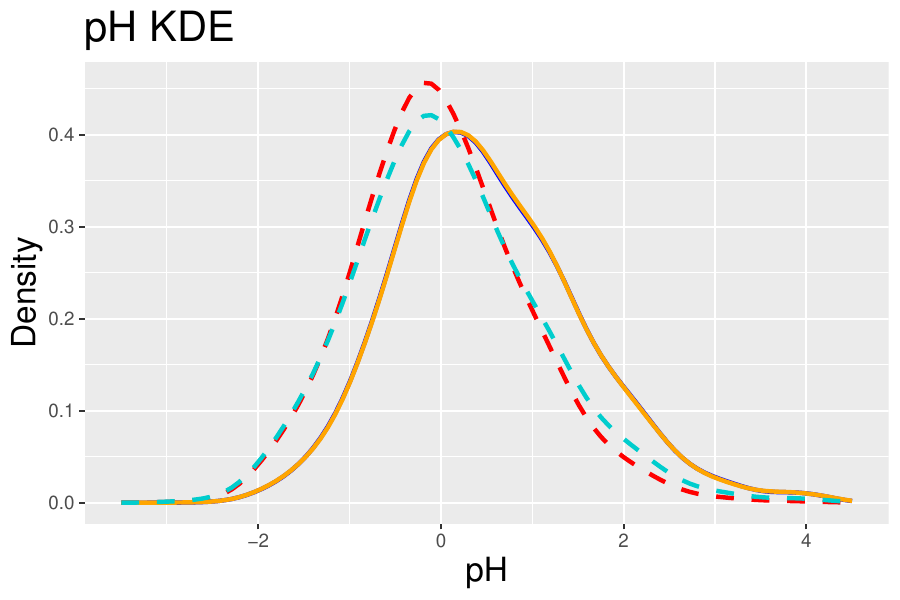}}}
	{\color{blue}\fbox{\includegraphics[width=0.31\textwidth]{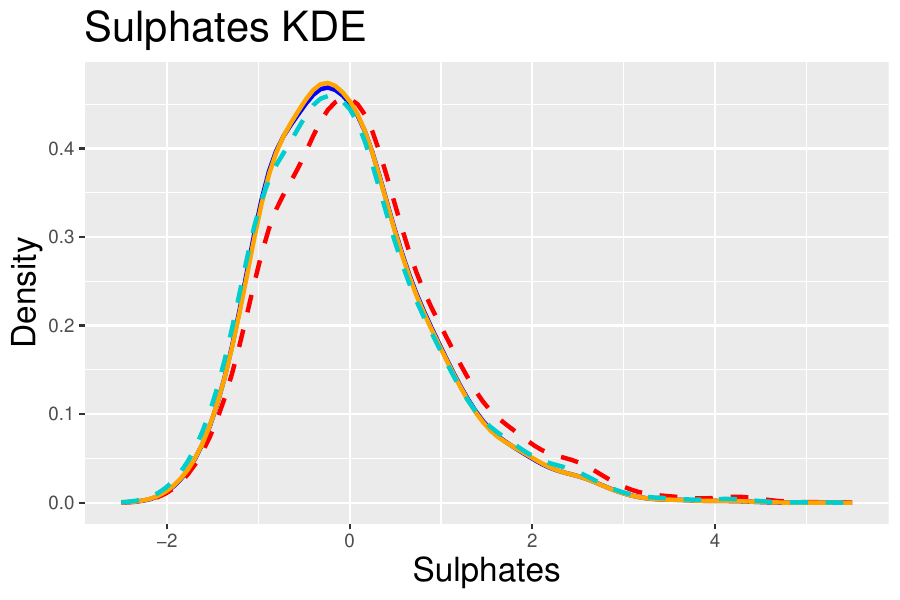}}}
	{\color{magenta}\fbox{\includegraphics[width=0.31\textwidth]{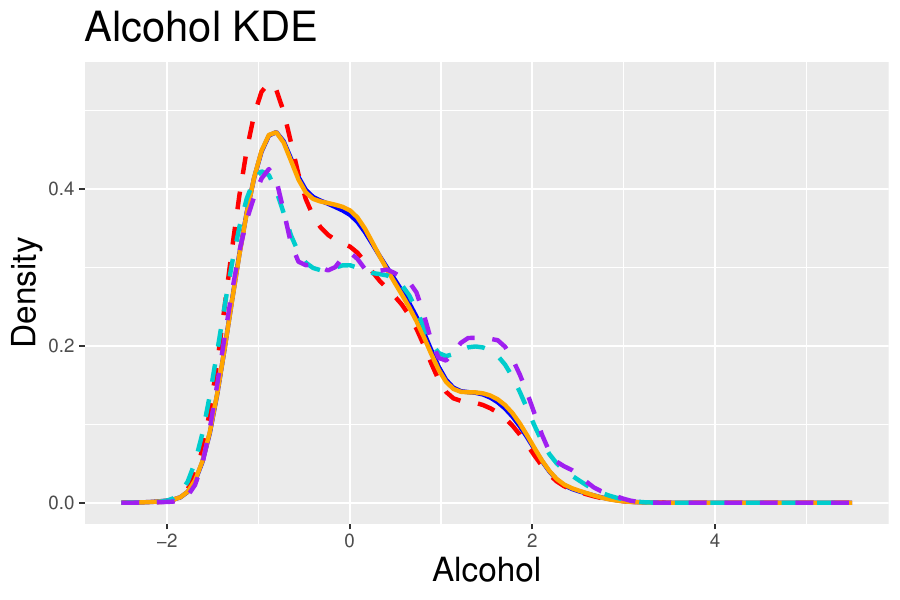}}}
	\vskip 5pt
	{\color{blue}\fbox{\includegraphics[width=0.31\textwidth]{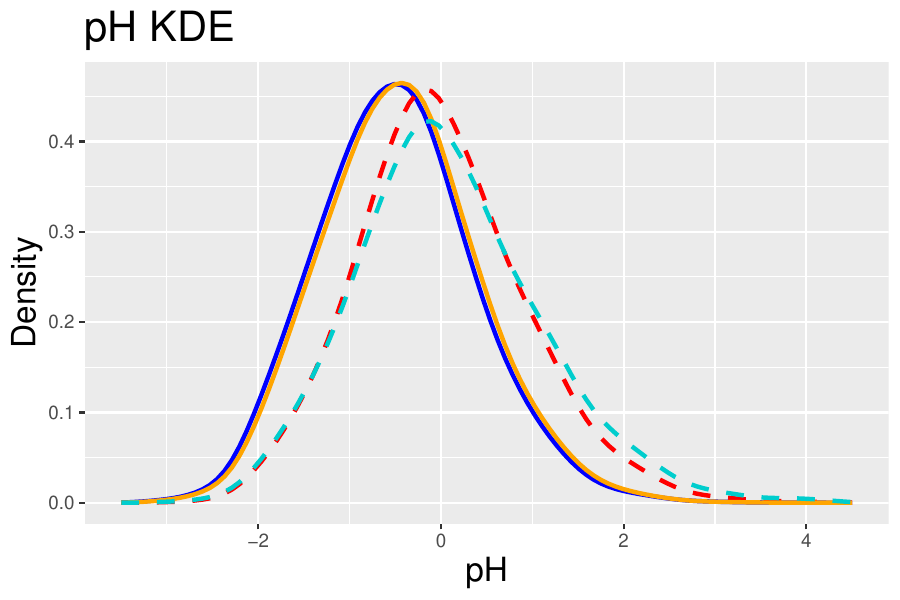}}}
	{\color{magenta}\fbox{\includegraphics[width=0.31\textwidth]{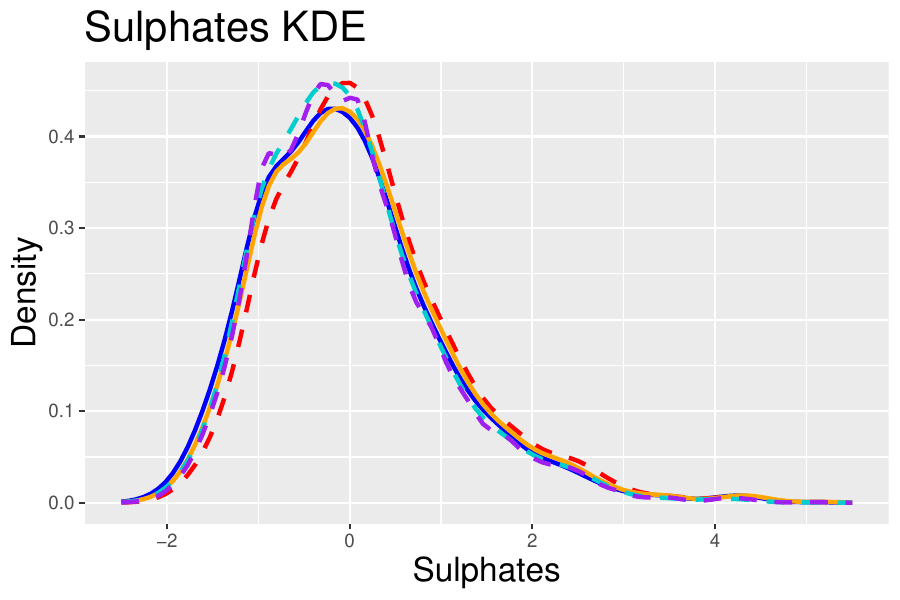}}}
	{\color{blue}\fbox{\includegraphics[width=0.31\textwidth]{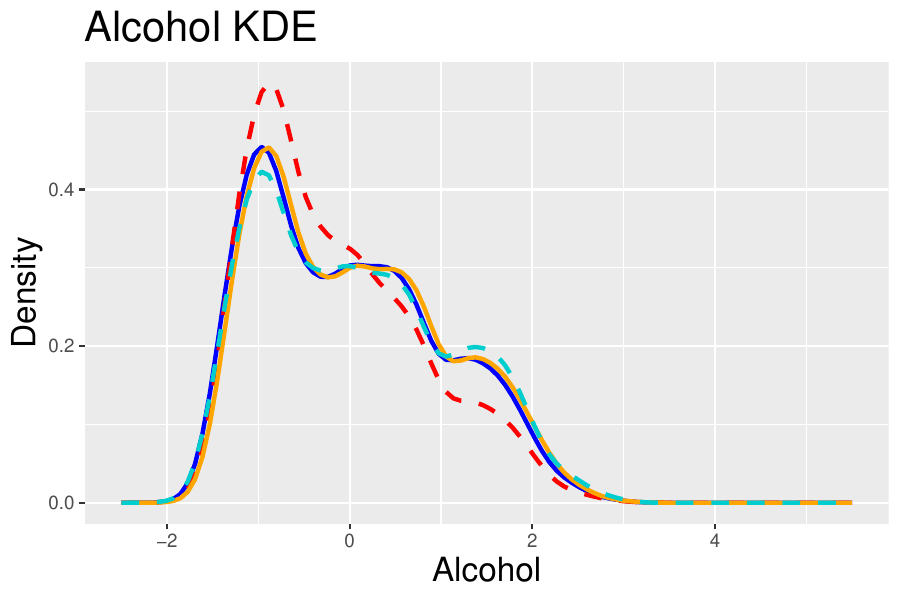}}}
	\vskip 5pt
	{\color{magenta}\fbox{\includegraphics[width=0.31\textwidth]{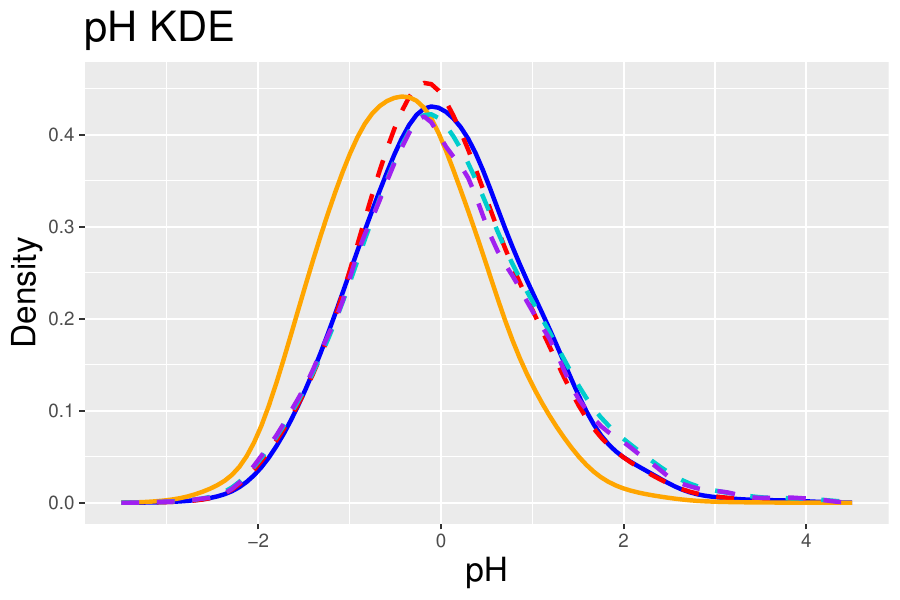}}}
	{\color{magenta}\fbox{\includegraphics[width=0.31\textwidth]{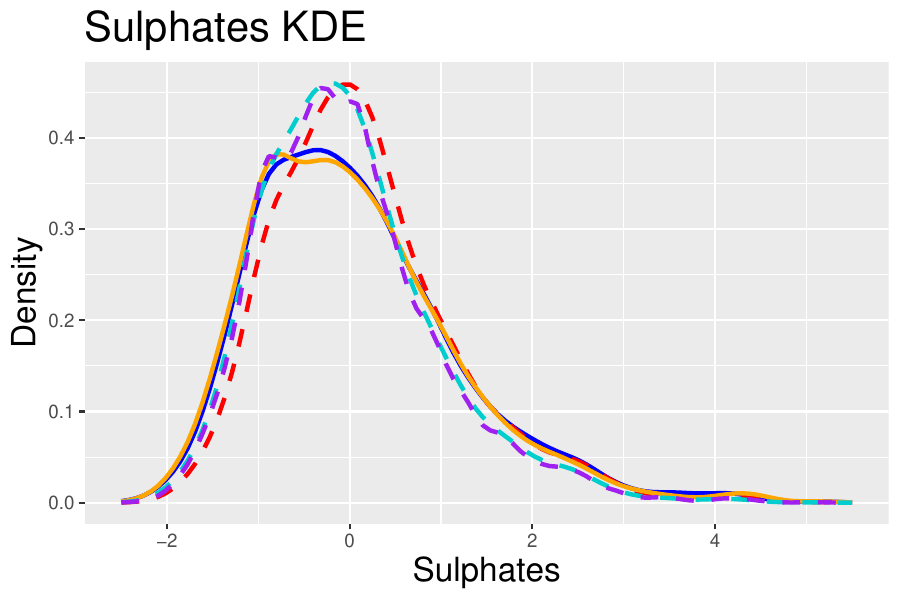}}}
	{\color{blue}\fbox{\includegraphics[width=0.31\textwidth]{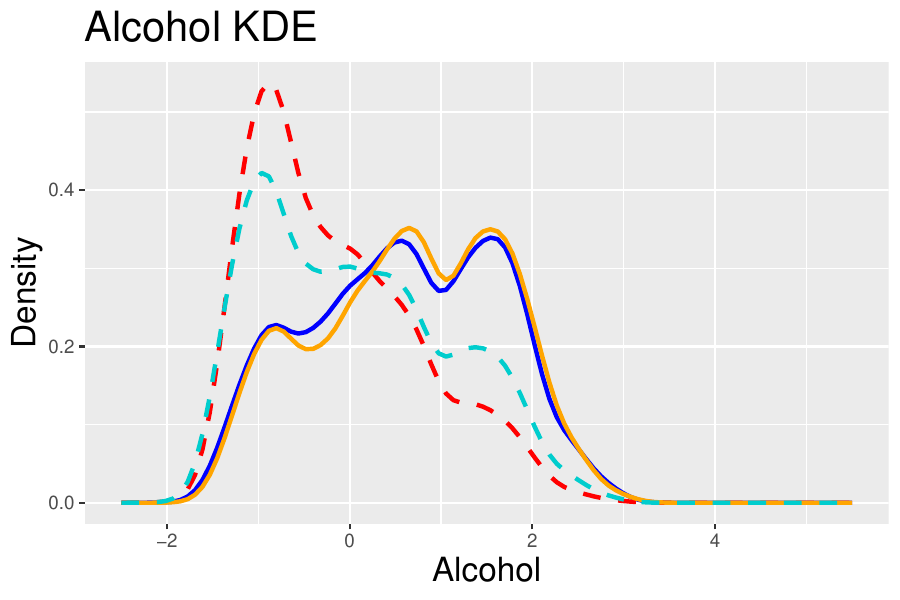}}}
	\vskip 5pt
	\includegraphics[width=0.25\textwidth]{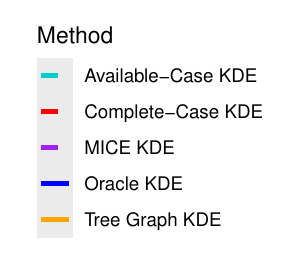}
	\hspace{0.5in}
	\includegraphics[width=0.15\textwidth]{figures/kde_treegraph.png}
	\hspace{0.5in}
	\includegraphics[width=0.3\textwidth]{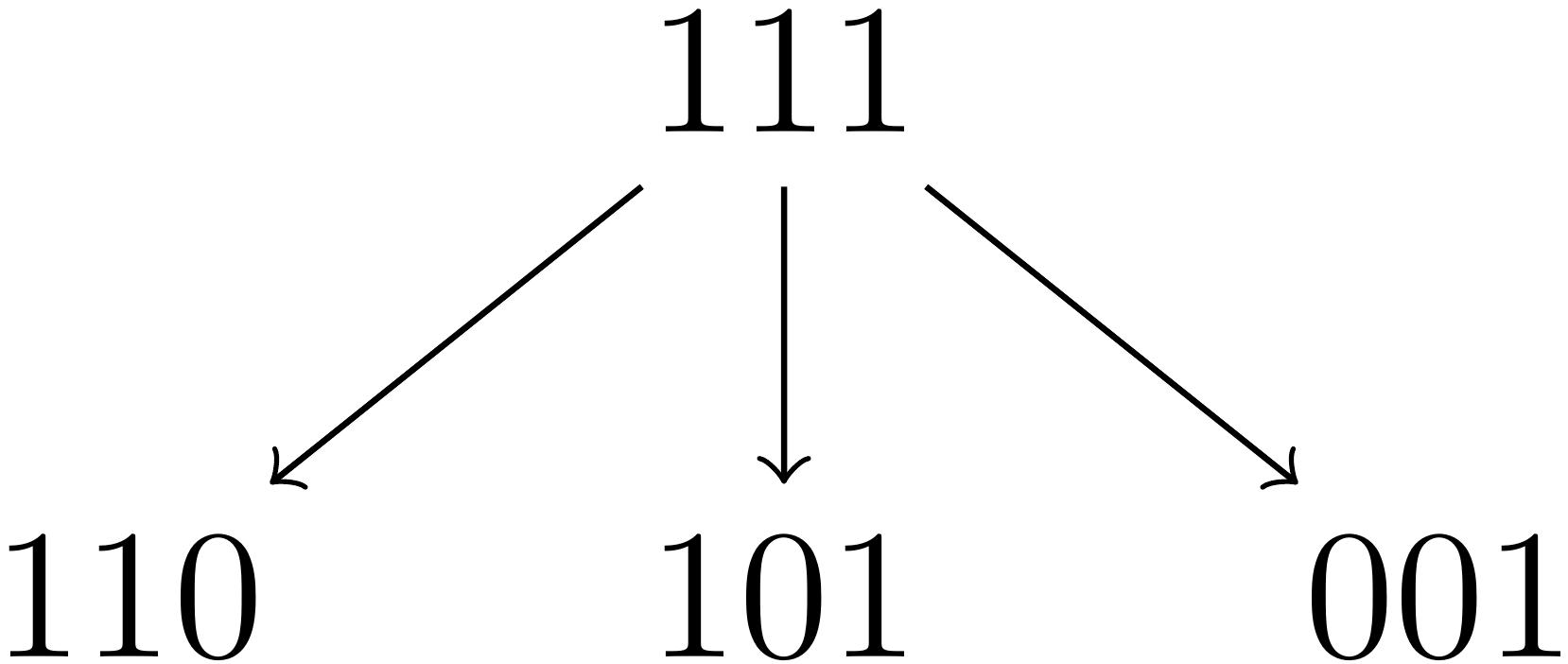}
	\caption{These are the fitted KDEs under a simulation where the data was generated MAR. Each column corresponds to dimensions 1, 2, and 3, respectively.  Rows 1, 2, and 3 depicts the fitted KDEs for the pattern $110$, $101$, and $001$, respectively.  The last row includes the plot of the two tree graphs that could have been learned from the data-driven algorithm.}
	\label{fig:realdata-wine-MAR}
\end{figure}

\section{Simulations}

\label{sect:simulations}

In our simulation study, we consider a setting with \( d = 3 \) bounded discrete variables. The data-generating process assumes that \( p(x | R = 1_d) \) follows a mixture of binomial product distributions, while the selection mechanism is modeled such that the selection odds \( P(R = r | x) / P(R = 1_d | x) \) follow a logistic regression model for all missing data patterns \( r \).  Under correct specification of the true tree graph, we assess consistency and coverage using an empirical bootstrap procedure. 

For a given tree graph, the simulation procedure consists of the following steps:  

\begin{itemize}
	\item \textbf{Data Generation:} We specify \( p(x \mid R = 1_d) \) as a mixture of binomial products, set the probabilities \( P(R = r) \) for each missing data pattern, and specify the logistic regression coefficients \( \gamma_r \) and intercepts \( \gamma_{0,r} \). This setup ensures that each conditional distribution \( p(x \mid R = r) \) remains a mixture of binomial products.  
	
	
	We generate the data according to the following parameters:
	\begin{center}
		\begin{tabular}{c|cccccc}
			$R$        & 111 & 110 & 101 & 100 & 010 & 001 \\ \hline
			$P(R)$   & 0.3 & 0.2 & 0.1 & 0.15 & 0.15 & 0.1
		\end{tabular}
	\end{center}
	\begin{align*}
		p(x \mid R = \mathbf{1}_d; \bm{w}_{cc}, \bm{\theta}_{cc}) &\quad\text{with} \quad
		\bm{w}_{cc} = \begin{bmatrix} 0.3 & 0.5 & 0.2 \end{bmatrix}, \quad
		\bm{\theta}_{cc} = 
		\begin{bmatrix}
			0.70 & 0.75 & 0.70 \\
			0.50 & 0.50 & 0.40 \\
			0.20 & 0.30 & 0.10
		\end{bmatrix}
	\end{align*}
	
	The missingness mechanism is modeled using logistic regressions:
	\begin{align*}
		&\frac{P(R = 110 \mid x)}{P(R = 111 \mid x)} \propto \exp(0.1 x_1 + 0.1 x_2), \qquad \frac{P(R = 101 \mid x)}{P(R = 111 \mid x)} \propto \exp(0.3 x_1 + 0.1 x_3), \\
		&\frac{P(R = 100 \mid x)}{P(R = 101 \mid x)} \propto \exp(-0.1 x_1), \qquad
		\frac{P(R = 010 \mid x)}{P(R = 110 \mid x)} \propto \exp(0.1 x_2), \\
		&\frac{P(R = 001 \mid x)}{P(R = 101 \mid x)} \propto \exp(0.1 x_3).
	\end{align*}

	\item \textbf{Consistency Assessment:} For each given sample size, we generate $U=200$ random data sets based on the data generating process.  We estimate model parameters using an expectation-maximization (EM) algorithm for \( p(x \mid R = 1_d) \) and standard logistic regression for the selection mechanism.  Then, we report the estimated MSE over all $U=200$ point estimates.
	\item \textbf{Coverage Evaluation:} We report the estimated coverage for our bootstrap approach over all $U=200$ random data sets using Algorithm \ref{alg:empbootstrap} and $B=500$ bootstrap samples.  Confidence intervals are constructed by using estimating the standard errors with the bootstrap estimates and then, adding and subtracting them from the point estimate.
\end{itemize}

\begin{table}[!b]
	\label{table:sim}
	\centering
	\vskip 10pt
	\scalebox{0.9}{
		\begin{tabular}{lccccc|}
			\multicolumn{6}{c}{$\textbf{Mean Squared Error}~(\times 100)$} \\
			\cline{2-6}
			\multicolumn{1}{c|}{} & $n=500$ & $n=1000$ & $n=2000$ & $n=5000$ & $n=10000$ \\
			\hline
			\multicolumn{1}{|l|}{$\btheta_{cc}$} & 0.036 & 0.020 & 0.0092 & 0.0036 & 0.0019 \\
			\multicolumn{1}{|l|}{$\bm{w}_{cc}$} & 0.17 & 0.09 & 0.041 & 0.016 & 0.0073\\
			\multicolumn{1}{|l|}{$\bbeta$} & 0.29 & 0.14 & 0.07 & 0.026 & 0.013 \\
		\end{tabular}
	}
	\scalebox{0.9}{
		\begin{tabular}{lccccc|}
			\multicolumn{6}{c}{$\textbf{Estimated Coverage}$} \\
			\cline{2-6}
			\multicolumn{1}{c|}{} & $n=500$ & $n=1000$ & $n=2000$ & $n=5000$ & $n=10000$ \\
			\hline
			\multicolumn{1}{|l|}{$\btheta_{cc}$} & 0.94 & 0.94 & 0.94 & 0.95 & 0.94 \\
			\multicolumn{1}{|l|}{$\bm{w}_{cc}$} & 0.93 & 0.94 & 0.93 & 0.95 & 0.95 \\
			\multicolumn{1}{|l|}{$\bbeta$} & 0.98 & 0.98 & 0.98 & 0.97 & 0.98 \\
		\end{tabular}
		
	}
	\caption{These results are the estimated MSEs and estimated coverage after fitting a mixture model and running Algorithm \ref{alg:empbootstrap} for $U=200$ replicates.}
\end{table}

\begin{table}[!b]
	\label{table:datadriven}
	\centering
	\vskip 10pt
	\scalebox{0.9}{
		\begin{tabular}{lcccc|}
			\multicolumn{5}{c}{\textbf{Parent-Based Modeling}} \\
			\cline{2-5}
			\multicolumn{1}{c|}{} & $n=500$ & $n=1000$ & $n=2000$ & $n=5000$ \\
			\hline
			\multicolumn{1}{|l|}{Tree Graph 1} & 127 & 140 & 160 & 185  \\
			\multicolumn{1}{|l|}{Tree Graph 2} & 73 & 60 & 40 & 15
		\end{tabular}
	}
	\scalebox{0.9}{
		\begin{tabular}{lcccc|}
			\multicolumn{5}{c}{\textbf{Child-Based Modeling}} \\
			\cline{2-5}
			\multicolumn{1}{c|}{} & $n=500$ & $n=1000$ & $n=2000$ & $n=5000$ \\
			\hline
			\multicolumn{1}{|l|}{Tree Graph 1} & 183 & 194 & 198 & 200  \\
			\multicolumn{1}{|l|}{Tree Graph 2} & 15 & 6 & 2 & 0 \\
			\multicolumn{1}{|l|}{Tree Graph 3} & 2 & 0 & 0 & 0 \\
		\end{tabular}
		
	}
	\caption{These results are of the data-driven methods over $U=200$ random data sets.}
\end{table}

\begin{figure}[!h]
	\label{fig:sim-treegraphs}
	\centering
	\includegraphics[width=0.95\textwidth]{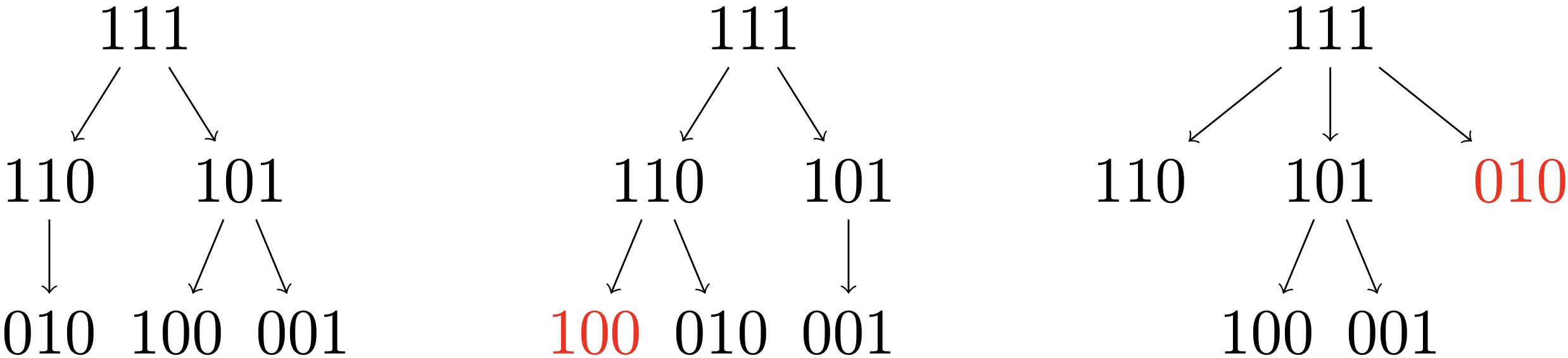}
	\caption{The left panel contains Tree Graph 1 (the tree graph used to generate the missing data).  The middle and right panels contain Tree Graphs 2 and 3, respectively, and they correspond to graphs incorrectly learned by the parent-based and child-based modeling approaches.  The nodes in the middle and right panels that are highlighted red indicate which ones were assigned to an incorrect parent.}
\end{figure}

In Table \ref{table:sim}, we generally see that the estimated MSE decreases at a linear rate, indicating that we have consistent performance.  Since our estimator is the MLE, it is also asymptotically efficient.  We also see that the coverage is roughly nominal and that the approach outlined in Algorithm \ref{alg:empbootstrap} works well.  For each data set, we also learn a tree graph using the parent-based and child-based modeling approaches, reporting the learned tree graphs in Figure \ref{fig:sim-treegraphs} and their frequencies in Table \ref{table:datadriven}.  From here, we can see that both data-driven methods generally select the correct tree graph with high frequency in high enough sample size.


\section{Power law odds}	\label{sec::power}

We provide an additional family of examples through a power law family.  When the odds can be modeled using a power law family, we can expand the family of distributions that we have conjugate odds for.  Modeling the odds using a power function is a nontraditional method, but it is similar to logistic regression in that it can be interpreted as a linear classifier with a more gradual boundary.

\begin{proposition}[Power function family, Pareto distribution]
	\label{prop:power-pareto}
	Suppose that $p(x|A=a)$ is a Pareto distribution
	$$p(x|A=a; \alpha, \beta) = \begin{cases} \dfrac{\alpha \beta^\alpha}{x^{\alpha+1}} & x \geq \beta \\ 0 & \text{o.w.} \end{cases}.$$
	Then, the associated odds model
	$$O_{a'}(x;\bm{\gamma}) := \frac{P(A=a'|x)}{P(A=a|x)} = \gamma_0x^\gamma, \quad \gamma_0 := \frac{P(A=a')}{P(A=a)} \cdot \frac{\alpha'}{\alpha} \cdot \beta^{\alpha'-\alpha}$$
	holds if and only if
	$$p(x|A=a'; \alpha', \beta) = \begin{cases} \dfrac{\alpha' {\beta}^{\alpha'}}{x^{\alpha'+1}} & x \geq \beta \\ 0 & \text{o.w.} \end{cases},$$
	where $\alpha' = \alpha - \gamma  \in \mathbb{R}^+$.
\end{proposition}

\begin{remark}
	We can also consider an odds model of the form
	$$\frac{P(A=0|x)}{P(A=1|x)} = \beta_0 + \beta_1 x^{\beta_2},$$
	which can be viewed as a weaker form of logistic regression.  They share similar properties in that the odds are always nonnegative for $x>0$.
	
	If the odds model is generalized to a sum of $K$ terms, then the resulting distribution $p(x|A=a')$ will be a mixture of Pareto distributions with the same shape parameter $\beta$.  Since the original distribution $p(x|A=a)$ has support in the positive reals, fitting the odds model with a polynomial can be done, provided the polynomial is strictly nonnegative.  We provide further examples in Appendix \ref{appendix:examples}.
\end{remark}

%



\section{Random sampling of tree graphs and connections to model averaging}	\label{sec::random}

A tree graph can also be generated randomly from the set $\mathcal{T}$.  First, we present Algorithm \ref{alg:uniformsample}, where we show how to sample a tree graph uniformly from $\mathcal{T}$.  We can randomly sample from the distribution of parents for each pattern $r\neq 1_d$.  By considering the set of possible parents for each pattern $r$ and choosing one uniformly at random, one can form a tree graph.  Every tree graph in $\mathcal{T}$ will be equally likely to be selected.

If one performs this sampling and constructs the corresponding point estimator many times, the set of point estimators may be averaged to form a final estimate.  We can view this as a form of model averaging.

\begin{algorithm}
	{\bf Require:} A set of missing patterns $\mathcal{R}$
	\begin{algorithmic}[1]
		\For{$r\in \mathcal{R}$}
		\If{$r\neq 1_d$}
		\State{Define $\text{PPA}_r := \{s : s > r, s\in\mathcal{R}\}$ as the set of potential parents of pattern $r$.}
		\State{Uniformly sample $s_r\sim \text{PPA}_r$.}
		\State{Form the parent set of pattern $r$ for graph ${T}$ as follows $\text{PA}_{T}(r) = \{s_r\}$.}
		\EndIf
		\EndFor
		\State{\textbf{return} Tree graph ${T}$}
	\end{algorithmic}
	\caption{Sampling a tree graph uniformly at random} \label{alg:uniformsample}
\end{algorithm}

We can extend this algorithm to randomly sample from an arbitrary distribution by combining Algorithm \ref{alg:uniformsample} with a rejection sampling scheme.  We present Algorithm \ref{alg:randomsample}, which serves as a minor modification to Algorithm \ref{alg:uniformsample} by introducing an acceptance criterion but generalizes the sampling to arbitrary distributions over $\mathcal{T}$.

\begin{algorithm}
	{\bf Require:} A set of missing patterns $\mathcal{R}$ and a distribution $p(t)$ over $\mathcal{T}$
	\begin{algorithmic}[1]
		\State{AcceptFlag = 0}
		\While{AcceptFlag = 0}
		\State{Sample $T$ uniformly from the space of all tree graphs using Algorithm \ref{alg:uniformsample}.}
		\State{Accept $T$ with probability $p(T)/\max_t p(t)$.}
		\If{Accepted}
		\State{AcceptFlag = 1}
		\EndIf
		\EndWhile
		\State{\textbf{return} Tree graph $T$}
	\end{algorithmic}
	\caption{Sampling a tree graph from an arbitrary PMF $p(t)$} \label{alg:randomsample}
\end{algorithm}

\begin{remark}[Bayesian and frequentist perspective]
	Algorithms \ref{alg:uniformsample} and \ref{alg:randomsample} allow a data analyst to place a prior on the set of tree graphs and combine them into a single point estimate.  However, we emphasize that while this somewhat mimics a Bayesian approach, this is not a Bayesian method because a prior is not place on the parameters and the final result is not a distribution.  The output remains a point estimate, thereby exhibiting frequentist properties.
\end{remark}


\section{Inference}	\label{sec::inference}


In this section, we describe a procedure for constructing confidence intervals.  Recall that in the tree graph and conjugate odds framework, we fit two types of models: a complete case model $p(x|1_d)$ and conjugate odds models $O_r(x_r) := P(\bR=r|x_r)/P(\bR\in\text{PA}_T(r)|x_r)$ for every $r\neq 1_d$.  This implicitly models the full-data distribution $p(x,r)$, which thereby implies specific forms for the distributions $p(x|r)$.  While constructing confidence intervals for the parameters of $p(x|1_d)$ and $O_r(x)$ is fairly straightforward, it is more challenging to construct confidence intervals for $p(x|\bR=r)$ for an arbitrary $r$.  This is because such intervals require accounting for the full joint sampling distribution of the parameters, incorporating joint uncertainty across both model components.  In the following subsection, we also describe an empirical bootstrap approaches to quantify uncertainty.




\begin{definition}[Primary model]
	In the tree graph and conjugate odds setting, we have two types of models: a complete case model $p(x|1_d)$ and an odds model $O_r(x_r)$.  We use the term primary model for pattern $r$ to refer to the model that corresponds to the pattern $r$.  If $r=1_d$, then the primary model is the complete case model $p(x|1_d)$.  Otherwise, it is the odds model $O_r(x_r)$.
\end{definition}

Each of the primary models described above is fit using observed data from at most two patterns.  This suggests that certain MLE parameters for the complete case model $p(x|1_d)$ and the odds models $O_r(x_r)$ may be independent.  We formalize that in Proposition \ref{prop:MLEindep}.

\begin{proposition}[Independence of certain MLE parameters]
	\label{prop:MLEindep}
	Let $\beta_r$ and $\beta_s$ be the parameters associated with the primary models for distinct patterns $r$ and $s$.  Suppose neither of these conditions hold:
	\begin{enumerate}
		\item One pattern is the parent of the other.
		\item The two patterns are siblings.
	\end{enumerate}
	Then, the MLE estimators $\hat{\beta}_r$ and $\hat{\beta}_s$ are independent.
\end{proposition}

Based on the results in Proposition \ref{prop:MLEindep}, we can specify exactly the form of the asymptotic covariance matrix through an undirected graph.  We describe the idea in Corollary \ref{cor:blockstructure} and Example  \ref{example:asymptoticcovariance} describes how this can be applied.

\begin{corollary}[Block structure of the asymptotic covariance matrix]
	\label{cor:blockstructure}
	In a tree graph $T$, convert each edge to an undirected edge, and add an undirected edge between every pair of siblings.  (This is similar to the idea of moralizing a directed graph except we connect the siblings rather than the parents.)  Call the resulting undirected graph $U$.  Then, the maximal cliques of $U$ exactly determine the block structure of the asymptotic covariance matrix.
\end{corollary}

\tikzset{c/.style={fill=gray}}
\tikzset{mystyle/.style={matrix of nodes,
		nodes in empty cells,
		row sep=-\pgflinewidth,
		column sep=-\pgflinewidth,
		nodes={draw,minimum width=0.85cm,minimum height=0.5cm,anchor=center}}}

\begin{example}[L-NCMV for 3 variables and its asymptotic covariance structure]
	\label{example:asymptoticcovariance}
	In this example, we consider the L-NCMV tree graph for 3 variables.  The results are recorded in Figure \ref{fig:exampleLNCMV}.  We obtain 4 maximal cliques: $\{111, 110, 101, 011\}$, $\{110, 100, 010\}$, $\{101, 001\}$, and $\{100, 000\}$.  
	
	\begin{figure}
		\label{fig:exampleLNCMV}
		\includegraphics[width=0.95\textwidth]{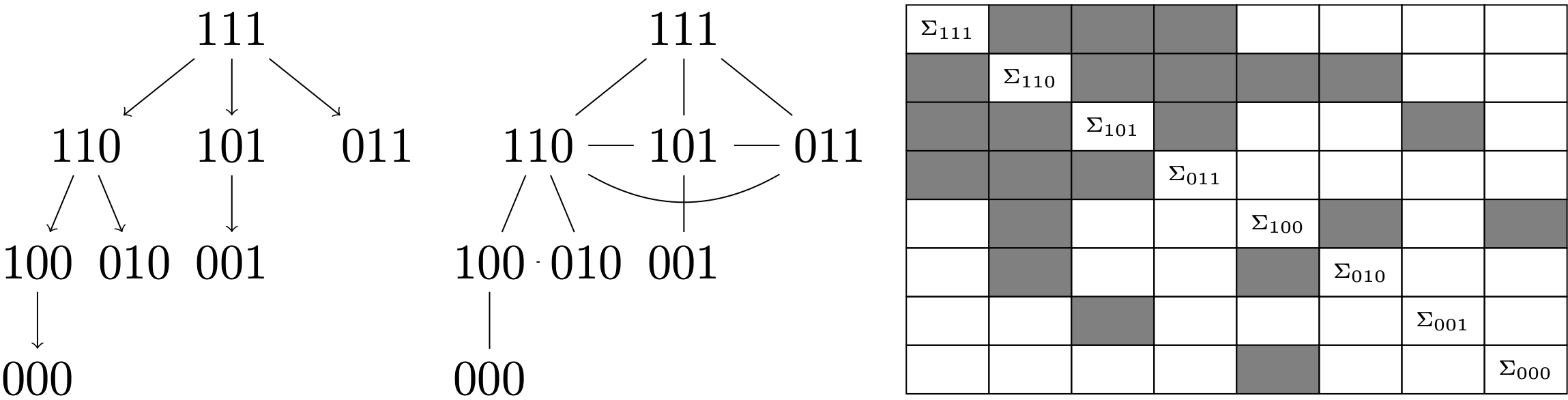}
		\caption{The left panel has the original tree graph, and the middle panel shows the resulting undirected graph after connecting the siblings.  Let $\Sigma_r$ denote the asymptotic covariance of the parameters associated with the primary model for pattern $r$.  The block structure of the asymptotic covariance matrix is depicted in the right panel, where the white regions refer to blocks of $0$s.  All other regions are not guaranteed to be $0$.}
		\label{fig:asymptoticMLE}
	\end{figure}
	
\end{example}

An interesting observation follows from Corollary \ref{cor:blockstructure}.  Since under the CCMV assumption, all models are fit using the complete case data, all MLEs will be correlated.  In contrast, under a GNCMV assumption, all models are fit with minimal data shared.  This leads to the idea of densest and sparsest asymptotic covariance matrices in Proposition \ref{prop:covariancematrix}.

\begin{proposition}
	\label{prop:covariancematrix}
	CCMV leads to densest asymptotic covariance matrix.  Any tree graph assumption belonging to $\mathcal{T}_{\text{GNCMV}}$ leads to the sparsest asymptotic covariance matrix.
\end{proposition}


\subsection{Empirical bootstrap}

In the bootstrap, we can overcome performing any analytic computation.  We have access to the joint bootstrap distribution, which mimics the joint sampling distribution.  Our empirical bootstrap approach utilizes resampling from the empirical distribution \citep{efron1979bootstrap}.  We describe the process of generating bootstrap samples and refitting the model to obtain bootstrap estimates in Algorithm \ref{alg:empbootstrap}.  Because we are operating under a smooth parametric model, the bootstrap is asymptotically valid, and an argument similar to the one provided by \cite{suen2023modelingmissingrandomneuropsychological} that uses the Berry-Esseen bound can be followed.

\begin{algorithm}
	{\bf Require:} $\{(\bX_{i,\bR_i}, \bR_i)\}_{i=1}^n$, $\hat{\btheta}$, $B$ (a large number, say 1,000)
	\begin{algorithmic}[1]
		\For{$b\in1,\ldots,B$}
		\State{Sample $n$ draws uniformly with replacement from $\{1, 2, \ldots, n\}$.  Put these into index set $I_b$.}
		\State{Set the $b$th bootstrapped data set $D_b^* = \{(\bX_{i,\bR_i}, \bR_i)\}_{i\in I_b}$.}
		\State{Fit the complete case distribution $p(\bx|\b\bR=1_d;\btheta^{*(b)})$ on $D^*_b[\b\bR=1_d]$.}
		\State{Obtain the selection odds $O_r(x;\hat{\bbeta}^{*(b)}) := P(\bR=r|x) / P(\bR=\text{PA}(r)|x)$ for every pattern $r\in\mathcal{R}$ on the data $D^*_b[\b\bR=r] \cup D^*_b[\b\bR=\text{PA}(r)]$.}
		\EndFor
		
		\State{\textbf{return} $\{\bbeta^{*(b)}\}_{b=1}^B, \{\btheta^{*(b)}\}_{b=1}^B$}
	\end{algorithmic}
	\caption{Empirical bootstrap procedure for obtaining confidence intervals}
	\label{alg:empbootstrap}
\end{algorithm}

Since we are under a parametric model, every statistical functional is a function of the parameters $\bbeta$ and $\btheta$.  In the situation that the statistical functional does not have a simple analytical form, we recommend computing a multiple imputation estimator, which serves as a Monte Carlo approximation.  For every bootstrap estimate $(\bbeta^{*(b)}, \btheta^{*(b)})$, we can construct the imputation distributions $p(x_{\bar{r}} | x_r, \bR=r)$ for every pattern $r$ and multiply impute.  After obtaining a completed data set, then we can compute the statistical functional by computing it on the multiply imputed data set.  Then, afterwards, we may pool these estimates together to construct a confidence interval.  We summarize this procedure in Algorithm \ref{alg:empbootstrapMI}.  This approach is generally computationally expensive because within each bootstrap iteration, we have to perform a multiple imputation step, but it overcomes the difficulty of finding a closed-form analytic expression for any general statistical functional we care about.  We note that another procedure could take an inverse probability weighting approach.  In general, however, we recommend a multiple imputation approach because this will be asymptotically more efficient than an IPW method.

\begin{algorithm}
	{\bf Require:} $\{(\bX_{i,\bR_i}, \bR_i)\}_{i=1}^n$, $\hat{\btheta}$, $B$ (a large number, say 1,000), $M$, $S(\cdot)$ (statistical functional)
	\begin{algorithmic}[1]
		\For{$b=1,2,\ldots,B$}
		\State{Sample $n$ draws uniformly with replacement from $\{1, 2, \ldots, n\}$.  Put these into index set $I_b$.}
		\State{Set the $b$th bootstrapped data set $D_b^* = \{(\bX_{i,\bR_i}, \bR_i)\}_{i\in I_b}$.}
		\State{Form the bootstrap empirical distribution $\hat{P}^{*(b)}$.}
		\State{Fit the complete case distribution $p(x|\bR=1_d;\btheta^{*(b)})$ on $D^*_b[\b\bR=1_d]$.}
		\State{Obtain the selection odds $O_r(x;\hat{\bbeta}^{*(b)}) := P(\bR=r|x) / P(\bR=\text{PA}(r)|x)$ for every pattern $r\in\mathcal{R}$ on the data $D^*_b[\b\bR=r] \cup D^*_b[\b\bR=\text{PA}(r)]$.}
		\For{$r\neq1_d$}
		\State{Construct the imputation distribution $p(x_{\bar{r}} | x_r, \bR=r; \bbeta^{*(b)}, \btheta^{*(b)})$ by renormalizing $O_r(x;\hat{\bbeta}^{*(b)}) \cdot p(x|\bR=1_d;\btheta^{*(b)})$.}
		\EndFor
		\For{$m=1,2,\ldots,M$}
		\For{$i=1,2,\ldots,n$}
		\If{$\bR_i\neq1_d$}
		\State{Impute $\tilde{\bX}_{i,\bar{\bR}_i}^{(b,m)}$.}
		\EndIf
		\EndFor
		\EndFor
		\State{Form the bootstrap empirical distribution $\hat{P}^{*(b),M}$.}
		\State{Compute the statistical functional $S(\hat{P}^{*(b),M})$ on the completed imputed data.}
		\EndFor
		
		\State{\textbf{return} $\{S(\hat{P}^{*(b),M})\}_{b=1}^B$}
	\end{algorithmic}
	\caption{Empirical bootstrap procedure with multiple imputation for obtaining confidence intervals}
	\label{alg:empbootstrapMI}
\end{algorithm}

\section{Sensitivity analysis}


In practical data analysis, it is essential to evaluate the influence of missing data assumptions on statistical estimators. Since such assumptions dictate the structure of the missing data mechanism, any misspecification can lead to biased or misleading inferences. In this paper, we focus on the tree graph as the primary missing data assumption and consider a structured sensitivity analysis framework to assess its impact.

Broadly, sensitivity analysis approaches can be categorized into deviations within the tree graph set and deviations outside the tree graph set. The former considers alternative graph structures that remain within the tree graph set while the latter relaxes the tree structure entirely, allowing for more flexible relationships. Both types of deviations allow one to assess the robustness of the estimator to different levels of structural perturbation.  We generally consider the former because that is most within the scope of this paper.

In all of these settings, the complete case model remains unchanged. However, models involving missing patterns (those dependent on assumptions about the missing data mechanism) are subject to perturbations.  By systematically examining these perturbations, we aim to quantify the sensitivity of inference to the assumed missing data structure. This approach provides a principled way to assess the degree to which conclusions depend on specific assumptions.

\label{section:sensitivity}

\subsection{Deviation within the tree graph set}


A natural approach to evaluating deviations in the tree graph framework is to consider a set of plausible tree graphs, denoted as $\widetilde{\mathcal{T}}$. These alternative graphs can be constructed by incorporating prior knowledge, existing partial orderings (such as the GNCMV framework) and data-driven methods. Exploring multiple tree structures allows us to assess the sensitivity of statistical inferences to different assumptions about the missing data mechanism.

Because the tree graph structure permits the use of conjugate odds imputation, as discussed earlier, we can perform statistical analyses for each tree in $\widetilde{\mathcal{T}}$, obtaining $|\widetilde{\mathcal{T}}|$ point estimates of the target parameter. Comparing these estimates provides insight into the impact of tree specification on inference, helping to determine whether certain structural choices lead to significant variation in results.



\subsection{Perturb selection odds models via exponential tilting}

Alternatively, one may consider deviating from the tree graph set, and there are multiple approaches one may take.  We discuss a straightforward one here in terms of the exponential tilting of the selection odds models.

A given selection odds model $O_r(x_r) = \frac{P(\bR=r|x_r)}{P(\bR=s|x_r)}$ can be commonly estimated using logistic regression, especially under our conjugate odds framework. To incorporate sensitivity analysis and assess the robustness of inferences under potential deviations from the assumed selection model, we introduce a perturbation mechanism via exponential tilting  \citep{kim2011semiparametric, shao2016semiparametric, zhao2017semiparametric}.  With many odds model and many variables, there can be an exponential number of sensitivity parameters one can have.  One approach is to consider variable-wise sensitivity parameters, where we have a sensitivity parameter for each of the $d$ variables.  The sensitivity parameter vector can be $\rho := (\rho_1, \rho_2, \ldots, \rho_d)$.

Specifically, we modify the selection odds model by multiplying it with an exponential adjustment term, leading to the perturbed selection odds model
$$O_r'(x) := O_r(x_r) \cdot \exp(\rho_{\bar{r}}^\top x_{\bar{r}}),$$
where $\rho_{\bar{r}}$ consists of the sensitivity parameters (one for each missing covariate under pattern $r$).  The exponential tilting formulation allows for a flexible and interpretable perturbation of the selection model. By appropriately choosing $\rho_{\bar{r}}$, one can examine a range of plausible missing data mechanisms, thereby assessing the sensitivity of the resulting inference.  Each element of $\rho_{\bar{r}}$ represents a potential deviation from the originally estimated selection model, effectively shifting the selection mechanism in a controlled manner.

Despite the introduction of the perturbation term, the new perturbed selection odds model remains within a parametric logistic regression framework. The exponential tilting approach does not alter the functional form of the selection odds model beyond a simple multiplicative adjustment. As a result, the model retains its parametric interpretability. When a given $\rho_j=0$, this corresponds to no perturbation.  Such a sensitivity parameter can be viewed as coefficient in a linear model, and one can specify its range based on one's belief of the relative impact of the missing variables to that of the observed variables.

\begin{remark}
	This approach is equivalent to the approach by \cite{Franks19045}.  In that approach, they consider a single variable $Y$ that is subject to missingness and impose a parametric assumption on selection probability $P(\bR=1|y) = \text{logit}(\alpha+\beta y)$, where $\beta$ is given a prior distribution.
	
	One can view their modeling approach as special case of our framework with a Bayesian perspective.  In our framework, their technique can be viewed as a tree graph approach with an exponential tilting sensitivity analysis.  To see this, consider the simple tree $1 \to 0$ that provides the identification assumption
	$$P(\bR=0|y)/P(\bR=1|y) \stackrel{1\to0}{=} P(\bR=0)/P(\bR=1) =: O_0.$$
	Then, $\beta$ can be interpreted as a sensitivity parameter through we can define a perturbed selection odds model that now depends on the variable $y$
	$$O'_0(y) := O_0 \cdot \exp(\beta y).$$
\end{remark}

\section{Proofs}

%

\subsection{Tree graphs and associated properties}

\label{subappendix:prooftreegraphprop}

\begin{proof}[Proof of Proposition \ref{prop:mnar}]
	Any given tree graph is equivalent to a missing not at random assumption.  Pick any missing data pattern $r_\ell$ with associated path $1_d \to r_1 \to r_2 \to \cdots \to r_\ell$.  The selection odds factorizes with respect to the tree graph, so we have the following decomposition
	\begin{align*}
		\frac{P(\bR=r_\ell|X)}{P(\bR=r_{1_d}|X)} &= \prod_{i=1}^\ell \frac{P(\bR=r_i|X)}{P(\bR=r_{i-1}|X)} \\
		&= \prod_{i=1}^\ell \frac{P(\bR=r_i|X_{r_i})}{P(\bR=r_{i-1}|X_{r_{i}})} \\
		&= \prod_{i=1}^\ell f_{r_i}(X_{r_{i}})
	\end{align*}
	for functions $\{f_{r_i}\}_{r_i\in\mathcal{R}}$.
	Multiplying both sides by $P(\bR=1_d|X)$ implies that $P(\bR=r_\ell|X)$ depends on $X_{\overline{r_\ell}}$, which implies it is MNAR assumption.\\
	
	Moreover, we know that the pattern-mixture model factorization holds since it is equivalent to the selection model (see Theorem 4 of \citealt{chen2022}).  Factor the full data distribution as
	\begin{align*}
		p(x, r) = p(x_{\bar{r}} | x_r, r) \cdot p(x_r | r) \cdot p(r).
	\end{align*}
	The extrapolation distributions $\{p(x_{\bar{r}} | x_r, r)\}_{r\in\mathcal{R}}$ are the only distributions not identified from the observed data.  However, the tree graph provides a way to identify each distribution from the observed data.  As above, pick any missing data pattern $r_\ell$ with associated path $1_d \to r_1 \to r_2 \to \cdots \to r_\ell$.  Then, we have
	\begin{align*}
		p(x_{\bar{r}_\ell} | x_{r_\ell}, r_\ell) &\stackrel{T}{=} p(x_{\bar{r}_\ell} | x_{r_\ell}, r\in\text{PA}_{T}(r_\ell)) \\
		&= p(x_{\bar{r}_\ell} | x_{r_\ell}, r\in\text{PA}_{T}(r_\ell))
	\end{align*}
	
	Thus, the full data distribution is nonparametrically identified.
\end{proof}

\begin{proof}[Proof of Proposition \ref{prop:mcar}]
Through rules of probability, we know that
$$p(x|\bR=r) = p(x|\bR=1_d) \cdot \frac{P(\bR=r|x)}{P(\bR=1_d|x)} \cdot \frac{P(\bR=1_d)}{P(\bR=r)},$$
and $p(x|\bR=r)$ is identified because the odds $\frac{P(\bR=r|x)}{P(\bR=1_d|x)}$ simplifies as a product of identifiable terms by Proposition \ref{prop:mnar}.
Under missing completely at random ($X\perp R$), these odds rewrite as $\dfrac{P(\bR=r|x)}{P(\bR=1_d|x)} = \dfrac{P(\bR=r)}{P(\bR=1_d)}$.  Therefore, the above equation simplifies as
$$p(x|\bR=r) = p(x|\bR=1_d) \cdot \frac{P(\bR=r)}{P(\bR=1_d)} \cdot \frac{P(\bR=1_d)}{P(\bR=r)} = p(x|\bR=1_d).$$
Finally, this means that
$$p(x|\bR=1_d) = p(x),$$
and so, the tree graph correctly recovers the data distribution under MCAR.
\end{proof}

\begin{proof}[Proof of Proposition \ref{prop:equivprop}]
	To prove equivalence of all three statements, we prove in a cycle.
	\begin{itemize}
		\item [($1\Rightarrow 2$)] Definition \ref{def:treegraph} implies the single parent property.  We prove the contrapositive.  Suppose there exists one pattern $r$ with two parents, labeled $s_1$ and $s_2$.  Then, the path from $1_d \to \cdots \to s_1 \to r$ and $1_d \to \cdots \to s_2 \to r$ both exist in the graph, which means that there is more than one path to $r$ from $1_d$.
		\item [($2\Rightarrow 3$)] Next, suppose that every pattern $r\neq 1_d$ in $T$ has exactly one parent.  There are exactly $2^d$ patterns in $T$ with $1_d$ as the source, so there are $2^d-1$ that require a parent.  Thus, $T$ must contain at least $2^d-1$ edges, and since every pattern $r\neq1_d$ has only parent, there are no more $2^d-1$ total edges.
		\item [($3\Rightarrow 1$)] Lastly, we prove by contradiction.  Suppose that instead of $2^d-1$ edges, there are $2^d$ total edges in $T$.  By the Pigeonhole Principle, there exists one pattern with at least $\lceil 2^d / (2^d-1) \rceil = 2$ parents.  Denote this pattern by $r$.  If $r$ has at least 2 parents, then there are at least two paths from $1_d$ to $r$, which implies this is not a tree graph, thereby resulting in a contradiction.
	\end{itemize}

	
	
\end{proof}

\begin{lemma}
	\label{lemma:comb-ident}
	The following combinatorial identities hold
	$$\sum_{m=0}^d \binom{d}{m} = 2^d, \quad \sum_{m=0}^d m\binom{d}{m} = d \cdot 2^{d-1}.$$
\end{lemma}

\begin{proof}[Proof of Lemma \ref{lemma:comb-ident}]
	We will prove both equations using combinatorial arguments.  We prove the first equation first.  Observe that the LHS counts the number of ways to form a committee of size $0$ to $d$ from $d$ individuals.  Alternatively, one can count the number of committees by noting that each of the $d$ individuals has two choices: to be in the committee or not.  We then obtain $2^d$ on the RHS.  Thus, equality holds.
	
	For the second equation, note that the LHS counts the number of ways to form a committee of any size with a leader.  We can alternatively count this by selecting the leader first from $d$ individuals and then forming a committee from the remaining $d-1$ individuals.  This precisely gives us $d \cdot 2^{d-1}$ on the RHS, and equality holds.
\end{proof}

\begin{proof}[Proof of Proposition \ref{prop:enumerate}]
	In a regular pattern graph, the observed variables that a pattern has is exactly a subset of its parents' observed variables.  Therefore, every missing pattern $r$ has $2^m-1$ parents, where $m$ is the number of 0s in $r$.  There are $\binom{d}{m}$ patterns with $m$ 0s.  Thus, using Property \ref{fact:reduct} and since $m$ ranges from $1$ to $d$, we have
	\begin{equation*}
		|\mathcal{T}_d| = \prod_{m=1}^d (2^m-1)^{\binom{d}{m}}.
	\end{equation*}
	Since $2^{m-2} < 2^m - 1$ when $m > \log_2(4/3) \approx 0.415$, we have the following lower bound
	\begin{align*}
		|\mathcal{T}_d| &\geq \prod_{m=1}^d (2^{m-2})^{\binom{d}{m}} \\
		&= 2^{\sum_{m=1}^d (m-2) \binom{d}{m}}.
	\end{align*}
	Focusing on the term in the exponent, we have
	\begin{align*}
		\sum_{m=1}^d (m-2) \binom{d}{m} &= \sum_{m=1}^d m \binom{d}{m} - 2 \sum_{m=1}^d \binom{d}{m} \\
		&= \sum_{m=0}^d m \binom{d}{m} - 2 \left(\sum_{m=0}^d \binom{d}{m} - 1\right) \\
		&= d\cdot 2^{d-1} - 2(2^d-1) \\
		&= (d/2-2)\cdot 2^d + 2,
	\end{align*}
	where the second to last inequality can be obtained via standard combinatorial arguments (see Lemma \ref{lemma:comb-ident} for completeness).  This implies that
	$$|\mathcal{T}_d| \geq 2^{(d/2-2)\cdot 2^d + 2} = 2^{\Omega(d\cdot 2^d)},$$
	which is precisely super-exponential. 
\end{proof}

\begin{proof}[Proof of Proposition \ref{prop:GNCMV}]
	First, we bound the size of $T\in\mathcal{T}_{\text{GNCMV}}$.  For every pattern with $m$ observed variables, there are a total of a $\binom{d}{m}$ patterns.  Thus, we have
	$$|\mathcal{T}_{\text{GNCMV}}| = \prod_{m=0}^{d-1} \binom{d}{m+1}^{\binom{d}{m}}.$$
	For sufficiently large $d$, we further obtain
	$$\log|\mathcal{T}_{\text{GNCMV}}| = \sum_{m=0}^{d-1} \binom{d}{m} \log\binom{d}{m+1} \geq \sum_{m=0}^{d-1} \binom{d}{m} = \Omega(2^d).$$
	So, $\mathcal{T}_{\text{GNCMV}}$ forms a large class.  Next, for every $T\in\mathcal{T}_{\text{GNCMV}}$, we prove the following two properties:
	\begin{enumerate}
		\item It achieves the maximum possible depth of $d$.
		
		Since every pattern is the farthest it can be from $1_d$, this implies that the longest chain in the graph is formed via the path from $1_d$ to $0_d$.  This chain has length $d$, which implies that this NCMV graph has the maximum possible depth, in contrast to CCMV.
		
		\item Every pattern $r$ in $T$ is positioned at the maximum possible distance from the source node $1_d$, thereby corresponding to the most information flow.
		
		By construction, every pattern $r\neq 1_d$ has a parent $s$ such that $s$ contains exactly one more observed variable than $r$.  Therefore, this implies that length of the path from the source node $1_d$ to any pattern that contains $k$ 0s is exactly $k$.  Moreover, this is maximum distance away from the source node it can be because 
	\end{enumerate}
\end{proof}

\subsection{Conjugate odds properties}

\begin{proof}[Proof of Proposition \ref{prop:conjugatemixtures}]
	Suppose that $p(x|A=a)$ belongs to the $K$-mixture model
	$$\mathcal{M}_K(\mathcal{P}) := \left\{p = \sum_{j=1}^K w_j p_j \ \biggr| \ p_j \in\mathcal{P}, \sum_{j=1}^K w_j = 1, w_j > 0\ \forall j \right\}$$
	such that every component is an element of $\mathcal{P}$, and the odds $\mathcal{O}(\mathcal{P})$ is a conjugate odds for $\mathcal{P}$.
	
	To be precise, suppose we can write $p(x|A=a)$ as
	$$p(x|A=a) = \sum_{j=1}^K w_j p_j(x|A=a)$$
	for positive weights $w_j$ that sum to 1.  Now, suppose that $P(A=a'|x)/P(A=a|x)$ is conjugate for $p_j(x|A=a)$ for all $j$.  Then, we have
	\begin{align*}
		p(x|A=a') &\propto p(x|A=a) \cdot \frac{P(A=a'|x)}{P(A=a|x)} \\
		&= \sum_{j=1}^K w_j p_j(x|A=a) \cdot \frac{P(A=a'|x)}{P(A=a|x)} \\
		&= \sum_{j=1}^K w_j \cdot \zeta_j \cdot p_j(x|A=a')
	\end{align*}
	for some $\{\zeta_j\}_j$ that are all positive due to the conjugacy of the odds model.  Finally, this implies that
	$$p(x|A=a') = \sum_{j=1}^K \tilde{w}_j p_j(x|A=a')$$
	for some set of perturbed weights $\{\tilde{w}_j\}_j$.  So, we have that $p(x|A=a')$ is also $K$-mixture model with components belonging to $\mathcal{P}$, and the result follows.
\end{proof}

\begin{proof}[Proof of Proposition \ref{prop:expfam-conjugate}]
	We assume the $p(x|A=a)$ has the following exponential family parameterization
	$$p(x|A=a) = h(x)g(\eta)\exp(\eta^\top T(x)).$$
	We also further assume that the odds is a logistic regression and linear in the sufficient statistic
	$$\log \frac{P(A=a'|x)}{P(A=a|x)} = \gamma_0 + \gamma^\top T(x).$$
	We have the following equality
	\begin{align*}
		p(x|A=a') &= \frac{P(A=a'|x)}{P(A=a|x)} \cdot p(x|A=a) \biggr / \int_{-\infty}^\infty \frac{P(A=a'|x)}{P(A=a|x)} \cdot p(x|A=a) \ dx. \label{eq:newdist-expfam}
	\end{align*}
	Focusing on the unnormalized distribution, we have
	\begin{align}
		\frac{P(A=a'|x)}{P(A=a|x)} \cdot p(x|A=a) &= \exp(\gamma_0 + \gamma^\top T(x)) \cdot h(x)g(\eta)\exp(\eta^\top T(x)) \nonumber \\
		&= \exp(\gamma_0) h(x) g(\eta) \exp((\eta+\gamma)^\top T(x)).
	\end{align}
	As $h(x)\exp((\eta+\gamma)^\top T(x))$ is an unnormalized exponential family distribution with natural parameter $\eta+\gamma$, it follows that
	$$\int_{-\infty}^\infty h(x)\exp((\eta+\gamma)^\top T(x)) \ dx = \frac{1}{g(\eta+\gamma)}.$$
	Returning to equation \eqref{eq:newdist-expfam}, we see that
	\begin{align*}
		p(x|A=a') &= \frac{P(A=a'|x)}{P(A=a|x)} \cdot p(x|A=a) \biggr / \int_{-\infty}^\infty \frac{P(A=a'|x)}{P(A=a|x)} \cdot p(x|A=a) \ dx \\
		&= \frac{\exp(\gamma_0) \cdot h(x) \cdot g(\eta) \cdot \exp((\eta+\gamma)^\top T(x))}{\exp(\gamma_0) \cdot g(\eta) \cdot 1/g(\eta+\gamma)} \\
		&= h(x) g(\eta+\gamma) \exp((\eta+\gamma)^\top T(x)), \\
		&= h(x) g(\eta') \exp((\eta')^\top T(x))
	\end{align*}
	with $\eta' := \eta + \gamma$.  Finally, solving for $\gamma_0$, we have
	$$\exp(\gamma_0) := \frac{g(\eta')}{g(\eta)} \cdot \frac{P(A=a')}{P(A=a)} \iff \gamma_0 := \log \frac{P(A=a')}{P(A=a)} + \log \frac{g(\eta')}{g(\eta)},$$
	as desired.
\end{proof}

\begin{proof}[Proof of Corollary \ref{cor:expontentialtilting}]
	Suppose that $P(A=a'|x)/P(A=a|x)$ is modeled using a logistic regression.  That is, we model the log-odds like
	$$\log \frac{P(A=a'|x)}{P(A=a|x)} = g(x)$$
	for some function $g$.  Then, we have
	$$p(x|A=a') \propto p(x|A=a) \cdot \frac{P(A=a'|x)}{P(A=a|x)} = p(x|A=a) \cdot \exp(g(x)),$$
	and this is precisely an exponential tilting, as desired.
\end{proof}

\begin{proof}[Proof of Corollary \ref{cor:mixexpfam-conjugate}]
	By definition, we can write $p(x|A=a')$ as
	\begin{align*}
		p(x|A=a') &\propto p(x|A=a) \cdot \frac{P(A=a'|x)}{P(A=a|x)} \\
		&= \exp(\gamma_0 + \gamma^\top T(x)) \sum_{k=1}^K w_k \cdot h(x) g(\eta_k) \exp(\eta^\top_k T(x)) \\
		&= \sum_{k=1}^K w_k \cdot \exp(\gamma_0) h(x) g(\eta_k) \exp((\eta_k+\gamma)^\top T(x)).
	\end{align*}
	Then, simplifying with algebra and renormalizing yields
	$$p(x|A=a') = \sum_{k=1}^K \widetilde{w}_k \cdot h(x) g(\eta_k+\gamma) \exp((\eta_k+\gamma)^\top T(x)),$$
	where
	$$\widetilde{w}_k := \frac{w_k \cdot g(\eta_k)}{g(\eta_k+\gamma)} \bigg/ \sum_{k'=1}^K \frac{w_{k'} \cdot g(\eta_{k'})}{g(\eta_{k'}+\gamma)}.$$
	
	This concludes the proof.  Additionally, we note that in this proof we assumed that each mixture was the same distribution, but this argument generalizes to other distributions.  For example, instead of just a mixture of Gaussians, one could have a mixture of a Gaussian and a Binomial.
\end{proof}

\begin{proof}[Proof of Proposition \ref{prop:power-pareto}]
	We assume the $p(x|A=a)$ has the following Pareto distribution parameterization
	$$p(x|A=a; \alpha, \beta) = \begin{cases} \dfrac{\alpha \beta^\alpha}{x^{\alpha+1}} & x \geq \beta \\ 0 & \text{o.w.} \end{cases}.$$
	We also further assume that the odds obeys the following parametric model
	$$\frac{P(A=a'|x)}{P(A=a|x)} = \gamma_0x^\gamma, \quad \gamma_0 := \frac{P(A=a')}{P(A=a)} \cdot \frac{\alpha'}{\alpha} \cdot \beta^{\alpha'-\alpha}.$$
	We have the following equality
	\begin{align*}
		p(x|A=a') &= \frac{P(A=a'|x)}{P(A=a|x)} \cdot p(x|A=a) \biggr / \int_\beta^\infty \frac{P(A=a'|x)}{P(A=a|x)} \cdot p(x|A=a) \ dx.
	\end{align*}
	Focusing on the unnormalized distribution, we have
	\begin{align*}
		\frac{P(A=a'|x)}{P(A=a|x)} \cdot p(x|A=a) &= \frac{\alpha \beta^\alpha \gamma_0}{x^{\alpha-\gamma+1}}.
	\end{align*}
	The normalizing constant must be
	$$\int_\beta^\infty \frac{\alpha \beta^\alpha \gamma_0}{x^{\alpha-\gamma+1}} \ dx = \frac{\alpha \beta^\alpha \gamma_0}{(\alpha-\gamma)\beta^{\alpha-\gamma}} = \frac{\alpha \beta^\gamma\gamma_0}{\alpha-\gamma}.$$
	Thus, we have
	$$p(x|A=a'; \alpha', \beta) = \begin{cases} \dfrac{\alpha' \beta^{\alpha'}}{x^{\alpha'+1}} & x \geq \beta \\ 0 & \text{o.w.} \end{cases}$$
	for $\alpha' := \alpha - \gamma$.
	Finally, solving for $\gamma_0$, we have
	$$\gamma_0 := \frac{P(A=a')}{P(A=a)} \cdot \frac{\alpha'}{\alpha} \cdot \frac{\beta^{\alpha'}}{\beta^\alpha},$$
	as desired.
\end{proof}

\begin{proof}[Proof of Corollary \ref{cor:productdist}]
	Suppose we have the decomposition $X:= (X_1,X_2)$, where $X_1 \perp X_2 \ | \ A$.  Let $p(x_1|A=a)$ and $p(x_2|A=a)$ be two exponential family distributions such that
	$$p(x_1|A=a) = h_1(x_1)g_1(\eta_1)\exp(\eta_1^\top T_1(x_1)),$$
	$$p(x_2|A=a) = h_2(x_2)g_2(\eta_2)\exp(\eta_2^\top T_2(x_2)).$$
	Then, we have
	\begin{align*}
		p(x|A=a) &= p(x_1|A=a) \cdot p(x_2|A=a) \\
		&= h_1(x_1)h_2(x_2)g_1(\eta_1)g_2(\eta_2)\exp(\eta_1^\top T_1(x_1)+\eta_2^\top T_2(x_2)).
	\end{align*}
	Finally, it follows that 
	\begin{align*}
		p(x|A=a') &\propto p(x|A=a) \cdot P(A=a'|x)/P(A=a|x) \\
		&= h_1(x_1)h_2(x_2)g_1(\eta_1)g_2(\eta_2)\exp((\eta_1+\gamma_1)^\top T_1(x_1)+(\eta_2+\gamma_2)^\top T_2(x_2)).
	\end{align*}
	Thus, the natural parameter is $\zeta := (\eta_1+\gamma_1, \eta_2+\gamma_2)$ and $T(x):=(T_1(x_1), T_2(x_2))$, as desired.
\end{proof}

\subsection{Tree graphs and conjugate odds}

\begin{proof}[Proof of Theorem \ref{theorem:treegraphandconjugateodds}]
	By assumption, the missingness mechanism can be specified using a tree graph.  Then, by Proposition \ref{prop:mnar}, we obtain an identification formula for the selection odds
	\begin{align*}
		\frac{P(\bR=r_\ell|X)}{P(\bR=1_d|X)}
		&\stackrel{\text{tree graph}}{=} \prod_{i=1}^\ell \frac{P(\bR=r_i|X_{r_i})}{P(\bR=r_{i-1}|X_{r_{i}})},
	\end{align*}
	where $1_d =: r_0 \to r_1 \to r_2 \to \cdots \to r_\ell$ is the unique path in the tree graph from the source $1_d$ to pattern $r_{\ell}$.
	
	We proceed with the proof inductively.  First, partition the set of patterns $\mathcal{R}$ into $\mathcal{R}_0, \mathcal{R}_1, \mathcal{R}_2, \ldots, \mathcal{R}_d$, where $\mathcal{R}_k$ denotes the set of patterns in the tree graph that are exactly $k$ edges away from the source node $1_d$.  Here, $\mathcal{R}_0$ is trivially the set $\{1_d\}$.
	
	For the base case, it is sufficient to consider the set $\mathcal{R}_1$, and note that for any $r\in\mathcal{R}_1$, we have
	$$p(x|\bR=r) \propto p(x|\bR=1_d) \cdot \frac{P(\bR=r|x)}{P(\bR=1_d|x)}.$$
	Therefore, by conjugacy of the odds, $p(x|\bR=r)$ is the same probability family as $p(x|\bR=1_d)$ for any $r\in\mathcal{R}_1$.
	
	Next, fix $k$, and suppose that for all $r\in\mathcal{R}_k$, $p(x|\bR=r)$ is the same probability family as $p(x|\bR=1_d)$.  Then, for any $s\in\mathcal{R}_{k+1}$, there exists $r'\in\mathcal{R}_k$ such that $r'\to s$ ($r'$ is the unique parent of $s$).  It follows that
	$$p(x|\bR=s) \propto p(x|\bR=r) \cdot \frac{P(\bR=s|x)}{P(\bR=r|x)}.$$
	Again, by the inductive hypothesis and conjugacy of the selection odds, it must follow that $p(x|\bR=s)$ also belongs to the same probability family as $p(x|\bR=r)$ and thus, also $p(x|\bR=1_d)$ by transitivity.
\end{proof}

\subsection{Inference}

\begin{proof}[Proof of Proposition \ref{prop:MLEindep}]
	
	
	We prove this directly.  Suppose that the patterns do not have a direct parent-child relationship.
	
	Let $\mathcal{X}_{r} := \{X_{i,R_i} \ | \ R_i=r\}$ be the observed data under pattern $r$, so by definition, $\mathcal{X}_{a} \cap \mathcal{X}_{a'} = \emptyset$ for $a\neq a'$.  There are two types of models fit on the data.  The first model is the complete case model, which is only using data $\mathcal{X}_{1_d}$.  The pattern $1_d$ contains no siblings.  Any odds based on the patterns with $1_d$.
	
	The conjugate odds model $O_r(x) := P(\bR=r|x)/P(\bR=r'|x)$ is fit using the data $\mathcal{X}_{r} \cup \mathcal{X}_{r'}$.  If two conjugate odds models are fit using completely separate data, then their resulting parameter estimates will be independent (this can be viewed as a form of sampling splitting).
	
	Now, consider two distinct patterns $r$ and $s$ such that $r'\to r$ and $s'\to s$.  Suppose we fit two conjugate odds models $O_r(x)$ and $O_s(x)$ using the data $\mathcal{X}_{r} \cup \mathcal{X}_{r'}$ and $\mathcal{X}_{s} \cup \mathcal{X}_{s'}$.  Then, failing to satisfy the first property necessarily implies that $s\neq r'$ and $r\neq s'$.  Failing to satisfy the second property implies that $r'\neq s'$.  All together, we have
	\begin{align*}
		(\mathcal{X}_{r} \cup \mathcal{X}_{r'}) \cap (\mathcal{X}_{s} \cup \mathcal{X}_{s'}) &= ((\mathcal{X}_{r} \cup \mathcal{X}_{r'}) \cap \mathcal{X}_{s}) \cup ((\mathcal{X}_{r} \cup \mathcal{X}_{r'}) \cap \mathcal{X}_{s'}) \\
		&= \emptyset \cup \emptyset \\
		&= \emptyset.
	\end{align*}
	So, the two models are fit on separate data sets.  This implies that the estimated model parameters are independent, and therefore, have covariance 0.
\end{proof}

\begin{proof}[Proof of Corollary \ref{cor:blockstructure}]
	Consider every maximal clique in the undirected graph.  If there is a path between two patterns in the undirected graph, then the estimated parameters for each of the models are correlated.  Moreover, for every submatrix that is determined by the maximal clique, the submatrix is full; that is, it contains only nonzero elements.
\end{proof}

\begin{proof}[Proof of Proposition \ref{prop:covariancematrix}]
	We start by proving the first claim.  In the CCMV case, the associated undirected graph is a clique, and the complete case data is used to fit every conjugate odds model $O_r(x) = O_r(x_r) := P(\bR=r|x_r)/P(\bR=1_d|x_r)$.  Therefore, the estimated parameters for all the conjugate odds models and the complete case model are all dependent.  Thus, the correlation is nonzero, and CCMV provides the densest covariance matrix.
	
	Now, we consider the second claim: any GNCMV assumption provides the sparsest asymptotic covariance matrix.  The undirected graph associated with any GNCMV tree graph contains only maximal cliques of size 2, thereby leading to the sparsest possible matrix.
\end{proof}

\subsection{Proofs of additional results in the appendix}

\begin{proof}[Proof of Proposition \ref{prop:power-beta}]
	We assume the $p(x|A=a)$ has the following Beta distribution parameterization
	$$p(x|A=a; \alpha, \beta) = \frac{x^{\alpha-1}(1-x)^{\beta-1}}{B(\alpha,\beta)}.$$
	We also further assume that the odds obeys the following parametric model
	$$O_{a'}(x;\bm{\gamma}) := \frac{P(A=a'|x)}{P(A=a|x)} = \gamma_0x^{\gamma_1} (1-x)^{\gamma_2}, \quad \gamma_0 := \frac{P(A=a')}{P(A=a)} \cdot \frac{B(\alpha+\gamma_1, \beta+\gamma_2)}{B(\alpha,\beta)}.$$
	
	We have the following equality
	\begin{align*}
		p(x|A=a') &= \frac{P(A=a'|x)}{P(A=a|x)} \cdot p(x|A=a) \biggr / \int_0^1 \frac{P(A=a'|x)}{P(A=a|x)} \cdot p(x|A=a) \ dx.
	\end{align*}
	Focusing on the unnormalized distribution, we have
	\begin{align*}
		\frac{P(A=a'|x)}{P(A=a|x)} \cdot p(x|A=a) &= \frac{\gamma_0 x^{\alpha+\gamma_1-1}(1-x)^{\beta+\gamma_2-1}}{B(\alpha,\beta)}.
	\end{align*}
	The normalizing constant must be
	$$\int_0^1 \frac{\gamma_0 x^{\alpha+\gamma_1-1}(1-x)^{\beta+\gamma_2-1}}{B(\alpha,\beta)} \ dx = \frac{\gamma_0 B(\alpha+\gamma_1, \beta+\gamma_2)}{B(\alpha,\beta)}.$$
	Thus, we have
	$$p(x|A=a'; \alpha', \beta) = \frac{x^{\alpha'-1}(1-x)^{\beta'-1}}{B(\alpha',\beta')},$$
	where $\alpha' = \alpha + \gamma_1  \in \mathbb{R}^+$ and $\beta' = \beta + \gamma_2  \in \mathbb{R}^+$.
	Finally, solving for $\gamma_0$, we have
	$$\gamma_0 := \frac{P(A=a')}{P(A=a)} \cdot \frac{B(\alpha+\gamma_1, \beta+\gamma_2)}{B(\alpha,\beta)},$$
	as desired.
\end{proof}

\begin{proof}[Proof of Proposition \ref{prop:power-dirichlet}]
	We assume the $p(x|A=a)$ has the following Dirichlet distribution parameterization
	$$p(x|A=a; \bm\alpha) = \frac{1}{B(\bm\alpha)} \prod_{j=1}^K x_j^{\alpha_j-1}$$
	We also further assume that the odds obeys the following parametric model
	$$O_{a'}(x;\bm{\gamma}) := \frac{P(A=a'|x)}{P(A=a|x)} = \gamma_0 \prod_{j=1}^K x_j^{\gamma_j}, \quad \gamma_0 := \frac{P(A=a')}{P(A=a)} \cdot \frac{B(\bm\alpha+\bm\gamma)}{B(\bm\alpha)}.$$
	
	We have the following equality
	\begin{align*}
		p(x|A=a') &= \frac{P(A=a'|x)}{P(A=a|x)} \cdot p(x|A=a) \biggr / \int_{\Delta_{K-1}} \frac{P(A=a'|x)}{P(A=a|x)} \cdot p(x|A=a) \ dx.
	\end{align*}
	Focusing on the unnormalized distribution, we have
	\begin{align*}
		\frac{P(A=a'|x)}{P(A=a|x)} \cdot p(x|A=a) &= \frac{\gamma_0}{B(\bm\alpha)} \prod_{j=1}^K x_j^{\alpha_j+\gamma_j-1}.
	\end{align*}
	The normalizing constant must be
	$$\int_{\Delta_{K-1}} \frac{\gamma_0}{B(\bm\alpha)} \prod_{j=1}^K x_j^{\alpha_j+\gamma_j-1} \ dx = \frac{\gamma_0 B(\bm\alpha+\bm\gamma)}{B(\bm\alpha)}.$$
	Thus, we have
	$$ \frac{1}{B(\bm\alpha')} \prod_{j=1}^K x_j^{\alpha_j'-1},$$
	where $\alpha_j' = \alpha_j + \gamma_j  \in \mathbb{R}^+$ for $j>0$.
	Finally, solving for $\gamma_0$, we have
	$$\gamma_0 := \frac{P(A=a')}{P(A=a)} \cdot \frac{B(\bm\alpha+\bm\gamma)}{B(\bm\alpha)},$$
	as desired.
\end{proof}

\begin{proof}[Proof of Theorem \ref{theorem:marginalinvariance}]
	Let $\mathcal{B} = \{j_1, j_2, \ldots, j_B\} \subseteq [d]$ be a set of indices that correspond to several variables, and $\mathcal{R}_{\mathcal{B}} := \{s : s_j = 0, j\in \mathcal{B}\}$.  Further, suppose that for any $r \in\mathcal{R}_{\mathcal{B}}$, $\text{Anc}_{G_1}(r) = \text{Anc}_{G_2}(r)$.  To prove sufficiency, it suffices to show that the implied distribution $p(x_{j_1}, x_{j_2}, \ldots, x_{j_B})$ is the same.  We approach this from a pattern-mixture model standpoint.
	
	Note that in a tree graph, a pattern's set of ancestors determines exactly the path from $1_d$ to that pattern.  A given pattern's set of ancestors has a total ordering, and this total ordering uniquely determines the path from $1_d$, as there is only a single path (in a tree graph).  Therefore, since the patterns share the same ancestors in the two tree graphs, the patterns all have the same implied distributions.
\end{proof}


\begin{proof}[Proof of Proposition \ref{prop:closure}]
To show both claims, it suffices to show that the set of tree graphs is the smallest generating set for the set of pattern graphs.  We do this by showing two facts: every generating set must be a superset of the set of tree graphs, and the set of tree graphs is a generating set.

First, note that any generating set must contain the set of tree graphs because every tree graph is minimal (see Proposition \ref{prop:equivprop}).  Next, the set of tree graphs is a generating set.  For any pattern graph $G$, construct an associated set of tree graphs $\mathcal{T}_G$, where each tree graph $T\in \mathcal{T}_G$ is made by choosing a single element in each parent set of $G$.

Then, note that performing this operation over all possible pattern graphs $G$ and taking a union forms a generating set.  More precisely, if we let $\mathcal{G}$ to be the set of all pattern graphs, we have
$$\bigcup_{G\in\mathcal{G}} \mathcal{T}_G \subseteq \mathcal{T}.$$
But also since $\bigcup_{G\in\mathcal{G}} \mathcal{T}_G$ is a generating set, we have
$$\bigcup_{G\in\mathcal{G}} \mathcal{T}_G \supseteq \mathcal{T}.$$
So, equality holds, and we are done.
\end{proof}

\section{Further Conjugate Odds Examples}

\label{appendix:examples}

\subsection{Further logistic odds examples}

\begin{example}[Negative binomial]
The negative binomial distribution is widely used in modeling discrete data with overdispersion with one notable example in single cell RNA data.  Suppose that $X|A=a\sim \text{NegBin}(r,p)$ with $r$ known.  The sufficient statistic is $T(x) = x$ with natural parameter $\eta = \log p$.  Suppose that 
$$\log \frac{P(A=a'|x)}{P(A=a|x)} = \gamma_0 + \gamma_1 x$$
where $\gamma = \gamma_1$.  Then, via Proposition \ref{prop:expfam-conjugate}, $X|A=a'$ is negative binomially distributed with natural parameter
$$\eta' := \eta + \gamma_1 = \log p + \gamma_1.$$  Translating this back to the standard parameterization, we have
\begin{align*}
	p' := \exp(\log p + \gamma_1).
\end{align*}
\end{example}


\begin{corollary}[Product distribution]
\label{cor:productdist}
Suppose we have the decomposition $X:= (X_1,X_2)$, where $X_1 \perp X_2 \ | \ A$.  Let $p(x_1|A=a)$ and $p(x_2|A=a)$ be two exponential family distributions such that
$$p(x_1|A=a) = h_1(x_1)g_1(\eta_1)\exp(\eta_1^\top T_1(x_1)),$$
$$p(x_2|A=a) = h_2(x_2)g_2(\eta_2)\exp(\eta_2^\top T_2(x_2)).$$
It follows that
\begin{align*}
	p(x|A=a') &= \underbrace{h_1(x_1) h_2(x_2)}_{h(x)} \underbrace{g_1(\eta_1+\gamma_1) g_2(\eta_2+\gamma_2)}_{g(\zeta)} \exp( (\eta_1+\gamma_1)^\top T_1(x_1) + (\eta_2+\gamma_2)^\top T_2(x_2)) \\
	&= h(x)g(\zeta)\exp(\zeta^\top T(x))
\end{align*}
with natural parameter $\zeta := (\eta_1+\gamma_1, \eta_2+\gamma_2)$ and $T(x):= (T_1(x_1), T_2(x_2))$.
\end{corollary}









\subsection{Further power law odds examples}

%
\begin{proposition}[Power function family, Beta distribution]
\label{prop:power-beta}
Suppose that $p(x|A=a)$ is a Beta distribution
$$p(x|A=a; \alpha, \beta) = \frac{x^{\alpha-1}(1-x)^{\beta-1}}{B(\alpha,\beta)}.$$
Then, the associated odds model
$$O_{a'}(x;\bm{\gamma}) := \frac{P(A=a'|x)}{P(A=a|x)} = \gamma_0x^{\gamma_1} (1-x)^{\gamma_2}, \quad \gamma_0 := \frac{P(A=a')}{P(A=a)} \cdot \frac{B(\alpha+\gamma_1, \beta+\gamma_2)}{B(\alpha,\beta)}$$
holds if and only if
$$p(x|A=a'; \alpha', \beta) = \frac{x^{\alpha'-1}(1-x)^{\beta'-1}}{B(\alpha',\beta')},$$
where $\alpha' = \alpha + \gamma_1  \in \mathbb{R}^+$ and $\beta' = \beta + \gamma_2  \in \mathbb{R}^+$.
\end{proposition}


\begin{proposition}[Power function family, Dirichlet distribution]
\label{prop:power-dirichlet}
Suppose that $X$ is a random variable belonging to the $K-1$ simplex such that
$$p(x|A=a; \bm\alpha) = \frac{1}{B(\bm\alpha)} \prod_{j=1}^K x_j^{\alpha_j-1}$$
is a Dirichlet distribution.  Then, the associated odds model
$$O_{a'}(x;\bm{\gamma}) := \frac{P(A=a'|x)}{P(A=a|x)} = \gamma_0 \prod_{j=1}^K x_j^{\gamma_j}, \quad \gamma_0 := \frac{P(A=a')}{P(A=a)} \cdot \frac{B(\bm\alpha+\bm\gamma)}{B(\bm\alpha)}$$
holds if and only if
$$p(x|A=a'; \bm\alpha') = \frac{1}{B(\bm\alpha')} \prod_{j=1}^K x_j^{\alpha_j'-1},$$
where $\alpha'_j = \alpha_j + \gamma_j \in \mathbb{R}^+$ for $j>0$.
\end{proposition}


\section{Additional Comments on Tree Graphs}

\subsection{Congruency}

We now discuss the concept of congruency.  When conducting real data analysis, we only need to construct a tree graph using patterns that are observed in the data and can ignore any pattern that is not observed in the real data.  We formalize this notation in the following section, utilizing the fact that patterns that are not observed in the real data have Lebesgue measure 0.


\begin{definition}[Congruency]
Two tree graphs $G_1$ and $G_2$ are said to be \textit{congruent} with respect to the observed data if $G_1$ and $G_2$ are identical after removing any patterns that do not appear in the observed data.  We omit the phrase ``with respect to the observed data'' when it is clear from context.  
\end{definition}

In essence, two tree graphs $G_1$ and $G_2$ being congruent with respect to the observed data means that any statistical functional of the full-data distribution is the same regardless if the assumption $G_1$ or $G_2$ was made.  We now introduce the idea of a representor graph to represent a set of congruent tree graphs.

\begin{definition}[Representor]
A representor of a set of tree graphs $\mathcal{T}$ is the tree graph that comprises only the patterns observed in the data and is congruent to every tree graph in $\mathcal{T}$.
\end{definition}

The graph comprising of only the patterns that are observed in the real data implies that the graph does not contain superfluous information that is ignored by the observed data.  Moreover, if the graph is congruent to every graph in $\mathcal{T}$, it is the minimal graph that represents all of the patterns.

\begin{example}
Suppose the observed data has $d=3$ variables with only the following patterns: 111, 101, 011, and 001.  Consider the following tree graphs in Figure \ref{fig:treegraph-example}, labeled from left to right as $G_1$, $G_2$, and $G_3$.  We see that $G_1$ and $G_2$ are congruent with respect to the observed data, but $G_3$ is not congruent to $G_1$ or $G_2$.  Moreover, the respective representor graphs of $\{G_1,G_2\}$ and $\{G_3\}$ are found in Figure \ref{fig:representor}.

\begin{figure}
	\centering
	\includegraphics[width=0.65\textwidth]{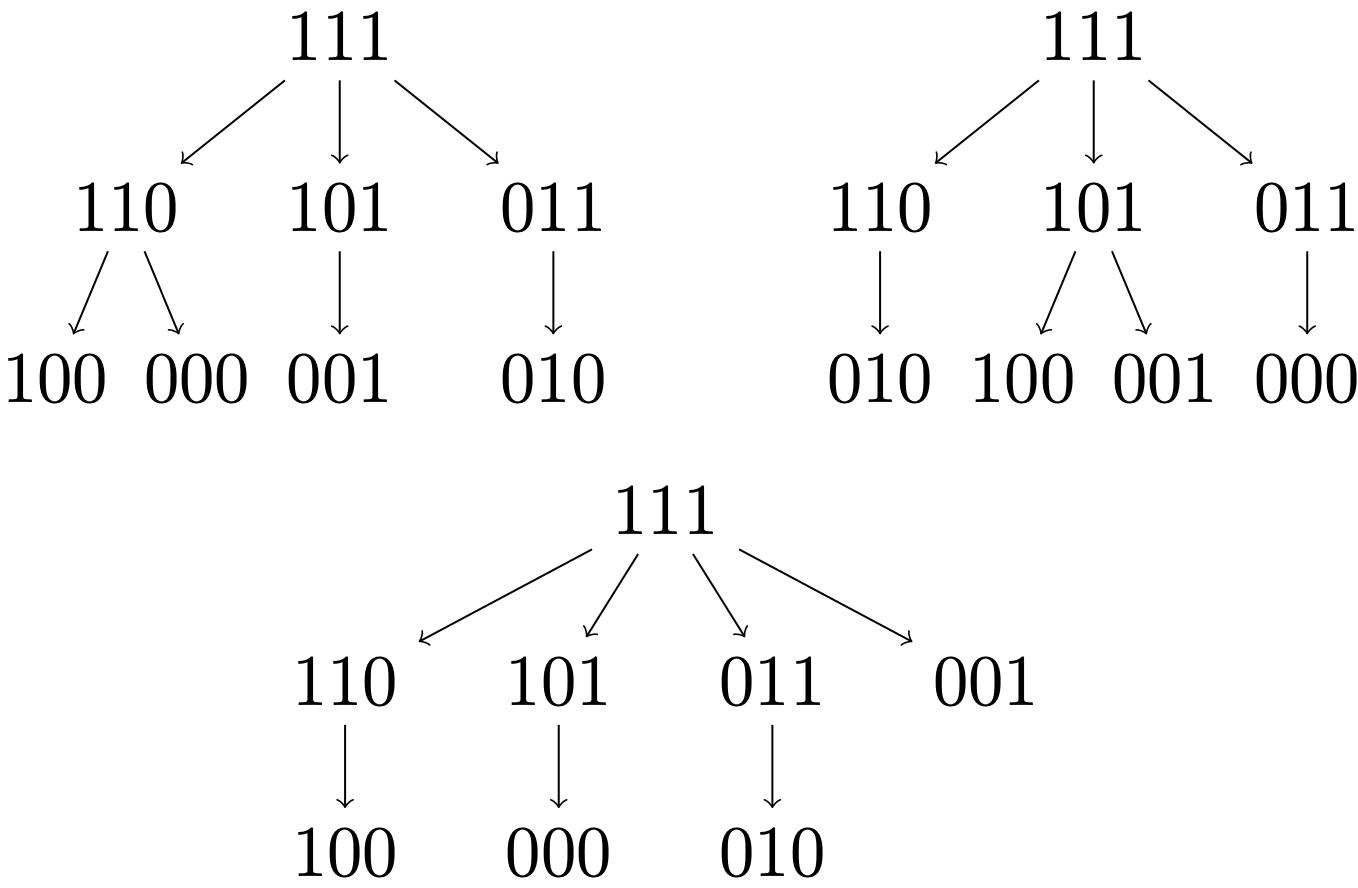}
	\caption{We provide three examples of tree graphs for $d=3$ variables.}
	\label{fig:treegraph-example}
\end{figure}

\begin{figure}[b]
	\centering
	\includegraphics[width=0.45\textwidth]{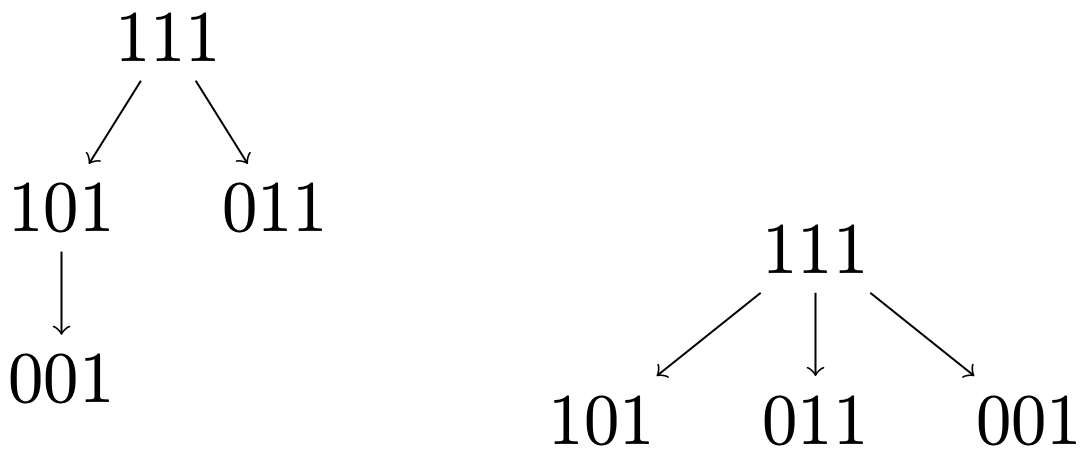}
	\caption{We provide examples of representor graphs, corresponding to the graphs in Figure \ref{fig:treegraph-example}.}
	\label{fig:representor}
\end{figure}

\end{example}

\begin{remark}[Selecting a threshold]
In practice, choosing the patterns that should appear in the representor can be done in various ways.  A straightforward way would be to only consider the patterns that are observed in the data.  More generally, one can consider thresholding based on the number of observations and only include missing patterns with a number of observations at least the threshold.

The threshold can be selected to be any constant $C>0$.  For example, if we only keep patterns such that the number of observations is at least $C=1$, this corresponds to selecting the patterns that are observed in the data.  On the other hand, we may choose to keep patterns such that the number of observations is at least $C=100$, which implies that we are seeking a sufficiently large enough effective sample size.  One advantage to choosing $C>1$ is to avoid potential problems with model fitting.
\end{remark}

\subsection{Invariance for a specific statistical functional}

One may hypothesize that for a given parameter of interest, certain tree graphs may lead to the same identification formula for that parameter.  More formally, we now consider invariance for a specific statistical functional.

\begin{theorem}[Sufficient conditions for marginal distribution invariance]
\label{theorem:marginalinvariance}
Let $\mathcal{B} = \{j_1, j_2, \ldots, j_B\} \subseteq [d]$ be a set of indices that correspond to several variables.  Define the set $\mathcal{R}_{\mathcal{B}} := \{s : s_j = 0, j\in \mathcal{B}\}$ to contain exactly the patterns that have a $0$ in each index in $\mathcal{B}$.  Suppose that there are two distinct tree graphs $G_1$ and $G_2$ such that for each $r\in\mathcal{R}_{\mathcal{B}}$, $r$ has the same ancestors in $G_1$ and $G_2$.  Then, any statistical functional of the distribution $p(x_{j_1}, x_{j_2}, \ldots, x_{j_B})$ is the same regardless of the graph $G_1$ or $G_2$.

\end{theorem}





A consequence of this theorem is that all tree graphs for $d=2$ induce unique marginal distributions.

Next, we now describe a way to combine tree graphs into a single pattern graph assumption using a merge operation.  Graphically, the merge operation is very simple.  Consider an example in Figure \ref{fig:merge}.  We formalize the merge operation in mathematical notation as follows.

\begin{definition}[Merge operation]
Consider the pattern graph $G := G_1 \cup G_2$, where $\cup$ between two pattern graphs denotes the merge operation.  The resulting graph $G := G_1 \cup G_2$ is constructed such that $\text{PA}_G(r) = \text{PA}_{G_1}(r) \cup \text{PA}_{G_2}(r)$ for every pattern $r$.  It satisfies the following properties:
\begin{itemize}
	\item $G$ is a pattern graph.
	\item $G$ has at least the same number of edges as $G_1$ and $G_2$.
\end{itemize}
\end{definition}

\begin{proposition}[Tree graphs generate pattern graphs]
\label{prop:closure}
The closure of the set of tree graphs under the merge operation is the set of pattern graphs.  Moreover, the set of all tree graphs is the smallest set that generates pattern graphs.
\end{proposition}

\begin{figure}[!t]
\centering
\includegraphics[width=0.95\textwidth]{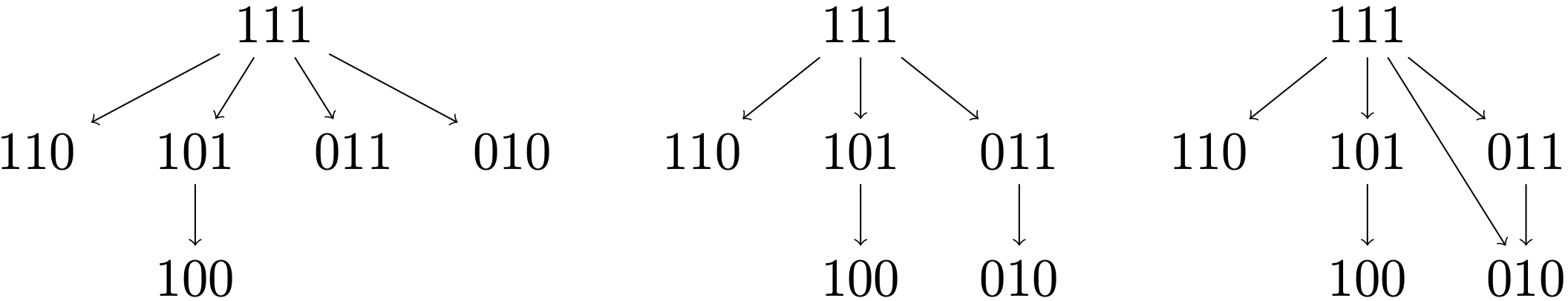}
\caption{Merge property.  Label the graphs from left to right as $G_1$, $G_2$, and $G_3$.}
\label{fig:merge}
\end{figure}

\newpage
\bibliographystyle{apacite}
\bibliography{references.bib}

\end{document}